\documentclass[12pt]{article}
\usepackage[utf8]{inputenc}
\usepackage{amsmath}
\usepackage{amsthm}
\usepackage{verbatim}
\usepackage{graphicx}
\usepackage{amssymb} 
\usepackage{commath}
\usepackage{authblk}
\usepackage{dsfont}
\usepackage{changepage}
\usepackage{tikz}
\usepackage{xr}
\usepackage{setspace} 
\usepackage{multirow}
\usepackage{natbib}
\usepackage{booktabs}
\usepackage{setspace,float}
\usepackage{footmisc}
\usepackage{bbm}
\usepackage{threeparttable}
\usepackage{charter}
\usepackage{mathtools}
\usepackage{xfrac}
\definecolor{Myblue}{HTML}{3a99d9}
\definecolor{Myred}{HTML}{C0392B}
\definecolor{Mybluegreen}{HTML}{347986}
\definecolor{Mygreen}{HTML}{35837b}
\usepackage{xcolor}
\usepackage[colorlinks = true,
            linkcolor = Myblue,
            urlcolor  = Myblue,
            citecolor = Myblue,
            anchorcolor = Myblue]{hyperref}
\usepackage{authblk}
\usepackage{mathrsfs}
\usepackage[english]{babel}
\usepackage[ruled,vlined]{algorithm2e}
\usepackage{tikz}
\usepackage[title]{appendix}
\usepackage[font=small,labelfont=bf]{caption}
\usepackage{enumerate}
\usepackage{dirtytalk}
\usetikzlibrary{decorations.pathreplacing, arrows,positioning}
\renewcommand{\arraystretch}{1.5}
\newtheorem{theorem}{Theorem}[section]
\newtheorem{assumption}{Assumption}

\newtheorem{corollary}{Corollary}[section]
\newtheorem{lemma}{Lemma}[section]
\newtheorem{proposition}{Proposition}[section]
\theoremstyle{definition}
\newtheorem{definition}{Definition}[section]
\newtheorem{remark}{Remark}[section]

\DeclareMathOperator{\E}{\mathbb{E}}

\DefineFNsymbols{mySymbols}{{\ensuremath\dagger}{\ensuremath\ddagger}\S\P
   *{**}{\ensuremath{\dagger\dagger}}{\ensuremath{\ddagger\ddagger}}}
\setfnsymbol{mySymbols}

\DeclareMathOperator*{\argmax}{arg\,max}

\DeclareMathOperator*{\argsup}{arg\,sup}
\DeclareMathOperator*{\arginf}{arg\,inf}
\makeatletter
\newcommand*\rel@kern[1]{\kern#1\dimexpr\macc@kerna}
\newcommand*\widebar[1]{%
  \begingroup
  \def\mathaccent##1##2{%
    \rel@kern{0.8}%
    \overline{\rel@kern{-0.8}\macc@nucleus\rel@kern{0.2}}%
    \rel@kern{-0.2}%
  }%
  \macc@depth\@ne
  \let\math@bgroup\@empty \let\math@egroup\macc@set@skewchar
  \mathsurround\z@ \frozen@everymath{\mathgroup\macc@group\relax}%
  \macc@set@skewchar\relax
  \let\mathaccentV\macc@nested@a
  \macc@nested@a\relax111{#1}%
  \endgroup
}
\makeatother
\usepackage[a4paper]{geometry}
\renewcommand\labelenumi{(\roman{enumi})}
\renewcommand\theenumi\labelenumi
\def\mathbi#1{\textbf{\em #1}}
\usepackage{setspace}
\usepackage{xr}
\makeatletter
\newcommand*{\addFileDependency}[1]{
\typeout{(#1)}
\@addtofilelist{#1}
\IfFileExists{#1}{}{\typeout{No file #1.}}
}\makeatother

\usepackage{autonum}

\begin{document}
\title{Robust Network Targeting with Multiple Nash Equilibria\thanks{I am especially grateful to Tim Christensen, Aureo de Paula, Toru Kitagawa, and Andrei Zeleneev for their continuous support, and Abhishek Ananth, Dmitry Arkhangelsky, Karun Adusumilli, Miguel Ballester, Andrew Chesher, Pedro Carneiro, Ben Deaner, James Duffy, Hugo Freeman, Andrea Galeotti, Duarte Gonçalves, Joao Granja, Raffaella Giacomini, Wayne Gao, Chen-Wei Hsiang, Hu Hao, Keisuke Hirano, Sukjin Han, Stephen Hansen, Guido Imbens, Hidehiko Ichimura, Hua Jin, Ziyu Jiang, Dennis Kristensen, Ivana Komujer, Frank Kleibergen, Maximilian Kasy, Nathan Canen, Soonwoo Kwon, Daniel Lewis, Jessie Li, Michael Leung, Yikai Li, Anna Mikusheva, Chuck Manski, Konrad Menzel, Robert Miller, Sophocles Mavroeidis, Lars Nesheim, Whitney Newey, Guillaume Pouliot, Martin Pesendorfer, Adam Rosen, Imran Rasul, Jonathan Roth, Jeff Rowley, Liyang Sun, J\"{o}rg Stoye, Shuyang Sheng, Xiaoxia Shi, Yuya Sasaki, Davide Viviano, Kaspar Wuthrich, Yike Wang, Mengshan Xu, Tian Xie, Weisheng Zhang, Yanziyi Zhang, and Yuanqi Zhang for beneficial comments and discussions.}}
\author{Guanyi Wang\thanks{CeMMAP and Department of Economics, University College London. Email$\colon$ \href{mailto:guanyi.wang.17@ucl.ac.uk}{guanyi.wang.17@ucl.ac.uk}}}
\maketitle
\begin{center}
\href{https://guanyiwang.github.io/assets/JMP.pdf}{Latest version here}
\end{center} 
\begin{abstract}
Many policy problems involve designing individualized treatment allocation rules to maximize the equilibrium social welfare of interacting agents. Focusing on large-scale simultaneous decision games with strategic complementarities, we develop a method to estimate an optimal treatment allocation rule that is robust to the presence of multiple equilibria. Our approach remains agnostic about changes in the equilibrium selection mechanism under counterfactual policies, and we provide a closed-form expression for the boundary of the identified set of equilibrium outcomes. To address the incompleteness that arises when an equilibrium selection mechanism is not specified, we use a maximin welfare criterion to select a policy based on the ``least favourable" equilibrium outcome, and implement this policy using a greedy algorithm.
We establish performance guarantees for our method by deriving a welfare regret bound, which accounts for sampling uncertainty and the use of a greedy algorithm. We demonstrate our method with an application to the microfinance dataset of \citet{banerjee2013diffusion}.
\end{abstract}


\pagebreak
\section{Introduction}

Many policy problems involve allocating treatment among a network of interacting agents. Examples include technology diffusion \citep{parente1994barriers, alvarez2023strategic}, teenage smoking \citep{nakajima2007measuring}, consumer adoption decisions \citep{banerjee2013diffusion,keane2013comparing}, and education and migration \citep{hsiao2022educational}. 
Research in these fields highlights the role of spillover effects, particularly those arising from strategic interactions. 

Among other things, these strategic interactions lead to the presence of multiple Nash equilibria, which complicates the process of finding an optimal treatment allocation policy. To handle this multiplicity, counterfactual policy analysis ``\textit{has made simplifying assumptions which either change the outcome space or impose ad hoc selection mechanisms in regions of multiplicity}'' \citep{tamer2003incomplete}. Consequently, this approach ``\textit{potentially introduces misspecifications and nonrobustness in the analysis of substantive questions}'' \citep{de2013econometric}. To address this problem, can we develop a treatment allocation rule that remains optimal even under the least favorable equilibrium?

Focusing on a class of network models where units participate in a simultaneous decision game with strategic complementarity \citep{brock2001discrete, ballester2006s, molinari2008identification, jia2008happens, echenique2009testing, lazzati2015treatment,graham2023scenario}, this paper develops a method for constructing a maximin optimal treatment allocation rule that is robust to the presence of multiple Nash equilibria.
A planner allocates a binary treatment among a target population of $N$ units embedded within a network, where each unit's covariates and the network structure are observable. Each unit then simultaneously chooses a binary action to maximize its own utility, which depends on its own characteristics and treatment, as well as the characteristics, treatments, and expected choice of its neighbors\footnote{This incomplete information setting \citep{brock2001discrete, bajari2010estimating, de2012inference}, is our primary focus. Section \ref{sec:complete} extends our results to the complete information setting.}. Our goal is to learn a treatment allocation policy that maximizes social welfare for the target population.

To determine the optimal treatment allocation rule for the target population, we assume that there exists data for a social network of units who have been assigned treatment in the past. This sample may differ from the target population in terms of both the number of units and the network structure. Data is assumed to be available for each unit's covariates, decisions, and assignments, as well as those of their neighbors. After assessing how individual outcomes vary in response to different treatment allocations among the training sample, we analyze the optimal treatment allocation strategy, taking into account the covariates and network structure of the target population. Consider, for example, targeted information provision in villages with the aim of increasing microfinance adoption, as discussed in \citet{banerjee2013diffusion}. By analyzing heterogeneous choices among the units in villages selected by policymakers, we then estimate whom to better target in external villages.

There are both theoretical and practical challenges to studying optimal treatment allocation in the presence of strategic interactions. The primary theoretical challenge is incompleteness \citep{jovanovic1989observable} of the model when there are multiple Nash equilibria\footnote{See the detailed surveys by \citet{de2013econometric, molinari2020microeconometrics, kline2020econometric}.}. Without assuming an equilibrium selection mechanism, our model predicts a set of equilibrium outcomes under a counterfactual policy. From a theoretical perspective, one cannot judge which equilibrium outcome is more likely than the others. This paper allows for these multiple equilibria and imposes no assumptions on equilibrium selection. Instead, we provide set-identified equilibrium social welfare for any treatment allocation policy, along with a closed-form expression that characterizes the bounds of this set.

As the counterfactual equilibrium social welfare is only set-identified, we cannot directly target equilibrium welfare when designing a treatment allocation rule. To address this uncertainty, we refine the optimality of treatment allocation using the maximin welfare criterion. This criterion is employed in the robust decision theory literature (e.g., \citealp{chamberlain2000econometric}), and the robust mechanism design literature (e.g., \citealp{morris2024implementation}). Under the maximin welfare criterion, our objective is to design a treatment allocation policy that maximizes social welfare evaluated under the least favourable equilibrium selection rule. 

In terms of implementation, there are two challenges. We adopt a parametric utility function specification, and the first challenge is estimating the parameters of this utility function. We assume the existence of a one-period training data set that contains a finite number $n$ of units, along with their covariates and the network structure\footnote{We allow the training data to come from our target population, with caveats about the private information of each unit. A more detailed discussion is provided in Section \ref{sec:estimation}.}. An existing treatment allocation policy is assumed, and we observe each unit's choice under this policy. We estimate parameters using the two-step maximum likelihood estimator proposed by \citet{leung2015two}. However, in the context of a network game setting, the asymptotic behavior of this estimator cannot be characterized without assuming how the network structure changes as the number of units increases (i.e., whether the network is dense or sparse). Although non-asymptotic results could elucidate how the sampling uncertainty of this estimator is influenced by network structure, the current literature lacks such analysis. Addressing this gap is one of the primary focuses of our paper.

The second challenge to implementation is finding the maximin optimal treatment allocation, which requires optimizing an of an objective function dependent on a system of simultaneous equations. In the presence of strategic interactions, when a treatment is assigned to a unit, it not only influences their behavior but also that of their neighbors. This, in turn, affects the payoff of their neighbors’ neighbors, propagating feedback effects throughout the network and presenting a complex combinatorial optimization problem.  To tackle this complexity, we propose a greedy algorithm. This algorithm sequentially assigns treatment to the agent who yields the highest marginal welfare gain at each step. However, this class of algorithm generally lacks a performance guarantee. We address this by characterizing the performance guarantee through the features of our objective function.

We evaluate the performance of our proposed method based on its regret, which is defined as the difference between the largest achievable welfare and the welfare achieved by our proposed method, evaluated under the least favourable equilibrium selection rule. Regret arises from two sources of uncertainty: The first is due to the use of estimated structural parameters, and reflects sampling uncertainty. The second is due to the use of a greedy algorithm.

This paper makes three theoretical contributions: (\romannumeral 1) It provides a closed-form expression for the identified set corresponding to the equilibrium outcomes under any arbitrary policy intervention. The heavy computation costs due to the large number of equilibria have limited the range of empirical applications in the literature to static models with a small number of players and choice alternatives. Our approach avoids computing the set of equilibria and hence allows for a feasible characterization of the identified region for the equilibrium social welfare; (\romannumeral 2) It presents the first non-asymptotic result on regret with strategic interactions. It shows that, under regularity conditions, the regret introduced by sampling uncertainty shrinks at the rate $\log(n)/\sqrt{n}$; (\romannumeral 3) It offers a theoretical performance guarantee for the regret associated with using a greedy algorithm to solve optimization problems involving systems of simultaneous equations, a topic previously unexplored in the existing literature.

To demonstrate how our method can be implemented and quantitatively evaluate its performance, we apply it to the data of \citet{banerjee2013diffusion}. We design a policy to maximize the take-up rate of microfinance products among households across various villages. For each village in the sample, we estimate the utility function parameters. These estimates are then used to assess the presence of strategic complementarity in each village. We find that strategic complementarities are present in 16 out of the 43 villages. For these villages, we construct an individualized treatment allocation rule using our greedy algorithm. Empirical results revealed that the occurrence of multiple equilibria can vary depending on the allocation method used. We compare the welfare outcomes achieved by our algorithm with those obtained by the NGO Bharatha Swamukti Samsthe (BSS). Our results indicate that, for all 16 villages exhibiting strategic complementarity, our method achieves notably higher welfare levels, with improvements ranging from $20\%$ to $270\%$, and an average improvement of $116\%.$  Additionally, the lower bound of welfare under our method consistently exceeds the maximal welfare attained under the allocation rule used by BSS and a rule that assigns treatment at random. These substantial welfare gains highlight the benefits of individualized targeting in the presence of strategic interference, which demonstrates the efficacy of our approach in optimizing resource allocation and improving social welfare.

\subsection{Literature Review}
This paper is related to several literatures in economics and econometrics, including strategic interactions, statistical treatment rules, robust decision theory, robust mechanism design, and greedy algorithms.

Pioneering contributions to the econometric aspects of game-theoretic models include works by \citet{jovanovic1989observable} and \citet{bresnahan1991empirical}, which explore the empirical challenges associated with models that feature multiple equilibria. The recent literature on the econometrics of strategic interactions includes simultaneous decision games with complete information, such as \citet{tamer2003incomplete}, \citet{bajari2010estimating,bajari2010identification}, \citet{de2018identifying}, \citet{sheng2020structural}, and \citet{chesher2020structural}; simultaneous decision games with incomplete information, such as \citet{de2012inference,de2020testable}, \citet{menzel2016inference}, and \citet{ridder2020two}; and sequential decision games: \citet{aguirregabiria2007sequential,aguirregabiria2019identification}, \citet{mele2017structural}, \citet{leung2019inference}, and \citet{christakis2020empirical}. 

Focusing on a game with complete information, \citet{tamer2003incomplete}  
obtains bounds for structural parameters while remaining fully agnostic about the equilibrium selection mechanism. Motivated by this, \citet{sheng2020structural} uses a sub-network approach to provide bounds for structural parameters in a network formation setting. \citet{chesher2020structural} partially identifies structural parameters using the Generalized Instrumental Variable approach \citep{chesher2017generalized}. \citet{bajari2010estimating} point identifies the parameters for a game with incomplete information, along with providing a semi-parametric estimator. However, estimation methods typically require observing repeated samples of the game, which may not be feasible for social network games. Motivated by this, \citet{leung2015two} studies a two-step maximum likelihood estimator in a large network setting, while \citet{ridder2020two} studies a two-step GMM estimator in a similar setting.
The paper adopts an existing estimator for structural parameters and treats this as an intermediate step in estimating an optimal policy.

One important task in obtaining an optimal policy is predicting the equilibrium outcome under counterfactual policies. In the strategic interaction literature, counterfactual analysis has been studied among others by \citet{jia2008happens}, \citet{aguirregabiria2010dynamic}, and \citet{canen2020decomposition} under various assumptions on the equilibrium selection mechanism. \citet{ciliberto2009market} does not restrict the equilibrium selection rule but considers only some candidate counterfactual policies. Additionally, \citet{lee2009multiple} investigates ATM network games by enumerating all Nash equilibria and analyzing how different learning algorithms select among them. None of these studies considers the aggregate equilibrium outcome, which aggregates the equilibrium outcomes of each unit, under a counterfactual policy. Here, we consider a social planner, and hence it is crucial to evaluate the aggregate social welfare. Remaining fully agnostic about the equilibrium selection mechanism, we provide a counterfactual analysis of the aggregate social welfare. 

Strategic interactions are closely related to social interaction models. These were introduced by \citet{manski1993identification}, which examines spillover effects through strategic interactions using a linear social interaction model with unique equilibrium. 
\citet{brock2001discrete} extends this model to a nonlinear setting and considers multiple equilibria. \citet{goldsmith2013social} considers the endogeneity of the network formation process. \citet{de2024identifying} recovers unknown network structure using a linear social interaction model.

This paper contributes to the growing literature on statistical treatment rules, which were introduced into econometrics by \citet{manski2004statistical} and \citet{dehejia2005program}. The recent literature includes \citet{stoye2009minimax,stoye2012minimax}, \citet{hirano2009asymptotics,hirano2020asymptotic}, \citet{chamberlain2011bayesian},   \citet{kitagawa2018should}, 
\citet{ananth2020optimal}, \citet{athey2021policy},
\citet{mbakop2021model},
\citet{kitagawa2021constrained},
\citet{sun2021empirical}, \citet{munro2021treatment}, \citet{christensen2022optimal},
\citet{adjaho2022externally}, \citet{kitagawa2022policy}, \citet{kitagawa2023individualized,KITAGAWA2023109}, \citet{viviano2024policy}, \citet{fernandez2024robust}, and \citet{munro2024treatment}. In contrast to the i.i.d. setting considered in most of these papers, we consider a setting where the spillover effects of treatment assignment are important. A small number of papers in the literature considers spillover effects. These include \citet{viviano2024policy}, \citet{ananth2020optimal}, \citet{munro2021treatment}, and \citet{kitagawa2023individualized,KITAGAWA2023109}. Apart from \citet{kitagawa2023individualized}, none of those papers consider spillover effects introduced by strategic interaction or the related complications, such as multiple equilibria.

\citet{viviano2024policy} and \citet{ananth2020optimal} focus on estimating direct and indirect treatment effects to derive optimal allocation policies based on data. We focus on the strategic interaction setting where each unit’s behavior is influenced by the behaviors of nearby individuals within a network. These interactions are naturally modeled using game theory \citep{jackson2015games}, which we adopt in our analysis.
In addition, using a game theoretical approach enables us to evaluate social welfare directly through the individual's utility and account for the general equilibrium effects. \citet{munro2021treatment} focuses on a competitive equilibrium where spillover effects are mediated through the equilibrium price. Their approach uses the mean-field limit to characterize the asymptotic behavior of treatment effects, focusing on settings where a unique mean-field equilibrium price exists. \citet{kitagawa2023individualized} focuses on a sequential decision game with a Markovian structure, which leads to a unique stationary joint distribution of units' decisions \citep{mele2017structural}. This stationary distribution allows each state to be revisited instead of converging to multiple distinct equilibria. As a consequence, their stationary welfare differs from the equilibrium welfare that we consider here and we do not require the existence of a potential function. \citet{KITAGAWA2023109} considers the spillover effects from vaccination. 

Although the source of uncertainty is different, our paper robustly addresses the incompleteness introduced by multiple equilibria in a manner inspired by the robust decision theory literature (see the recent survey by \citealp{chamberlain2020robust}). \citet{chamberlain2000econometric,chamberlain2000econometrics} consider decision-making when there is uncertainty due to a partially specified subjective distribution. Their robust decision rule maximizes the risk function evaluated at the least-favourable distribution.
\citet{hansen2001robust,hansen2008robustness} achieve robustness by working within a neighborhood of a reference model and maximizing the minimum of expected utility over that neighborhood. \citet{manski2003partial} faces a similar problem to us where some part of the model is missing from the data, and obtains a robust identification region by incorporating the maximum and minimum value of the unobserved component. \citet{giacomini2021} applies the robust Bayes approach of \citet{berger1994} to a set-identified model, and shows asymptotic equivalence between the identified set and the set of posterior means obtained from using a multiple priors. See also \citet{giacomini2021robust} and references therein. \citet{christensen2023counterfactual} relaxes parametric assumptions about the distribution of latent variables in a structural model. Their robust counterfactual set is obtained by maximizing (minimizing) the counterfactual through the distribution of latent variables over a neighborhood of the prespecified parametric distribution.

Finally, this paper is closely related to network games and mechanism design, as exemplified by \citet{morris2000contagion}, \citet{ballester2006s}, \citet{galeotti2010network}, and \citet{galeotti2020targeting} for network games, and \citet{mathevet2010supermodular}, \citet{gonccalves2020statistical}, \citet{fu2021full}, \citet{morris2024implementation}, and \citet{brooks2024structure} for mechanism design. Network games explore how network characteristics influence behavior. \citet{jackson2008social} and \citet{jackson2015games} provide comprehensive summaries. \citet{galeotti2020targeting} employs a principal component approach to analyze how interventions that change characteristics impact outcomes and develops strategies for optimal interventions within network games.  However, they assume a unique equilibrium, whereas this paper focuses on models with multiple equilibria. Following the same setting, \citet{sun2023structural} examines optimal interventions that alter network structure, while \citet{kor2022welfare} considers interventions that affect both characteristics and network structure. Nonetheless, these studies differ from ours in terms of utility specification, objective function, and the definition of the action space.

This paper can be viewed as a specific instance of mechanism design, where treatments are allocated to incentivize units’ equilibrium behavior towards achieving desired objectives. 
Closely related is \citet{morris2024implementation}, which characterizes the set of outcomes achievable from the smallest equilibrium, referred to as the smallest implementable outcome, in a supermodular game. Moreover, within a convex potential game, they show that the optimal outcome—realized by implementing information to maximize the smallest equilibrium—results in all players selecting the same action.
While our implementation approach differs, the lower bound of the set-identified social welfare in this paper is similar in concept to these smallest implementable outcomes. However, it is more complex to characterize this set as the number of players increases.
Additionally, \citet{morris2024implementation} leaves open the question of which implementation strategies are needed to achieve these outcomes, a gap this paper addresses.

\textit{Outline}.$\quad$ The rest of this paper proceeds as follows: Section \ref{sec:model} introduces the game setting and the solution concept. Section \ref{sec:counter} discusses counterfactual analysis. Section \ref{sec:treat} focuses on treatment allocation and implementation. Section \ref{sec:theory} presents theoretical results related to the implementation of our proposed method.  We apply our proposed method to the Indian micro-
finance data, which is studied by \citet{banerjee2013diffusion}, and demonstrate its performance in Section \ref{sec:empirical}. Section \ref{sec:complete} extends our analysis to the complete information setting. Section \ref{sec:conclude} concludes. All proofs and derivations are shown in
Appendix \ref{appendixA} to Appendix \ref{appsec:previous}.


\section{Model}\label{sec:model}
\subsection{Setup}
Let $\mathcal{N}=\{1,2,..., N\}$ be the target population. Each unit $i$ has a $K$-dimensional vector of characteristics $X_i$ observable to the researcher. $X_i$ is assumed to have bounded support, and we standardize the measurements of $X_i$ to be nonnegative, such that $X_i\in \mathcal{X}\in\mathbb{R}_{+}^{K}$. Let $X=[X_{1}^{\intercal},...,X_{N}^{\intercal}]\in\mathcal{X}^N$ be an $N \times K$ matrix whose $i$th row contains the characteristics of unit $i$, and let $\mathcal{X}^N$ represent the set of all such possible matrices $\mathcal{X}$. Let $D=\{D_1,...,D_N\}\in\mathcal{D}=\{0,1\}^N$ be a vector of binary treatment allocations. For $i \in \mathcal{N}$, $D_i=1$ if unit $i$ is treated and $D_i = 0$ if not.

The social network is represented by an $N \times N$ binary adjacency matrix, denoted by $G=\{G_{ij}\}_{i,j\in\mathcal{N}}\in\mathcal{G}=\{0,1\}^{N\times N}$. $G$ is assumed to be fixed and exogenous, irrelevant to treatment allocation. $G_{ij}=1$ indicates that units $i$ and $j$ are connected, while $G_{ij}=0$ indicates that they are not. Let $\mathcal{N}_i\coloneqq\{j:G_{ij}\neq 0\}$ denote the set of neighbors of unit $i$. $\widebar{N}$ denotes the maximum number of edges connected to any unit in the network (i.e., $\widebar{N}=\max_i\vert\mathcal{N}_i\vert$), while $\underline{N}$ denotes the minimum (i.e., $\underline{N}=\min_i\vert\mathcal{N}_i\vert$). We adopt the convention of no self-links (i.e., $G_{ii}=0$ for all $i \in \mathcal{N}$). This framework can accommodate both directed networks, where $G_{ij}$ and $G_{ji}$ can differ, and undirected networks, where $G_{ij}=G_{ji}$ for all $i,j\in\mathcal{N}$. Additionally, we allow the strength of spillover effects to depend not only on the adjacency matrix $G_{ij}$ but also on the covariates and treatment statuses of units $i$ and $j$. 

We consider a counterfactual equilibrium social welfare in the context of a large simultaneous decision game. 
We use the following notation for our simultaneous decision game. $Y_{i}\in\mathcal{Y}=\{0,1\}$ denotes unit $i$'s decision. The decision vector for all units is denoted by $Y= (Y_1,...,Y_N)\in\mathcal{Y}^N$, with $y\in\{0,1\}^N$ representing the realized decision outcomes. Additionally, we define a vector of idiosyncratic shocks $\varepsilon=\{\varepsilon_1,...,\varepsilon_N\}$, where $\varepsilon_i$ is the shock for unit $i\in\mathcal{N}$.
 
The game, denoted by $\Gamma$, comprises:

\vspace{0.3cm}
\noindent\textbf{Players}: \quad A set of individuals that we label $\mathcal{N}$, a social planner;

\vspace{0.3cm}
\noindent\textbf{Payoffs}: \quad The preferences (utilities) of units are denoted by $\{U_i(y, X, D, G; \theta)\}_{i=1}^N$. Following \citet{de2012inference} and \citet{galeotti2020targeting}, we endow units with a quadratic utility function
\begin{equation}
    U_i(y, X, D, G; \theta) = (\alpha_i-\varepsilon_i)y_i+\sum_{j\neq i}\beta_{ij}y_iy_j.
\end{equation}
where $\alpha_i\coloneqq\alpha_i(X,D,G)$ and $\beta_{ij}\coloneqq\beta_{ij}(X,D,G)$ are heterogeneous functions that capture unit $i$'s \textit{individual utility} and \textit{spillover utility}. The utility of $Y_i=0$ is normalised to $0$.

Given a network $G$, covariates $X=(X_1,...,X_N)$, and a treatment allocation $D=(D_1,...,D_N)$, the coefficient $\alpha_i$ on unit $i$'s choice depends upon their own covariates and treatment status as well as those of all of their neighbors; the coefficient $\beta_{ij}$ multiplying the quadratic term $y_iy_j$ depends upon their own covariates and treatment status as well as those of unit $j$. 
Since the choice variable is binary, if $\alpha_i$ and $\beta_{ij}$ are unconstrained, then this specification of the utility function is without loss of generality. We endow these utilities with certain properties, which are specified in Section \ref{sec:super} and Section \ref{sec:identi}.

\vspace{0.3cm}
\noindent\textbf{Information}:\quad The literature delineates two information environments: \textit{complete information} and \textit{incomplete information}. In a \textit{complete information} setting, players can observe all characteristics of other units. This setting is studied in \citet{tamer2003incomplete}, \citet{ciliberto2009market}, \citet{bajari2010identification} and \citet{chesher2020structural}. Since we consider a large network setting, it may not be plausible for players to have perfect information about all the other units \citep{ridder2020two}. Therefore, in our headline setting, we follow \citet{brock2001discrete}, \citet{aguirregabiria2007sequential}, \citet{bajari2010estimating}, and \citet{de2012inference}, and consider an \textit{incomplete information} setting. All units and the social planner are assumed to observe characteristics $X$ and the network structure $G$, but the vector of idiosyncratic shocks of units is assumed to be unobservable. The realization of $\varepsilon_i$ is unit $i$'s private information. All players are assumed to have a common belief about the distribution of $\varepsilon$. Formally, 
    \begin{assumption}\label{ass:epsilon}
        The set of idiosyncratic shocks $\varepsilon$ must satisfy the following conditions: 
        \begin{enumerate}
            \item The $\{\varepsilon_i\}_{i=1}^N$ is i.i.d. with a known distribution $F_{\varepsilon}$, which is common knowledge for all the players;
            \item The distribution of $\varepsilon_{i}$ has a density $f_{\varepsilon}$, which is bounded above by a constant $\tau$. In addition, $f_{\varepsilon}$ is continuously differentiable; 
            \item $\varepsilon_{i}\perp X,G,D$ for all $i\in\mathcal{N}$. 
        \end{enumerate}
    \end{assumption}
These assumptions are standard in the literature \citep{de2012inference,leung2015two,ridder2020two}. 
The third assumption can be replaced by a conditional independence assumption if we assume that $F_{\varepsilon\vert X,D,G}(\cdot\vert X,D,G)$ is known.

\vspace{0.3cm}
\noindent\textbf{Actions}: \quad 
    At the beginning of the game, the social planner assigns treatment $D_i$
  to each unit $i\in\mathcal{N}$ to maximize the \textit{planner's welfare}:
    \begin{equation}\label{eq:plannerwelfare}
        W_{X,G}(D) = \frac{1}{N}\sum_{i=1}^N\mathbb{E}_{\varepsilon}\left[g_i(Y,X,D,G)\vert X,G\right],
    \end{equation}
    subject to the capacity constraint $\kappa$ (i.e., $\sum_{i=1}^N D_i\leq \kappa$). The expectation in Eq.\ref{eq:plannerwelfare} is taken with respect to choices $Y$ given the observed covariates $X$, network structure $G$, and the treatment allocation rule $D$\footnote{With multiple equilibria and no assumption imposed on the equilibrium selection mechanism, the expectation becomes a set in which each element is conditional on a specific equilibrium selection mechanism. This concept will be further formalized in Section \ref{sec:equili}.}. The function $g_i: \mathcal{Y}^N\times \mathcal{X}^N\times \mathcal{D}\times\mathcal{G}\rightarrow \mathbb{R}$ allows social welfare to deviate from the utilitarian welfare function, which corresponds to $g_i(\cdot)=U_i(\cdot)$. We explore two common types of social welfare functions: Utilitarian welfare, and Engagement welfare. Section \ref{sec:counter} discusses each in detail.
    
    After receiving their allocated treatment, units choose action $Y$ simultaneously to maximize their own payoff given the realization of $\varepsilon$. With complete information unit $i$'s decision rule would be:
    \begin{equation}
        Y_i=\mathds{1}\Big\{U_i(1,Y_{-i},X,D,G)\geq 0\Big\},\quad \forall i\in\mathcal{N}.
    \end{equation}
   However, since unit $i$ only has partial information about other units, the realization of $Y_{-i}$ is not observed. Therefore, units make decisions that are best responses given their belief about other units' decisions given the public information and their own type. Formally, in the incomplete information setting,
   \begin{equation}\label{eq:Y}
       Y_i=\mathds{1}\Big\{\mathbb{E}_{\varepsilon}\big[U_i(1,Y_{-i},X,D,G)\vert X,D,G,\varepsilon_i\big]\geq 0\Big\},\quad \forall i\in\mathcal{N}.
   \end{equation}
   with 
   \begin{equation}
       \mathbb{E}_{\varepsilon}\big[U_i(1,Y_{-i},X,D,G)\vert X,D,G,\varepsilon_i\big] = \alpha_i+\sum_{j\neq i}\beta_{ij}\mathbb{E}_{\varepsilon}\big[Y_j\vert X,D,G,\varepsilon_i\big]-\varepsilon_i.
   \end{equation}
   As $\varepsilon$ is i.i.d. by Assumption \ref{ass:epsilon}, this can be simplified to:
   \begin{equation}\label{eq:expu}
       \mathbb{E}_{\varepsilon}\big[U_i(1,Y_{-i},X,D,G)\vert X,D,G,\varepsilon_i\big] = \alpha_i+\sum_{j\neq i}\beta_{ij}\mathbb{E}_{\varepsilon}\big[Y_j\vert X,D,G\big]-\varepsilon_i.
   \end{equation}
We have now established the game setting. To further elaborate, we introduce additional notation. Consider the action set $\mathcal{Y}$, defined as $\{0,1\}$. This set is a totally ordered set, endowed with the usual ordering relation $\leq$, characterized by reflexivity, antisymmetry, and transitivity\footnote{\textbf{Reflexive}: $\leq$ is reflexive if $y\leq y,$ for all $y\in\mathcal{Y}^N$. \textbf{Antisymmetric}: $\leq$ is antisymmetric if $y\leq y'$ and $y'\leq y$ implies $y=y'$. \textbf{Transitive}: if $y\leq y'$ and $y'\leq y''$ implies $y\leq y''$.}. The action profile space $\mathcal{Y}^N$, formed as a direct product of $\mathcal{Y}$, also constitutes a partially ordered set \citep[\S Example~2.2.1]{topkis1998supermodularity}. It is equipped with the \textit{product relation} $\leq$, where for any $y, y' \in \mathcal{Y}$, we have $y \leq y'$ if and only if $y_i \leq y_i'$ for all $i \in \mathcal{N}$. Given that $\mathcal{Y}^N$ is a partially ordered set, we can define a \textit{greatest} and \textit{least} element on it. A strategy profile $y$ is a greatest (least) element on $\mathcal{Y}^N$ if $y\geq y'$ ($y\leq y'$) for all $y'\in\mathcal{Y}^N$. In addition, The \textit{join} of any two elements $y, y'\in\mathcal{Y}^N$, written as $y\vee y'$, is defined as $\inf\{x\in\mathcal{Y}^N:x\geq y, x\geq y'\}$. The \textit{meet}, denoted as $y\wedge y'$, is symmetrically defined as: $\sup\{x\in\mathcal{Y}^N:x\leq y, x\leq y'\}$. In addition, a partially ordered set is called a \textit{lattice} if the join and meet of any pair of elements exist. A lattice is a \textit{complete lattice} if it contains the supremum and infimum of any subsets of it.

\subsection{Equilibrium}\label{sec:equili}

As the game introduced in the previous section features incomplete information, it is a Bayesian game \citep{harsanyi1967games}, and its Nash equilibria are \textbf{Bayesian Nash equilibria} (BNE). 
We use the pure strategy BNE solution concept. This is defined as:
\begin{definition}{(\textbf{Pure Strategy Bayesian Nash equilibrium})} Let $\mathbi{Y}$ be the set of all possible decision rules $\{y_i(\varepsilon_i)\}_{i=1}^N$, where $y_i(\varepsilon_i): \mathbb{R}\rightarrow\{0,1\}$ specifies unit $i$'s choice for each realization of their private information $\varepsilon_i$. A pure strategy BNE of game $\Gamma$ is a strategy profile $(y_1^*,...,y_N^*)$ such that, for every $i\in\mathcal{N}$,
\begin{equation}
\mathbb{E}_{\varepsilon}[U_i(y_i^*,y_{-i}^*)\vert X,D,G,\varepsilon_i]\geq \mathbb{E}_{\varepsilon}[U_i(y_i',y_{-i}^*)\vert X,D,G,\varepsilon_i],
\end{equation}
for all $y_i'\in\mathbi{Y}$, where $\mathbb{E}_{\varepsilon}[U_i(\cdot)\vert X,D,G,\varepsilon_i]$ is defined as in Eq.\ref{eq:expu}.  
\end{definition}
Following \citet{bajari2010estimating},  we represent the Bayesian Nash equilibrium in the conditional choice probability space. Denote the conditional choice probability (CCP) profile as $\sigma(X,D,G)= \{\sigma_i(X,D,G)\}_{i=1}^N$. An element of the CCP profile:
\begin{equation}\label{eq:ccp}
    \sigma_i(X,D,G) \coloneqq \mathbb{E}_{\varepsilon}[Y_i\vert X,D,G].
\end{equation}    
Combining the specification of $Y_i$ (Eq.\ref{eq:Y} and Eq.\ref{eq:expu}) with Eq.\ref{eq:ccp}, we have:
\begin{equation}\label{eq:fix}
    \sigma_i(X,D,G) = \int \mathds{1}\Big\{\alpha_i+\sum_{j\neq i}\beta_{ij}\sigma_j(X,D,G)\geq\varepsilon_i\Big\} dF_{\varepsilon}.
\end{equation}
Let $\Omega$ be a mapping from $[0,1]^N$ to $[0,1]^N$ that collects Eq.\ref{eq:fix} for all units. This is a non-linear simultaneous equation system. An equilibrium CCP profile $\sigma^*(X,D,G)$ is a fixed point of this simultaneous equation system: 
\begin{equation}\label{eq:omega}
    \sigma^*(X,D,G) = \Omega(\sigma^*(X,D,G)).
\end{equation}
This is one representation of the Bayesian Nash equilibrium. Alternatively, given an equilibrium CCP profile $\sigma^*$, a fixed $X,D,G$, and a realization of $\varepsilon$, we can define a Bayesian Nash equilibrium $\{y_{i}^*\}_{i=1}^N$ as: 
\begin{equation}\label{eq:equilisig}
    y_i^*=\mathds{1}\Big\{\alpha_i+\sum_{j\neq i}\beta_{ij}\sigma_j^*(X,D,G)\geq \varepsilon_i\Big\},\quad \forall i\in\mathcal{N}.
\end{equation}
As the right hand side of Eq.\ref{eq:fix} is equal to $F_{\varepsilon}(\alpha_i+\sum_{j\neq i}\beta_{ij}\sigma_{j})$, the existence of a fixed point is guaranteed by the \textit{Brouwer fixed-point theorem} \citep{brouwer1911abbildung}.  
As noted in \citet{echenique2009testing}, this type of simultaneous equation system can have multiple fixed points. In particular, games with strategic complementarity, as in our setting, tend to have a large number of equilibria \citep{takahashi2008number}.
Let $\Sigma\coloneqq\{\sigma:\sigma=\Omega(\sigma)\}$ denote the set of equilibria. Any equilibrium outcome in this set is a reasonable prediction. In other words, for a given $X, D, G$ and $\theta$, the model predicts a set of equilibrium outcomes $\sigma^*$. If we do not assume an equilibrium selection mechanism, this multiplicity introduces incompleteness \citep{jovanovic1989observable}. Incompleteness dramatically increases the difficulty of counterfactual analysis since the model can only identify a set of equilibrium CCP profiles $\Sigma$ with a newly implemented policy (i.e., a new treatment allocation rule $D$). 

With a newly implemented policy, the realized equilibrium depends on an equilibrium selection mechanism. Let $\xi:\Sigma\rightarrow[0,1]$ denote the probability distribution over equilibria, and let $\Delta(\Sigma)\coloneqq\{\xi:\sum_{\sigma^*\in\Sigma}\xi(\sigma^*)=1\}$ denote the set of all the probability distributions. The equilibrium selection mechanism is a mapping from the public information (i.e., $X, D, G$) to one particular element of $\Delta(\Sigma)$. Formally:
\begin{definition}{(\textbf{Equilibrium Selection Mechanism})}
    The equilibrium selection mechanism is denoted by $\lambda(\cdot\vert X,D,G)$ and the equilibrium selection mechanism space is defined as:
    \begin{equation}
   \Lambda \coloneqq \{\lambda:\mathcal{X}^N\times \mathcal{D}\times \mathcal{G}\rightarrow \Delta(\Sigma)\}.
\end{equation}
\end{definition}
If the equilibrium selection mechanism is observable, the conditional choice probability becomes complete by conditioning on $\lambda$. Since
\begin{equation}\label{eq:conlam}
    \Pr[Y_i=1\vert X,D,G,\lambda] = \sum_{\sigma^*\in\Sigma}\lambda(\sigma^*\vert X,D,G)\sigma_i^*,\quad\forall i\in\mathcal{N}.
\end{equation}
There are two main difficulties in characterizing the equilibrium outcome under a newly implemented policy. First, the equilibrium selection mechanism is not directly observable.  The identification of an equilibrium selection mechanism from data is studied in \citet{bajari2010identification} and \citet{aguirregabiria2019identification}, among others. This is useful in the identification of parameters, since parameter values are independent of $\lambda$. For counterfactual analysis, however, there is no guarantee that the equilibrium selection mechanism remains fixed when $X,D,G$ changes. The second difficulty is that the cardinality of $\Sigma$ increases dramatically with the number of units in the network. Hence, it is not feasible to evaluate the summation in Eq.\ref{eq:conlam}. To improve the tractability of counterfactual analysis, we focus on a game with strategic complementarity.

\subsection{Complementarity and Supermodular Games}\label{sec:super}

\textit{Strategic complementarity} in games implies that, given an ordering of strategies, a player's choice of a higher action incentivizes other players to similarly choose a higher action \citep{bulow1985multimarket}. In economics, complementarity is an important and empirically relevant concept \citep{molinari2008identification}. It has many policy applications, such as price setting \citep{alvarez2022price}, house prices \citep{guren2018house}, technology adoption \citep{alvarez2023strategic}, as well as the additional examples given in \citet{molinari2008identification}, \citet{lazzati2015treatment}, and \citet{graham2023scenario}. The theoretical literature has established that games with strategic complementarities have \textit{robust dynamic stability properties} \citep{milgrom1991adaptive,milgrom1994monotone}. This means they converge to the set of Nash equilibria even with simple learning dynamics (\citealp{fudenberg1998theory}; \citealp{chen2004does}). \citet{topkis1998supermodularity} shows that strategic complementarity and supermodularity are equivalent in finite strategy games. The mathematical property \textit{supermodularity} simplifies analysis. It captures the idea of increasing returns between the choice variables.
Therefore, to analyze the Bayesian Nash equilibrium of our game, we characterize it as a supermodular game. The definition of supermodular game is:
\begin{definition}\textbf{Supermodular Game \citep{milgrom1990rationalizability}}: 
A game $\Gamma$ is a supermodular game if, for each $i\in\mathcal{N}$:
\begin{enumerate}
    \item Strategy set $\mathcal{Y}$ is a complete lattice;
    \item Payoff $U_i:\mathcal{Y}^N\rightarrow\mathbb{R}$ is order upper semi-continuous in $y_i$ (for fixed $y_{-i}$) and order continuous in $y_{-i}$ (for fixed $y_{i}$), and has a finite upper bound;
    \item Payoff $U_i$ is supermodular in $y_{i}$ (for a fixed $y_{-i}$);
    \item Payoff $U_i$ has increasing differences in $y_{i}$ and $y_{-i}$.
\end{enumerate}
\end{definition}
The definitions of a \textit{supermodular function} and \textit{increasing differences} are:
\begin{definition}\textbf{Supermodular Function}:
    A function $U\colon \mathcal{Y}^N\rightarrow \mathbb{R}$ is \textit{supermodular} on $\mathcal{Y}^N$ if for all $y,y'\in\mathcal{Y}^N$:
    \begin{equation}
    U(y)+U(y')\leq U(y\wedge y')+U(y\vee  y').
    \end{equation}
\end{definition}

\begin{definition}\textbf{Increasing Differences}:
    A function $U\colon \mathcal{Y}^N\rightarrow \mathbb{R}$ has \textit{increasing differences} if for all $y_{-i} < y'_{-i}$ and $y_i< y_i'$:
    \begin{equation}
        U(y_i,y_{-i}')-U(y_i,y_{-i})\leq U(y_i',y_{-i}')-U(y_i',y_{-i}).
    \end{equation}
\end{definition}
\citet[\S Chapter~2.6.1]{topkis1998supermodularity} shows that, for a real valued utility function,  increasing differences is equivalent to complementarity between units' decisions. Given the definition of a supermodular game above, $U_i$ is a supermodular function on $\mathcal{Y}^N$ if and only if $U_i$ exhibits increasing differences on $\mathcal{Y}^N$ \citep[\S Theorem~2.6.1; \S Corollary~2.6.1]{topkis1998supermodularity}. Therefore, we have  equivalence between complementarity and supermodularity in our game. 
\textit{Topkis's characterization theorem} \citep[\S Section~3]{topkis1978minimizing} shows that
\begin{equation}
    \frac{\partial^2 U_i(y)}{\partial y_i\partial y_j}\geq 0, \quad \forall j\neq i
\end{equation}
is a necessary and sufficient condition to guarantee a utility function is a supermodular function on $\mathcal{Y}^N$. In our specification, this is equivalent to $\beta_{ij}\geq 0$ for $j\neq i$. 

Assuming that $\beta_{ij}\geq 0$ for $j\neq i$, our game is a supermodular game since $\{0,1\}^N$ is a complete lattice and our utility function is continuous. \textit{Tarski's fixed point theorem} \citep[\S Theorem~1]{tarski1955lattice} then guarantees the existence of pure strategy Bayesian Nash equilibrium $y^*$. In particular, there always exists a least BNE $\underline{y}^*$ and a greatest BNE $\widebar{y}^*$ \citep[\S Theorem~5]{milgrom1990rationalizability}. \textit{Tarski's fixed point theorem} can be applied to the conditional choice probability space instead of the strategy profile space to obtain an equivalent result. $[0,1]^N$ is also a complete lattice, and $\Omega:[0,1]^N\rightarrow[0,1]^N$ in Eq.\ref{eq:omega} is an increasing function given $\beta_{ij}\geq 0$. Therefore, we have a maximal equilibrium CCP profile $\widebar{\sigma}^*$ and a minimal equilibrium CCP profile $\underline{\sigma}^*$. In section \ref{sec:counter}, we show how strategic complementarity simplifies counterfactual analysis.

\section{Counterfactual Analysis for the Target Population}\label{sec:counter}
The goal of this paper is to obtain a treatment allocation that maximizes the equilibrium social welfare of the target population. To achieve this, we first need to characterize the counterfactual equilibrium social welfare if we implement a policy in the target population, which may have a different network structure to the training sample. As the equilibrium selection mechanism is unobservable, the literature typically obtains a point-identified prediction for welfare by assuming how counterfactual policies affect the equilibrium selection mechanism. For example, \citet{jia2008happens} assumes that a specific equilibrium is always played, \citet{aguirregabiria2010dynamic} assumes the equilibrium remains the same after intervention, and \citet{canen2020decomposition} assumes that the equilibrium selection mechanism is invariant to the intervention. However, it is impossible to test the appropriateness of these assumptions given the existing method. In contrast, following \citet{tamer2003incomplete}, we are fully agnostic about how policy changes the equilibrium selection mechanism. In other words, the question we focus on is:
\textit{If we are agnostic about the equilibrium selection mechanism, what counterfactual outcome does the model predict}? 

Recall that our social welfare function is:
    \begin{equation}\label{eq:sw}
        W_{X,G}(D) = \frac{1}{N}\sum_{i=1}^N\mathbb{E}_{\varepsilon}\left[g_i(Y,X,D,G)\vert X,G\right].
    \end{equation}
With multiple equilibria and no assumption imposed on the equilibrium selection mechanism, our model provides a set-valued equilibrium probability distribution fot $Y$ conditional on $X, D, G$. Therefore, the expectation in Eq.\ref{eq:sw} is also a set, with each element an expectation conditional on a particular $\lambda$. Formally,
\begin{equation}
    W_{X,G}(D) = \{W_{X,G,\lambda}(D):\lambda\in\Lambda\},
\end{equation}
where
\begin{equation}
    W_{X,G,\lambda}(D) = \frac{1}{N}\sum_{i=1}^N\mathbb{E}_{\varepsilon}\left[g_i(Y,X,D,G)\vert X,G,\lambda\right].
\end{equation}

This paper considers counterfactual analysis for two standard social welfare functions.
\begin{itemize}
    \item \textbf{Engagement Welfare}: In certain scenarios, a social planner may prioritize goals other than maximizing utilitarian welfare. For instance, in tax auditing, the planner might individualize the assignment of tax audits. Generally, units prefer not to pay taxes, so if maximizing utilitarian welfare were the sole objective, no one would be audited. In this case, a more appropriate target might be the average rate of tax compliance\footnote{This concept of welfare can be broadened to include situations where the policymaker aims to influence outcomes indirectly affected by individual decisions, such as total tax revenue, which depends on individuals' decisions to pay taxes.}. Engagement welfare is defined as 
    \begin{equation}
        W_{X,G,\lambda}(D)=\frac{1}{N}\sum_{i=1}^N \Pr(Y_i=1\vert X,D,G,\lambda).
    \end{equation}
    \item \textbf{Utilitarian Welfare at Equilibrium}: Utilitarian welfare at equilibrium is the average of the expected utilities of individuals when the system is in equilibrium. This measure is often targeted in policy interventions as it comprehensively reflects overall societal benefit (e.g., \citealp{brock2001discrete,galeotti2020targeting}). An example where the utilitarian welfare target is used is job training programs. Here policymakers allocate limited training resources to unemployed workers to assist them in finding new jobs \citep{bloom1997benefits}. In such scenarios, social welfare is defined as: 
    \begin{equation}\label{eq:uwe}
        \begin{split}
            W_{X,G,\lambda}(D)&=\frac{1}{N}\sum_{i=1}^N \mathbb{E}_{\varepsilon}\left[U_i(Y,X,D,G)+\varepsilon_iY_i\vert X,D,G,\lambda\right],
        \end{split}
    \end{equation}
    which only depends on the expectation of the deterministic component in the utility function\footnote{\citet[Section 4]{brock2001discrete} show that introducing a shock term in Eq.\ref{eq:uwe} would render the model analytically intractable.}. Plugging in our utility function specification, we have:
    \begin{equation}
        W_{X,G,\lambda}(D)=\frac{1}{N}\sum_{i=1}^N\alpha_i\Pr(Y_i=1\vert X,G,\lambda)+\frac{1}{N}\sum_{i=1}^N\sum_{j\neq i}\beta_{ij}\Pr(Y_iY_j=1\vert X,D,G,\lambda).
    \end{equation}
\end{itemize}
Consider the engagement welfare function. We define bounds for equilibrium welfare, given covariates $X$, network $G$ and an arbitrary treatment allocation rule $D$, as:
\begin{equation}
\begin{aligned}
    W_{X,G,\lambda}(D)\in \Big[\inf_{\lambda\in\Lambda}W_{X,G,\lambda}(D),\:
    \sup_{\lambda\in\Lambda}W_{X,G,\lambda}(D)\Big].
    \end{aligned}
\end{equation}
Accordingly, let $\underline{\lambda}$ be the least-favorable equilibrium selection mechanism and $\widebar{\lambda}$ the most-favorable equilibrium selection mechanism :
\begin{equation}
    \underline{\lambda}\coloneqq \arginf_{\lambda\in\Lambda} W_{X,G,\lambda}(D), \quad \widebar{\lambda}\coloneqq \argsup_{\lambda\in\Lambda} W_{X,G,\lambda}(D). 
\end{equation}

In general it is not possible to solve for these two extreme points. There are two obstacles. First, the number of equilibria increases rapidly with the number of units in the network. Evaluating the expectation with respect to the joint distribution of $Y$ thus becomes infeasible. Second, the space of the equilibrium selection mechanisms $\Lambda$ may be infinite. This complicates any search for the infimum and supremum $\lambda$ across $\Lambda$.

In the existing literature, counterfactual analysis  \citep{ciliberto2009market} often focuses instead on the conditional choice probability (CCP). With no assumptions on the equilibrium selection mechanism, the bounds of the counterfactual CCP are:
\begin{equation}\label{eq:ccpcounter}
    \Pr(Y_i=1\vert X,D,G,\lambda)\in\Big[\inf_{\lambda\in\Lambda}\Pr(Y_i=1\vert X,D,G,\lambda),\:\sup_{\lambda\in\Lambda}\Pr(Y_i=1\vert X,D,G,\lambda)\Big].
\end{equation}
These bounds can be computed using off-the-shelf methods (e.g., \citet{sheng2020structural} for complete information settings). However, exact bounds of social welfare cannot be directly obtained from the bounds of the CCP. This is because:
 \begin{equation}\label{eq:minneq}
    \inf_{\lambda\in\Lambda}\sum_{i=1}^N \Pr(Y_i=1\vert X,D,G,\lambda)\leq \sum_{i=1}^N \inf_{\lambda\in\Lambda} \Pr(Y_i=1\vert X,D,G,\lambda),
\end{equation}
and
\begin{equation}\label{eq:maxneq}
    \sup_{\lambda\in\Lambda}\sum_{i=1}^N \Pr(Y_i=1\vert X,D,G,\lambda)\geq \sum_{i=1}^N \sup_{\lambda\in\Lambda} \Pr(Y_i=1\vert X,D,G,\lambda).
\end{equation}
That is, the lower (and upper) bound of unit $i$'s conditional choice probability may be obtained under a different equilibrium selection mechanism to the bound for some unit $j \neq i$. Therefore, bounds for social welfare obtained by summing the bounds on the CCP will generally be loose. However, we show that Eq.\ref{eq:minneq} and Eq.\ref{eq:maxneq} hold with equality in a supermodular game. Formally,
\begin{theorem}\label{thm:equality}(\textbf{Engagement Welfare})
    For a supermodular game, the least favorable equilibrium selection rule $\underline{\lambda}$ and the most favorable equilibrium selection rule $\widebar{\lambda}$ are:
    \begin{equation}
        \underline{\lambda} \coloneqq \delta_{\underline{\sigma}^*}, \quad \widebar{\lambda} \coloneqq \delta_{\widebar{\sigma}^*},
    \end{equation}
    where $\delta_{\sigma}$ is the Dirac measure on the set of equilibria $\Sigma$. 
In addition, the following conditions are satisfied:
\begin{equation}
    \begin{split}
\inf_{\lambda\in\Lambda}\sum_{i=1}^N \Pr(Y_i=1\vert X,D,G,\lambda)
      =\sum_{i=1}^N\inf_{\lambda\in\Lambda} \Pr(Y_i=1\vert X,D,G,\lambda),
    \end{split}
\end{equation}
\begin{equation}
    \begin{split}
\sup_{\lambda\in\Lambda}\sum_{i=1}^N \Pr(Y_i=1\vert X,D,G,\lambda)
      =\sum_{i=1}^N\sup_{\lambda\in\Lambda} \Pr(Y_i=1\vert X,D,G,\lambda).
    \end{split}
\end{equation}
\end{theorem}
A proof of Theorem \ref{thm:equality} is provided in Appendix \ref{app:thequ}. This new result characterizes the most and least favorable equilibrium selection rules for aggregate social welfare. This approach enables us to leverage Tarski's fixed point theorem, which significantly reduces the computational burden by obviating the need to calculate all possible Nash equilibria. Furthermore, it establishes equivalence between the identified set of aggregate social welfare and the aggregation of identified sets of conditional choice probabilities, concepts studied in \citet{sheng2020structural} and \citet{gu2022counterfactual}. This equivalence is not guaranteed to hold in the absence of complementarity. When there are values of 
$\varepsilon$ with unordered multiple equilibria, such as $(Y_1=1, Y_2=0)$ and $(Y_1=0, Y_2=1)$ in the two-unit case, the process of identifying the least and most favorable $\lambda$ is significantly more complicated. Intuitively, the bounds coincide because strategic complementarity guarantees the existence of a least BNE and a greatest BNE for all the values of $\varepsilon$. Since the social welfare function is a monotonically increasing function of $\sigma$, it achieves its lower bound at the least equilibrium $\underline{\sigma}^*$ and its upper bound at the greatest equilibria $\widebar{\sigma}^*$. By definition, the conditional choice probability $\Pr(Y_i=1\vert X,D,G,\lambda)$ also achieves its lower bound under $\underline{\sigma}^*$ and its upper bound under $\widebar{\sigma}^*$ . The same argument can be applied to utilitarian social welfare to obtain the following corollary.
 
\begin{corollary} (\textbf{Utilitarian Welfare at Equilibrium})\label{coro:utility}
Under Assumption \ref{ass:epsilon}, given the specification of our utility function,
the predicted set of the expected utilitarian welfare under a counterfactual policy $D$ is given as:
\begin{equation}
    \begin{split}
        W_{X,G,\lambda}(D)\in&\Big[\frac{1}{N}\sum_{i=1}^N\alpha_i f(\alpha_i)+\frac{1}{N}\sum_{i=1}^N\sum_{j\neq i}\beta_{ij} \underline{\sigma}^*_i\underline{\sigma}^*_j, \frac{1}{N}\sum_{i=1}^N\alpha_i f(\alpha_i)+\frac{1}{N}\sum_{i=1}^N\sum_{j\neq i}\beta_{ij} \widebar{\sigma}^{*}_i\widebar{\sigma}^{*}_j\Big],
    \end{split}
\end{equation}
where 
\begin{equation}
    f(\alpha_i) = 
\begin{cases} 
\Pr(Y_i=1\vert X,D,G,\underline{\lambda}) & \text{if } \alpha_i > 0 \\
\Pr(Y_i=1\vert X,D,G,\widebar{\lambda}) & \text{if } \alpha_i \leq 0 .
\end{cases}
\end{equation} 
\end{corollary}
A proof of Corollary \ref{coro:utility} is provided in Appendix \ref{app:oroofcoroutility}. This result implies that it is sufficient to compute the minimal and maximal equilibrium CCP profile for the utilitarian welfare in the incomplete information setting. This result does not hold in the complete information setting, where we provide an alternative approach to compute the bounds of the identified set.

\begin{remark}
    The infimum and supremum of the planner's welfare, calculated over the equilibrium selection mechanism, are equivalent to the \textit{Choquet integral} (see \citealp{denneberg1994non} and \citealp{gilboa2009theory}) of the planner's welfare with respect to the capacity and its conjugate. The capacity $v$ and its conjugate $v^*$ are non-additive probability measures defined on the set of equilibrium $\Sigma$. In our case, $\Pr(Y_i=1\vert X,D,G,\underline{\lambda})$ is equivalent to Choquet integration with respect to the capacity $v(A)$ where $A=\{y^*:y^*_i=1\}$. Analogously, $\Pr(Y_i=1\vert X,D,G,\widebar{\lambda})$ is equivalent to Choquet integration with respect to the conjugate $v^*(A)$ where $A=\{y^*:y_i^*=1\}$. We refer to \citet{kaido2023applications} for the definition of capacity and a more detailed discussion on this topic. Despite the complexity typically associated with Choquet integration, which often requires approximate solution by simulation methods, our method provides a closed-form expression for the identified set which can be solved without numerical error. Moreover, the applications of Choquet integration extend to robust Bayesian analysis to manage multiple priors (see \citet{chamberlain2000econometric} and \citet{giacomini2021robust}), which is analogous to a setting with multiple equilibria.
\end{remark}

The details of the computation of the maximal and minimal equilibrium conditional choice probabilities (CCPs) are discussed in Section \ref{sec:comput}. The arguments above apply not only to the counterfactual analysis of treatment allocation policies but also to policy interventions that alter covariates or the network structure.


\section{Treatment Allocation}\label{sec:treat}
Our model allows for multiple equilibria, but can only predict a set of possible equilibrium outcomes, denoted as $W_{X,G}(D)$. Consequently, the expected value calculation that determines social welfare is not well-defined without specifying the equilibrium selection mechanism. Drawing on game theory \citep{morris2024implementation} and robust decision theory \citep{chamberlain2000econometric}, we apply the maximin welfare criterion to select a treatment allocation rule. This approach involves a social planner opting for choices that lead to higher welfare while preparing for the worst-case scenario of the least favorable equilibrium. Essentially, the planner anticipates the minimal equilibrium will be realized. For example, \citet[Section 4.1.3]{segal2003coordination} discusses scenarios in contracting where the worst-case equilibrium corresponds to the Pareto-efficient outcome for the parties involved. Moreover, in settings where action 0 is the default, games exhibiting strategic complementarity tend to converge toward their minimal equilibrium. 

The planner chooses $D$ to maximise welfare under the assumption that the minimal equilibrium, conditional on the chosen $D$, will be realized. We denote the set of feasible allocations by $\mathcal{D}_\kappa\coloneqq\{D\in\mathcal{D}:\sum_{i=1}^N D_i\leq \kappa\}$. Formally:
\begin{equation}\label{eq:target}
    D^*=\arg\max_{D\in\mathcal{D}_{\kappa}}\min_{\lambda\in\Lambda} W_{X,G,\lambda}(D).
\end{equation}
Recall that, by Theorem \ref{thm:equality}, the lower bound of equilibrium social welfare equals the summation of individual welfares. Thus, the maximin welfare optimisation problem simplifies to:
\begin{equation}
    \begin{split}
      \max_{D\in\mathcal{D}_{\kappa}} W_{X,G,\underline{\lambda}}(D),
    \end{split}
\end{equation}
where 
 $W_{X,G,\underline{\lambda}}(D)$ is the social welfare function evaluated at the minimal equilibrium. This formulation converts the maximin welfare problem into a straightforward maximization problem, providing a clear framework for solving the optimal treatment allocation problem.

\subsection{Implementation}
\subsubsection{Identification}\label{sec:identi}
The preceding discussion has assumed that true parameter values are observed. To implement our proposed method, we first describe the identification of structural parameters using a training sample
. Details on the estimation procedure are provided in Section \ref{sec:estimation}. 

The discussion of counterfactual analysis in Section \ref{sec:counter} makes no assumptions about the functional form of parameters $\{\alpha_i\}_{i\in\mathcal{N}}$ and $\{\beta_{ij}\}_{i,j\in\mathcal{N}}$. However, observable data is limited to units' choices, covariates $X$, the network structure $G$, and a predetermined treatment allocation $D$. In practice, we are restricted by what it is possible to identify given this data. 

For identification, we follow \citet{bajari2010estimating} and adopt the inverse-CDF procedure\footnote{This approach builds on \citet{hotz1993conditional} and \citet{aguirregabiria2007sequential}.}. Let $\varepsilon^n$ be the private information in the training data, which is distinct from $\varepsilon$ in the target population. Recall from Eq.\ref{eq:equilisig} that, given an equilibrium conditional choice probability profile in the training data $\sigma^{data}$, unit $i$ chooses their actions according to the decision rule
\begin{equation}
    Y_i^{n}=\mathds{1}\big\{\alpha_i+\sum_{j\neq i}\beta_{ij}\sigma_j^{data}\geq \varepsilon_i^n\big\},\quad \forall i\in\mathcal{N}.
\end{equation}
The equilibrium CCP profile is thus: 
\begin{equation}\label{eq:iden}
    \begin{split}
        \sigma_i^{data} 
        &= 
        \int \mathds{1}\big\{\varepsilon_i^n\leq \alpha_i+\sum_{j\neq i}\beta_{ij}\sigma_j^{data}   \big\}dF_{\varepsilon^n}=F_{\varepsilon^n}\big[\alpha_i+\sum_{j\neq i}\beta_{ij}\sigma_j^{data}\big].
    \end{split}
\end{equation}
Taking the inverse of the CDF of $\varepsilon$ on both sides in Eq.\ref{eq:iden} yields:
\begin{equation}\label{eq:ideq}
    F_{\varepsilon^n}^{-1}(\sigma_i^{data}) = \alpha_i+\sum_{j\neq i}\beta_{ij}\sigma_j^{data}.
\end{equation}
Even assuming that the equilibrium CCP profile in the training data is observable, identifying all the parameters in Eq.\ref{eq:ideq} remains challenging. Determining all utility parameters involves solving for $N\times N$ unknown parameters on the right-hand side of the above equation. However, the left-hand side of Eq.\ref{eq:ideq}, only provides information about $N$ scalars. Given these limitations, we define our utility function as follows to ensure identifiability and allow for the analysis of general treatment effects:
\begin{equation}\label{eq:specifi}
    \begin{split}
        U_i(y,X,D,G)&=y_i(\overbrace{\theta_0+\theta_1 D_i+ X_i^{\intercal}\theta_2+X_i^{\intercal}\theta_3D_i+\frac{1}{\vert \mathcal{N}_i\vert}\sum_{j\neq i}\theta_4m_{ij}G_{ij}D_j}^{\alpha_i}-\varepsilon_i)\\
        &+\sum_{j\neq i}\underbrace{\frac{1}{\vert \mathcal{N}_i\vert}(\theta_5 +\theta_6 D_iD_j)m_{ij}G_{ij}}_{\beta_{ij}}y_iy_j,
    \end{split}
\end{equation}
where $m_{ij}=m(X_i,X_j)$ is a (bounded) real-valued function of personal characteristics. $m_{ij}$ measures the distance between unit $i$'s characteristics and unit $j$'s characteristics;
the spillover effect is weighted by how similar two units appear. The utility that unit $i$ derives from an action is the sum of the net benefits that they accrue from their own actions and from those of their neighbors. 
We assume that a unit's utility is only affected by the actions of their direct neighbors, not one-link-away contacts. 
The payoff of action $Y_i=1$ has six components. 
When unit $i$ chooses action $Y_i=1$, they receive utility $\theta_0$ irrespective of their allocated treatment. 
They also receive additional utility $\theta_1D_i$ depending upon their own treatment status. 
Their utility also includes a heterogeneous component $X_i^{\intercal}(\theta_2+\theta_3D_i)$, which depends upon their characteristics $X_i$. 
Next, there is a spillover effect from the action of unit $j$.
If unit $j$ is a neighbor of unit $i$ that receives treatment, then this provides $\theta_4m_{ij}$ additional utility to unit $i$. The fifth and sixth components represent strategic complementarity. If unit $j$ is a neighbor of unit $i$ and selects $Y_j = 1$, then unit $i$'s payoff is increased by $\theta_5 m_{ij}$
The final component corresponds to choice spillovers between neighbors who receive treatment. 
If both unit $i$ and unit $j$ receive treatment and both choose action $1$, unit $i$ receives additional utility $\theta_6m_{ij}$. 

Accordingly, the structural parameters $\theta$ are uniquely determined by the conditional choice probabilities in the training sample, thus identifying the payoff function. This discussion provides only an informal overview of the identification process; a formal proof is available in \citet{bajari2010estimating}.

\subsubsection{Estimation}\label{sec:estimation}
For estimation, we employ the two-step maximum likelihood estimation procedure of \citet{leung2015two}. The first step involves estimating the equilibrium conditional choice probability from the training data. In the second step, structural parameters are estimated by maximizing the likelihood function given the estimated CCP profiles. To distinguish the training data from the target population, we denote  covariates as $\mathsf{X}=\{\mathsf{X}_i\}_{i=1}^{n}$, the treatment allocation as $\mathsf{D}=\{\mathsf{D}_i\}_{i=1}^{n}$, decisions as $\mathsf{Y}=\{Y_i\}_{i=1}^{n}$, and the network structure as $\mathsf{G}=\{\mathsf{G}_{ij}\}_{i,j=1}^{n}$. In addition, let $S = (\mathsf{X},\mathsf{D},\mathsf{G})$ .

Let $\{\hat{\sigma}_i^{data}\}_{i=1}^N$ be the CCP in the training data. Given that the training dataset contains only a single large network, two necessary conditions on the training data are required to estimate the conditional choice probability: \textit{symmetric equilibrium}\footnote{A symmetric equilibrium implies that two units will exhibit identical conditional choice probabilities if they receive the same treatment, share identical covariates, and have comparable neighbors, specifically in terms of the neighbors' treatments and covariates.} \citep{leung2015two} and \textit{network decaying dependence condition} \citep{xu2018social}. In general, each unit's choice depends on all public information across the network $G$ (i.e., $X$ and $D$ of all units), although direct payoffs may depend only on immediate spillovers. Under the \textit{network decaying dependence condition}, it is sufficient to consider only interactions within a relatively small distance.

Several estimation approaches have been proposed for CCP. These including the empirical frequency estimator \citep{hotz1993conditional}, sieve estimation \citep{bajari2010estimating}, flexible logit estimation \citep{arcidiacono2011conditional}, and logit Lasso estimation \citep{chernozhukov2022locally}. Here we leave aside the question of the most suitable procedure. Instead, we assume the existence of an estimator that satisfies the following statistical property:
\begin{assumption}{(\textbf{Sub-Gaussian CCP Estimator})}\label{ass:ccp}
There exists a positive constant $C_{\sigma}$ such that for every $t\geq 0$, we have
\begin{equation}
    \Pr\big(\vert \hat{\sigma}_i^{data}-\sigma_i^{data}\vert\geq t\big\vert S,\sigma^{data}\big)\leq 2\exp\big(-nt^2/C_{\sigma}^2\big),\quad \forall i=1,...,n.
\end{equation}
\end{assumption}
This assumption is satisfied by the empirical frequency estimator \citep{leung2015two,ridder2020two} and logit/probit estimation\footnote{By Hoeffding's inequality, the frequency estimator easily satisfies the Assumption \ref{ass:ccp}. The probit/logit estimator satisfies the Assumption \ref{ass:ccp} by the same argument as the second-stage MLE estimator in Section \ref{sec:theory}.}. The ultimate goal is to choose $\hat{\theta}$ that maximizes the likelihood function. As $\sigma^{data}$ is unobserved, we replace $\sigma^{data}$ in the likelihood function with $\hat{\sigma}^{data}$ and estimate $\boldsymbol{\theta}$ by maximizing the quasi-likelihood function $\hat{Q}_n(\hat{\sigma}^{data},\boldsymbol{\theta})$.
\begin{equation}
    \hat{Q}_n(\hat{\sigma}^{data},\boldsymbol{\theta})=\frac{1}{n}\sum_{i=1}^n \mathsf{Y}_i\log\big(F_{\varepsilon}(\hat{Z}_{i}^\intercal\boldsymbol{\theta})\big)+(1-\mathsf{Y}_i)\log\big(1-F_{\varepsilon}(\hat{Z}_{i}^\intercal\boldsymbol{\theta})\big),
\end{equation}
where
\begin{equation}\label{eq:Zdefine}
    \hat{Z}_i = \Big(1,\mathsf{D}_i, \mathsf{X}_i^\intercal, \mathsf{X}_i^\intercal \mathsf{D}_i, \frac{1}{\vert \mathcal{N}_i\vert}\sum_{j\neq i}m_{ij}\mathsf{G}_{ij}\mathsf{D}_j, \frac{1}{\vert \mathcal{N}_i\vert}\sum_{j\neq i}m_{ij}\mathsf{G}_{ij}\hat{\sigma}^{data}_j, \frac{1}{\vert \mathcal{N}_i\vert}\sum_{j\neq i}m_{ij}\mathsf{G}_{ij}\hat{\sigma}_j\mathsf{D}_i\mathsf{D}_j\Big)^\intercal.
\end{equation}
In addition, $Z_i$ denotes the vector of regressors that would be obtained if we replaced $\hat{\sigma}^{data}_i$ in $\hat{Z}_i$ with the true conditional choice probability $\sigma^{data}_i$.

\subsubsection{Computation of Equilibria}\label{sec:comput}
After obtaining estimated parameters, we compute the set of equilibrium social welfare for given covariates $X$, network structure $G$, and a treatment allocation $D$ in the target population. 
The lower bound and upper bound of this set are:
$\Pr(Y_i=1\vert X,D,G,\underline{\lambda};\hat{\theta})$ and $\Pr(Y_i=1\vert X,D,G,\overline{\lambda};\hat{\theta})$ for all unit $i\in\mathcal{N}$. We first rewrite these two conditional probabilities as:
\begin{equation}\label{eq:prup}
   \begin{split}
        \Pr(Y_i=1\vert X,D,G,\widebar{\lambda};\hat{\theta}) = \int \mathds{1}\Big\{\hat{\alpha}_i+\sum_{j\neq i}\hat{\beta}_{ij}\Pr(Y_{j}=1\vert X,D,G,\widebar{\lambda};\hat{\theta})\geq \varepsilon_i\Big\} dF_{\varepsilon},
   \end{split}
\end{equation}
\begin{equation}\label{eq:prlow}
   \begin{split}
        \Pr(Y_i=1\vert X,D,G,\underline{\lambda};\hat{\theta}) = \int \mathds{1}\Big\{\hat{\alpha}_i+\sum_{j\neq i}\hat{\beta}_{ij}\Pr(Y_{j}=1\vert X,D,G,\underline{\lambda};\hat{\theta})\geq \varepsilon_i\Big\} dF_{\varepsilon}.
   \end{split}
\end{equation}
From Theorem \ref{thm:equality}, $\Pr(Y_i=1\vert X,D,G,\lambda)$
achieves its upper (lower) bound when the equilibrium is $\widebar{\sigma}^*$ ($\underline{\sigma}^*$) for all $i\in\mathcal{N}$. Therefore, 
\begin{equation}
    \Pr(Y_i=1\vert X,D,G,\widebar{\lambda};\hat{\theta})=\hat{\bar{\sigma}}_i^*,\quad\forall i\in\mathcal{N},
\end{equation}
\begin{equation}
    \Pr(Y_i=1\vert X,D,G,\underline{\lambda};\hat{\theta})=\hat{\underline{\sigma}}_i^*,\quad\forall i\in\mathcal{N},
\end{equation}
where $\hat{\bar{\sigma}}_i^*$ and $\hat{\underline{\sigma}}_i^*$ represent the estimators for the maximal and minimal equilibria, respectively. Hence, we need only compute the least and greatest equilibrium CCP profile. \citet{topkis1979equilibrium} provides an easily implemented algorithm that is guaranteed to converge to the least and greatest equilibrium point of a supermodular game. Hold $X,D,G$ fixed.  To obtain the greatest fixed point $\widebar{\sigma}^*$, begin with $\widebar{\sigma}^0=\{1,...,1\}$. Define a sequence $\{\widebar{\sigma}^t\}_{t=0}^T: \widebar{\sigma}^{t+1}= \Omega(\widebar{\sigma}^{t})$. By construction, $\widebar{\sigma}^0\geq \Omega(\widebar{\sigma}^0)$. Since $\Omega(\cdot)$ is an increasing function, $\Omega(\widebar{\sigma}^0)\geq \Omega(\widebar{\sigma}^1)$. Therefore, $\widebar{\sigma}^0\geq \widebar{\sigma}^1\geq ...\geq \widebar{\sigma}^T$. Suppose the iteration convergences on the $M$-th step. Then $\widebar{\sigma}^M$ is the greatest equilibrium since, for all the other $\sigma^*$, $\widebar{\sigma}^M=\Omega^M(\sigma^0)\geq \Omega^M(\sigma^*)=\sigma^*$.  

With a symmetric argument, we can obtain the least equilibrium CCP profile. Here we begin with $\underline{\sigma}^0=\{0,...,0\}$. Define a sequence $\{\underline{\sigma}^t\}_{t=0}^T: \underline{\sigma}^{t+1}= \Omega(\underline{\sigma}^{t})$. By construction, we have $\underline{\sigma}^0\leq \Omega(\underline{\sigma}^0)$. Again,  $\Omega(\underline{\sigma}^0)\leq \Omega(\underline{\sigma}^1)$. Therefore, $\underline{\sigma}^0\leq \underline{\sigma}^1\leq ...\leq \underline{\sigma}^T$. Suppose the iteration convergences on the $M$-th step. Then $\underline{\sigma}^M$ is the least equilibrium since, for all the other $\sigma^*$, $\underline{\sigma}^M=\Omega^M(\sigma^0)\leq \Omega^M(\sigma^*)=\sigma^*$.

\subsubsection{Greedy Algorithm}
The previous sections describe the estimation of parameters and computation of the least equilibrium CCP profile. In this section, we propose an algorithm to allocate treatment in a manner that maximizes the worst-case social welfare given the estimated parameters. Define the empirical welfare function to be the welfare function with estimated structural parameters:
\begin{equation}
    W_{X,G,\lambda}^n(D) = W_{X,G,\lambda}(D;\hat{\boldsymbol{\theta}}).
\end{equation}
We seek to maximize the empirical welfare evaluated at the minimal equilibrium:
\begin{equation}\label{eq:wn}
    \Tilde{D}=\arg\max_{D\in\mathcal{D}_{\kappa}} W_{X,G,\underline{\lambda}}^n(D)=\arg\max_{D\in\mathcal{D}_{\kappa}} W_{X,G,\underline{\sigma}^*}^n(D).
\end{equation}
As shown in Eq.\ref{eq:prlow}, $\underline{\sigma}^{*}$ is a solution to a non-linear simultaneous equation system. The conditional choice probability $\underline{\sigma}^*_i$ of unit $i$ depends non-linearly on the conditional choice probability $\underline{\sigma}^*_j$ and treatment assignment of their neighbors $\{D_j:j\in\mathcal{N}_i\}$. Therefore, when a treatment is assigned
to one unit, it not only influences their behavior but also leads to spillover effects through the network. Hence, Eq.\ref{eq:wn} is a complicated combinatorial optimization problem. We propose a greedy algorithm\footnote{A greedy algorithm is a heuristic approach used in optimization problems; it makes a series of choices that appear to offer the most immediate benefit, building a solution step by step to achieve locally optimal results.} (Algorithm \ref{Aletar}) to solve this problem heuristically.

Intuitively, our greedy algorithm assigns treatment to the unit that contributes most to
the welfare objective, and repeats this until a capacity constraint binds. Specifically, in each
round, Algorithm \ref{Aletar} computes the marginal gain of receiving treatment for each untreated unit, evaluated at the least equilibrium CCP profile. We refer to the unit whose treatment induces the largest increase in the worst-case welfare as the most influential unit for that round. 
\vspace{0.3cm}

 \begin{algorithm}[H]\label{Aletar}
 \setstretch{1}
            \SetAlgoLined
\textbf{Input}: Weighted adjacency matrix $G$, covariates $X$, parameters $\hat{\boldsymbol{\theta}}$, capacity constraint $\kappa$\\
\textbf{Output}: Treatment allocation regime $\hat{D}_{G}$\\
\textbf{Initialization}: $D\gets 0_{N\times 1}$\\
\eIf{$\sum_{i=1}^N D_i< \kappa$}{
\For {$i$ with $D_i=0$}{
$D_i\gets 1$, denote new treatment vector as $D'$\\
$\underline{\sigma}^*(D') \gets$ \text{Computing the minimal equilibrium CCP profile given} $D'$\\
$\Delta_i \gets W_{X,G,\underline{\sigma}^*}^n(D')-W_{X,G,\underline{\sigma}^*}^n(D)$\\
}
$i^*\gets\argmax_{i}\Delta_i$\\
$D_{i^*} \gets 1$\\
}{
$\hat{D}_{G}\gets D$\\
}
\caption{Maximizing Over Treatment Allocation Rules}
 \label{algo}
\end{algorithm}

\section{Theoretical Analysis}\label{sec:theory}
In this section, we analyze the theoretical properties of our proposed treatment allocation
method. To simplify notation, denote the welfare of the targeted population $W_{X,G,\underline{\sigma}^*}(D)$ as $W(D)$, and empirical welfare $W_{X,G,\underline{\sigma}^*}^n(D)$ as $W_n(D)$. In addition, let $W(D^*)$ denote the welfare of the target population at its global optimizer $D^*$, and $W(\hat{D}_{G})$ denote the welfare of the target population welfare under the treatment allocation rule obtained by our proposed method. Let the \textit{regret} of the proposed treatment allocation policy be:
\begin{equation}
   R(\hat{D}_{G}) \coloneqq  \max_{D\in\mathcal{D}}W(D)-W(\hat{D}_{G}).
\end{equation}
We evaluate the performance of our proposed treatment allocation method using \textit{expected regret}, which is defined as: 
\begin{equation}
    \mathbb{E}_{\varepsilon^{n}}\left[R(\hat{D}_{G})\vert S,\sigma^{data}\right]\coloneqq  \max_{D\in\mathcal{D}}W(D)-\mathbb{E}_{\varepsilon^{n}}\left[W(\hat{D}_{G})\vert S,\sigma^{data}\right],
\end{equation}
where the expectation $\mathbb{E}_{\varepsilon^{n}}[\cdot]$ is taken with respect to the uncertainty in the training data\footnote{In the incomplete information setting, the unobserved variables represent units' private information. If the units in the training data are the target population, this may coincide for the training data and the target population. In this case, the expectation in the regret is taken with respect to the uncertainty in the target population. All discussions in this section are otherwise unchanged.} conditional on the observed covariates $\mathsf{X}$, treatment allocation $\mathsf{D}$, network $\mathsf{G}$, and equilibrium $\sigma^{data}$. This is because the randomness in our proposed method primarily arises from utilizing the estimated parameters, which involve only the training data. 
This criterion captures the average welfare loss when implementing estimated policy $\hat{D}_{G}$ relative to the maximum feasible population welfare. Recall that $W_n(\Tilde{D})$ is the maximum for empirical welfare. We decompose regret into (eight terms):
\begin{equation}
    \begin{split}
        W(D^*)-W(\hat{D}_{G}) &= W(D^*)-W_n(D^*)+W_n(D^*)-W_n(\Tilde{D})\\&\quad+W_n(\Tilde{D})-W_n(\hat{D}_{G})+W_n(\hat{D}_{G})-W(\hat{D}_{G}).
   \end{split}
\end{equation}
The first term measures the deviation arising from the use of the empirical social welfare function. This term is bounded by:
\begin{equation}
    W(D^*)-W_n(D^*)\leq \sup_{D\in\mathcal{D}} \vert W_n(D)-W(D) \vert.
\end{equation}
The second term measures the performance of the population welfare maximizer in the empirical social welfare function. This term is bounded by:
\begin{equation}
    W_n(D^*)-W_n(\Tilde{D}) \leq 0.
\end{equation}
The third term measures the loss caused by using a greedy algorithm to solve the optimization problem. This is discussed in Section \ref{sec:greedyregret}.
The final term also measures regret introduced by using the empirical social welfare function. This term is bounded by:
\begin{equation}
    W_n(\hat{D}_{G})-W(\hat{D}_{G})\leq \sup_{D\in\mathcal{D}} \vert W_n(D)-W(D) \vert.
\end{equation}
Combining all the above results, we conclude that expected regret is bounded by:
\begin{equation}\label{eq:regdecom}
\begin{split}
     \mathbb{E}_{\varepsilon^{n}}\left[R(\hat{D}_{G})\vert S,\sigma^{data}\right]&\leq  2\mathbb{E}_{\varepsilon^{n}}\Big[\max_{D\in\mathcal{D}}\lvert W_n(D)-W(D)\rvert\Big \vert S,\sigma^{data}\Big]+\mathbb{E}_{\varepsilon^{n}}\Big[ W_n(\Tilde{D})-W_n(\hat{D}_{G})\Big \vert S,\sigma^{data}\Big].
\end{split}
\end{equation}
In the remainder of this section, we provide a non-asymptotic upper bound for expected regret.

\subsection{Sampling Uncertainty}
For illustrative purposes, this section focuses on engagement welfare. We begin by addressing the regret resulting from the use of estimates in place of true parameters in the payoff function. This represents the sampling uncertainty of the proposed method. We impose the following assumption on the parameter space:
\begin{assumption}{(\textbf{Compactness})}\label{ass:paraspace}
    The parameter $\boldsymbol{\theta}$ lies in a compact set $\Theta\subseteq \mathbb{R}^{d_{\boldsymbol{\theta}}}.$
\end{assumption}
Assumption \ref{ass:paraspace} is standard. We now proceed to characterize the sampling uncertainty associated with using the empirical welfare function.
\begin{lemma}\label{lemma:unctheta}
Under Assumptions \ref{ass:epsilon} and \ref{ass:paraspace}, 
    \begin{equation}
        \mathbb{E}_{\varepsilon^{n}}\Big[\max_{D\in\mathcal{D}}\lvert W_n(D)-W(D)\rvert\Big \vert S,\sigma^{data}\Big]\leq C_1\mathbb{E}_{\varepsilon^{n}}\Big[ \Vert\hat{\theta}-\theta_0\Vert_1\Big \vert S,\sigma^{data}\Big],
    \end{equation}
    where $C_1$ is a constant that depends on the distribution $F_{\varepsilon^{n}}$, and the supports of the parameter space, the covariates space $\mathcal{X}$, the network space $\mathcal{G}$ and the treatment allocation space $\mathcal{D}$.
\end{lemma}
A proof of Lemma \ref{lemma:unctheta} is provided in Appendix \ref{app:unctheta}. Lemma \ref{lemma:unctheta} enables us to characterize the regret of maximizing the empirical welfare through the sampling uncertainty of the structural parameter estimators (i.e., $\mathbb{E}_{\varepsilon^{n}}\big[\Vert \hat{\theta}-\theta\Vert_1 \big\vert S,\sigma^{data}\big]$). As there is no closed-form expression for MLE $\hat{\theta}$ in our case, we study the sampling uncertainty of $\hat{\theta}$ through the sampling uncertainty of its associated empirical process 
\begin{equation}\label{eq:epsu}
    \mathbb{E}_{\varepsilon^{n}}\big[\vert\mathbb{G}_n(\hat{\theta})-\mathbb{G}_n(\theta_0)\vert \big\vert S,\sigma^{data}\big],
\end{equation}
where the empirical process $\mathbb{G}_n(\theta)$ is defined as
\begin{equation}
    \mathbb{G}_n(\theta)\coloneqq\hat{\mathbb{M}}(\theta)-M(\theta),
\end{equation}
\begin{equation}
    \hat{\mathbb{M}}(\theta)=\frac{1}{n}\sum_{i=1}^n Y_i\log\big(F_{\varepsilon}(\hat{Z}_{i}^\intercal\boldsymbol{\theta})\big)+(1-Y_i)\log\big(1-F_{\varepsilon}(\hat{Z}_{i}^\intercal\boldsymbol{\theta})\big),
\end{equation}
and
\begin{equation}
    M(\theta)=\frac{1}{n}\sum_{i=1}^n\mathbb{E}_{\varepsilon^{n}} \left[Y_i\log\big(F_{\varepsilon}(Z_{i}^\intercal\boldsymbol{\theta})\big)+(1-Y_i)\log\big(1-F_{\varepsilon}(Z_{i}^\intercal\boldsymbol{\theta})\big)\Big\vert S,\sigma^{data}\right].
\end{equation}
Recall that $\hat{Z}_i$, as defined in Eq.\ref{eq:Zdefine}, serves as the regressor in our likelihood function. Since we are using a quasi-likelihood ML estimator, the criterion function $\hat{\mathbb{M}}(\theta)$ is evaluated at the estimated equilibrium in the data $\hat{\sigma}^{data}$. As a result, Eq.\ref{eq:epsu} contains two sources of sampling uncertainty: uncertainty from $\hat{\theta}$, and uncertainty from $\hat{\sigma}^{data}$.
The difference between $\mathbb{G}_n(\hat{\theta})$ and $\mathbb{G}_n(\theta_0)$, is given by:
\begin{equation}\label{eq:empricalprocess}
    \begin{split}
        \mathbb{G}_n(\hat{\theta})-\mathbb{G}_n(\theta_0)=\hat{\mathbb{M}}(\hat{\theta})-\hat{\mathbb{M}}(\theta_0)+M(\theta_0)-M(\hat{\theta}).
    \end{split}
\end{equation}
To study the relationship between the estimator and its associated empirical process, we start with a second-order Taylor expansion with Lagrange remainder for both terms in Eq.\ref{eq:empricalprocess}:
\begin{equation}\label{eq:mvt}
    \begin{split}
        \hat{\mathbb{M}}(\theta_0)-\hat{\mathbb{M}}(\hat{\theta})
        =\frac{1}{2}(\hat{\theta}-\theta_0)^{\intercal}\nabla_{\theta}^2\hat{\mathbb{M}}(\acute{\theta})(\hat{\theta}-\theta_0),
    \end{split}
\end{equation}
\begin{equation}\label{eq:mvvt}
    M(\hat{\theta})-M(\theta_0)=\frac{1}{2}(\hat{\theta}-\theta_0)^{\intercal}\nabla_{\theta}^2M(\grave{\theta})(\hat{\theta}-\theta_0),
\end{equation}
for some $\acute{\theta}\in\mathbb{R}^{d_{\theta}}$ and $\grave{\theta}\in\mathbb{R}^{d_{\theta}}
$ on the segment from $\theta_0$ to $\hat{\theta}$. Let $\eta_{max}^0$ denote the largest eigenvalue (in magnitude) of $\nabla_{\theta}^2\hat{\mathbb{M}}(\acute{\theta})$, and $\eta_{max}^1$ denote the largest eigenvalue of $\nabla_{\theta}^2M(\grave{\theta})$. By Assumption \ref{ass:epsilon} (ii), the Hessian matrix is symmetric. Therefore, by the \textit{Courant-Fischer Theorem}\footnote{The largest eigenvalue $\eta_{max}$ of a $C\times C$ symmetric matrix $M$ is given by the maximum Rayleigh quotient (i.e.,
$\eta_{\max}=\max_{A\in\mathbb{R}^{C}\setminus\{0\}}\frac{A^{\intercal}MA}{A^{\intercal}A}$).}, we can characterize the relationship between the parameter sampling uncertainty and the deviation of the criterion function through: 
\begin{equation}\label{eq:mmthe}
    \hat{\mathbb{M}}(\theta_0)-\hat{\mathbb{M}}(\hat{\theta})\leq \frac{1}{2}\eta_{max}^0\Vert\hat{\theta}-\theta_0 \Vert_2^2.
\end{equation}
\begin{equation}\label{eq:mthe}
    M(\hat{\theta})-M(\theta_0)\leq \frac{1}{2}\eta_{max}^1\Vert\hat{\theta}-\theta_0 \Vert_2^2.
\end{equation}
Combining Eq.\ref{eq:empricalprocess} with Eq.\ref{eq:mmthe} and Eq.\ref{eq:mthe} yields
\begin{equation}\label{eq:firststage}
    -(\mathbb{G}_n(\hat{\theta})-\mathbb{G}_n(\theta_0))\leq  \frac{1}{2}(\eta_{max}^0+\eta_{max}^1)\Vert\hat{\theta}-\theta_0 \Vert_2^2.
\end{equation}
Applying the mean value Theorem to the left-hand side of Eq.\ref{eq:firststage}, we have
\begin{equation}\label{eq:mvtg}
    \mathbb{G}_n(\hat{\theta})-\mathbb{G}_n(\theta_0)=(\hat{\theta}-\theta_0)^{\intercal}\nabla_{\theta}\mathbb{G}_n(\Tilde{\theta}),
\end{equation}
for some $\Tilde{\theta}\in\mathbb{R}^{d_{\theta}}$ on the segment from $\theta_0$ to $\hat{\theta}$. Since $\hat{\theta}$ is the maximizer of $\hat{\mathbb{M}}(\cdot)$ and $\theta_0$ is the maximizer of $M(\cdot)$, Eq.\ref{eq:mvtg} must be positive given the definition of $\mathbb{G}_n$. By the Cauchy–Schwarz inequality,
\begin{equation}\label{eq:leftineq}
    \mathbb{G}_n(\hat{\theta})-\mathbb{G}_n(\theta_0)=(\hat{\theta}-\theta_0)^{\intercal}\nabla_{\theta}\mathbb{G}_n(\Tilde{\theta}) = \vert (\hat{\theta}-\theta_0)^{\intercal}\nabla_{\theta}\mathbb{G}_n(\Tilde{\theta})\vert \leq \Vert\hat{\theta}-\theta_0 \Vert_2 \Vert\nabla_{\theta}\mathbb{G}_n(\Tilde{\theta} )\Vert_2.
\end{equation}
Combining Eq.\ref{eq:firststage} and Eq.\ref{eq:leftineq},
\begin{equation}
    -\Vert\hat{\theta}-\theta_0 \Vert_2 \Vert\nabla_{\theta}\mathbb{G}_n(\Tilde{\theta} )\Vert_2\leq \frac{1}{2}(\eta_{max}^0+\eta_{max}^1)\Vert\hat{\theta}-\theta_0 \Vert_2^2.
\end{equation}
Assuming $\hat{\theta}$, an MLE estimator, lies in the interior of the parameter space, the Hessian matrix $\nabla_{\theta}^2\hat{\mathbb{M}}(\acute{\theta})$ is negative definite. Therefore, $\eta_{max}^0$ is negative. In addition, as $\theta_0$ is the maximizer of $M(\theta)$, $\eta_{max}^1$ also must be negative. Hence,
\begin{equation}\label{eq:naivbound}
    \begin{split}
        \Vert\hat{\theta}-\theta_0 \Vert_1\leq d_{\theta}\Vert\hat{\theta}-\theta_0 \Vert_2 &\leq -\frac{2d_{\theta}}{\eta_{max}^0+\eta_{max}^1}\Vert\nabla_{\theta}\mathbb{G}_n(\Tilde{\theta} )\Vert_2\\
        &\leq -\frac{2d_{\theta}}{\eta_{max}^0+\eta_{max}^1}\Vert\nabla_{\theta}\mathbb{G}_n(\Tilde{\theta} )\Vert_1.
    \end{split}
\end{equation}
If $\eta^0_{max}$ and $\eta^1_{max}$ did not depend on the sample size $n$ (i.e., if they were constant), we could study the finite sample properties of $\Vert\hat{\theta}-\theta_0 \Vert_1$ through the finite sample properties of the empirical process $\nabla_{\theta}\mathbb{G}_n(\Tilde{\theta})$. However, the Hessian matrix is a function that depends on the sample, so $\eta_{\max}^0$ and $\eta_{\max}^1$ also depend on $n$. This prevents us from characterizing the finite sample properties of our estimator. 

To overcome this difficulty, we establish a uniform constant upper bound for the largest eigenvalue of the Hessian matrices, which guarantees their strict negativity. We denote the uniform constant upper bound for the largest eigenvalue as maximal largest eigenvlaue and as smallest To obtain a  over all the possible samples, we impose the following three assumptions:
\begin{assumption}{(\textbf{Shape})}\label{ass:smooth}
    $F_{\varepsilon}(x)$ satisfies the condition:
 $\frac{F_{\varepsilon}'(x)^2}{F_{\varepsilon}(x)-1}< F_{\varepsilon}''(x) < \frac{F_{\varepsilon}'(x)^2}{F_{\varepsilon}(x)}$ for all $x\in\mathbb{R}$.
\end{assumption}
\begin{assumption}{(\textbf{Full Rank})}\label{ass:fullrank}
    Let $Z$ denote $[Z_1,...,Z_n]$ and $\hat{Z}$ denote $[\hat{Z}_1,...,\hat{Z}_n]$. We assume that the matrices \( Z \) and \( \hat{Z} \) each have full row rank.
\end{assumption}
\begin{assumption}{(\textbf{Non-Zero Treatment})}\label{ass:treatment} There exists a constant $C_d>0$ such that $\frac{1}{n}\sum_{i=1}^n \mathsf{D}_i\geq C_d, \forall n\in\mathbb{Z}_{+}.$
\end{assumption}
Assumption \ref{ass:smooth} imposes a regularity condition on the shape of the CDF function. The following are two examples of common distributions that satisfy this assumption:

    \begin{itemize}
        \item \textbf{Logistic Distribution}: The cumulative distribution function of the Logistic distribution is: $F_{\varepsilon}(x)=\frac{1}{1+\exp(-x)}$. The corresponding probability density function is: $F_{\varepsilon}'(x)=\frac{\exp(-x)}{(1+\exp(-x))^2}.$ Finally, the second derivative of the CDF is: $F_{\varepsilon}''(x)=\frac{\exp(-3x)-\exp(-x)}{(1+\exp(-x))^4}$. Therefore:
        \begin{equation}
            \frac{F_{\varepsilon}'(x)^2}{F_{\varepsilon}(x)-1}=\frac{-\exp(-x)-\exp(-2x)}{(1+\exp(-x))^4}<F_{\varepsilon}''(x)<\frac{\exp(-3x)+\exp(-2x)}{(1+\exp(-x))^4}=\frac{F_{\varepsilon}'(x)^2}{F_{\varepsilon}(x)}.
        \end{equation}
        \item \textbf{Gaussian Distribution}: Denote the Gaussian cumulative distribution function by $F_{\varepsilon}(x) = \Phi(x)$, its first derivative (the probability density function) by $F_{\varepsilon}'(x) = \phi(x)$, and its second derivative by $F_{\varepsilon}''(x) = -x \phi(x)$. We aim to demonstrate that: $F_{\varepsilon}''(x)<\frac{F_{\varepsilon}'(x)^2}{F_{\varepsilon}(x)}$. Substituting the known forms of $F_{\varepsilon}'(x)$ and $F_{\varepsilon}(x)$, this inequality simplifies to: $\frac{\phi(x)}{\Phi(x)}>-x$. For $x \geq 0$, this inequality is always satisfied. When $x < 0$, we require that: $\frac{\phi(-x)}{\Phi(-x)}>x$ for all $x>0$. As $\frac{\phi(x)}{\Phi(x)}$ is the inverse Mills' ratio, 
        \begin{equation}
            \frac{\phi(-x)}{\Phi(-x)}=\frac{\phi(x)}{1-\Phi(x)}=\mathbb{E}[X\vert X>x]>x, \forall x>0.
        \end{equation}
        By employing a symmetric argument, we find that: $\frac{F_{\varepsilon}'(x)^2}{F_{\varepsilon}(x)-1}< F_{\varepsilon}''(x)$.
    \end{itemize}
Assumption \ref{ass:treatment} guarantees that the average number of treated units in the training data is non-zero for any network size.
Building on Eq.\ref{eq:naivbound}, the following Lemma uniformly characterizes the relationship between the sampling uncertainty of $\hat{\theta}$ and the sampling uncertainty inherent in the empirical process $\mathbb{G}_n(\cdot)$. Formally:
\begin{lemma}{(\textbf{Sampling Uncertainty})}\label{pro.samplingunc}
Under Assumption \ref{ass:epsilon}, \ref{ass:smooth}, \ref{ass:fullrank}, and \ref{ass:treatment}, we have
    \begin{equation}
        \mathbb{E}_{\varepsilon^{n}}\big[\Vert\hat{\theta}-\theta_0 \Vert_1\big\vert S,\sigma^{data}\big]\leq C_2\mathbb{E}_{\varepsilon^{n}}\big[\Vert\nabla_{\theta}\mathbb{G}_n(\Tilde{\theta} )\Vert_1\big\vert S,\sigma^{data}\big],
    \end{equation}
   where $C_2$ is a constant that depends on the distribution $F_{\varepsilon^{n}}$, and the dimension and supports of the parameter space, the covariates space $\mathcal{X}$, the network space $\mathcal{G}$ and the treatment allocation space $\mathcal{D}$.
\end{lemma}
Proof of Lemma \ref{pro.samplingunc} is provided in Appendix \ref{app:prooflemmasmpling}, where we establish a uniform upper bound for the largest eigenvalues (i.e., $\eta_{max}^0$ and $\eta_{max}^1$), termed the maximal largest eigenvalue. We show that this value is strictly negative and is encapsulated within the constant $C_2$. As $C_2$ in Lemma \ref{pro.samplingunc} is a constant, characterising the concentration of the empirical process $\nabla_{\theta}\mathbb{G}_n(\Tilde{\theta})$ is sufficient. The next section does so.

\subsection{Finite Sample Result}
Recall we are using a two-step ML estimation procedure, so the first step of estimation introduces additional sampling uncertainty through $\hat{\sigma}^{data}$. We incorporate these two layers of sampling uncertainty in the following lemma:
\begin{lemma}\label{lemma:gaussian}
    Under Assumption \ref{ass:epsilon} to \ref{ass:fullrank}, we have
    \begin{equation}
      \mathbb{E}_{\varepsilon^{n}}\big[\sup_{\theta\in\Theta}\Vert\nabla_{\theta}\mathbb{G}_n(\theta )\Vert_1\big\vert S,\sigma^{data}\big]\leq \frac{C_3+C_4\sqrt{\log(n)}}{\sqrt{n}},
    \end{equation}
    where $C_3$ and $C_4$ are constants that depend only on the support of covariates, the distribution of $\varepsilon$, $C_{\sigma}$, the covariates space $\mathcal{X}$, the network space $\mathcal{G}$ and the treatment allocation space $\mathcal{D}$.
\end{lemma}
A proof of Lemma \ref{lemma:gaussian} is provided in Appendix \ref{app:prooflemmagaus}. Lemma \ref{lemma:gaussian} analyzes the finite sample property of the empirical process (i.e., $\mathbb{E}_{\varepsilon^{n}}\big[\Vert\nabla_{\theta}\mathbb{G}_n(\theta )\Vert_1\big\vert S,\sigma^{data}\big]$).
By combining the results of Lemma \ref{lemma:unctheta}, Lemma \ref{pro.samplingunc} and Lemma \ref{lemma:gaussian}, we have our main theorem.  This theorem characterizes the sampling uncertainty of using the empirical welfare:
\begin{theorem}{(\textbf{Sampling Uncertainty of Regret})}\label{thm:uncerntain}
    Under Assumption \ref{ass:epsilon} to \ref{ass:treatment}, the sampling uncertainty of the two-step MLE estimator is bounded by:
    \begin{equation}
    \mathbb{E}_{\varepsilon^{n}}\big[\Vert\hat{\theta}-\theta_0 \Vert_1\big\vert S,\sigma^{data}\big]\leq C_2\frac{C_3+C_4\log(n)}{\sqrt{n}}    
    \end{equation}
    In addition, the sampling uncertainty of the empirical welfare is bounded by:
    \begin{equation}
\mathbb{E}_{\varepsilon^{n}}\Big[\max_{D\in\mathcal{D}}\lvert W_n(D)-W(D)\rvert\Big\vert S,\sigma^{data}\big]\leq C_1C_2\frac{C_3+C_4\log(n)}{\sqrt{n}}.
    \end{equation}
\end{theorem}
\noindent A proof of Theorem \ref{thm:uncerntain} is provided in Appendix \ref{appthm:uncertain}. This new result characterizes the finite sample properties of the sampling uncertainty that emerges when utilizing empirical welfare in settings of strategic interaction. It shows that the regret associated with empirical welfare converges at a rate influenced by the size of the network in the training data, as well as by the covariates and the chosen distribution for private information. Furthermore, this result characterizes the performance of the two-step maximum likelihood estimation from a finite sample perspective. This analysis can be extended to general M-estimators, including the Generalized Method of Moments and broader MLE frameworks.

\subsection{Regret due to our Greedy Algorithm}\label{sec:greedyregret}
Now, we evaluate the second term of Eq.\ref{eq:regdecom}, which is the regret introduced by our greedy algorithm,
\begin{equation}
\mathbb{E}_{\varepsilon^{n}}\Big[ W_n(\Tilde{D})-W_n(\hat{D}_{G})\Big \vert S,\sigma^{data}\Big].
\end{equation}
In general, the gap between a greedy optimizer and the global optimizer in terms of the value of the objective function is unknown. For monotone non-decreasing submodular set functions, \citet{nemhauser1978analysis} shows that a greedy algorithm achieves results within $(1-1/e)$ of the global maximum value. Although our optimization problem does not involve a submodular function, our empirical findings in Section \ref{sec:empirical} indicate that our greedy algorithm performs well, a result echoed in other applications such as experimental design \citep{lawrence2002fast}. Building on these findings, \citet{bian2017guarantees} provides a theoretical performance guarantee for using a greedy algorithm on non-submodular functions by leveraging the submodularity ratio and curvature of the objective function.

Submodularity, the submodularity ratio, and the curvature of a set function $f$ are defined as follows.
\begin{definition}{(\textbf{Submodularity})} A set function is a submodular function if:
\begin{equation}
    \sum_{k\in R\setminus S} [f(S\cup \{k\})-f(S)]\geq  f(S\cup R)-f(S), \quad\forall S,R\subseteq\mathcal{N}.
\end{equation}
\end{definition}
\begin{definition}
    (\textbf{Submodularity Ratio}) The submodularity ratio of a non-negative set function $f(\cdot)$ is the largest $\gamma$ such that
\begin{equation}
    \sum_{k\in R\setminus S} [f(S\cup \{k\})-f(S)]\geq \gamma [f(S\cup R)-f(S)], \quad\forall S,R\subseteq\mathcal{N}.
\end{equation}
\end{definition}

\begin{definition}
    (\textbf{Curvature}) The curvature of a non-negative set function $f(\cdot)$ is the smallest value of $\xi$ such that
\begin{equation}
f(R\cup \{k\})-f(R)\geq (1-\xi)[f(S\cup \{k\})-f(S)],\quad   \forall S\subseteq R\subseteq \mathcal{N}, \forall k\in \mathcal{N}\setminus R.
\end{equation}
\end{definition}
Submodularity is similar to diminishing returns. It states that adding an element to a smaller set yields a greater benefit than adding it to a larger set. \citet{lovasz1983submodular} highlights that, in discrete optimization, submodularity plays a role analogous to convexity in continuous optimization. The submodularity ratio measures how close a set function is to being submodular \citep{das2011submodular}. Curvature quantifies the extent to which a set function deviates from being additive.

We evaluate the theoretical performance of our greedy algorithm in scenarios where the treatment exerts both direct and indirect positive effects on equilibrium welfare, as indicated by positive values for $(\hat{\theta}_1+X_i^{\intercal}\hat{\theta}_3)$ and $\hat{\theta}_5$. Additionally, our empirical analysis explores a variety of other scenarios, including those with a negative direct effect but a positive indirect effect, among others. The results indicate that the algorithm performs well across a range of conditions.

To characterize the submodularity ratio and curvature of the objective function, we first represent it as a set function, which is a real-valued mapping defined over treatment allocations sets, $\mathcal{D}\subset \mathcal{N}$ (i.e., $\mathcal{D}=\{i\in\mathcal{N}:D_i=1\}$):
\begin{multline}
      W_n(\mathcal{D})= \sum_{i\in \mathcal{D}}F_{\varepsilon}\Big[\hat{\theta}_0+\hat{\theta}_1 +X_i^{\intercal}(\hat{\theta}_2+\hat{\theta}_3) +\frac{1}{\vert\mathcal{N}_i\vert} \sum_{\substack{j\in\mathcal{D}\setminus\{i\}}}(\hat{\theta}_4+\hat{\theta}_6\underline{\sigma}_j)m_{ij}G_{ij}+\frac{\hat{\theta}_5}{\vert\mathcal{N}_i\vert}\sum_{\substack{j\in \mathcal{N}\setminus\{i\}}}m_{ij}G_{ij}\underline{\sigma}_j\Big]\\
      +\sum_{k\in\mathcal{N}\setminus\mathcal{D}}F_{\varepsilon}\Big[\hat{\theta}_0+X_k^{\intercal}\hat{\theta}_2+\frac{\hat{\theta}_4}{\vert\mathcal{N}_k\vert} \sum_{\substack{\ell\in\mathcal{D}}}m_{k\ell}G_{k\ell}+\frac{\hat{\theta}_5}{\vert\mathcal{N}_k\vert}\sum_{\ell\in\mathcal{N}\setminus\{k\}}m_{k\ell}G_{k\ell}\underline{\sigma}_\ell\Big].
\end{multline}
   
Let $\gamma$ denote the submodularity ratio. For a nondecreasing function, $\gamma$ ranges between $[0,1]$ and is $1$ if and only if the function is submodular. Similarly, the curvature, denoted by $\xi$, of a nondecreasing function ranges between $[0,1]$, and is $0$ if and only if the function is supermodular. As our objective function involves a system of simultaneous equations, evaluating its curvature and submodularity ratio directly is challenging. Instead, we focus on the upper bound for curvature and submodularity, and ensure that their values remain within $(0,1)$. Combining this result with \citet[Theorem 1]{bian2017guarantees} of leads to:
\begin{proposition}\label{pro:greedy}
    Under Assumption \ref{ass:epsilon} and Assumption \ref{ass:paraspace}, the curvature $\xi$ of $W_n(\mathcal{D})$ and the submodularity ratio $\gamma$ of $W_n(\mathcal{D})$ are in $(0,1)$. 
         The greedy algorithm enjoys the following approximation guarantee for the problem in Eq.\ref{eq:wn}:
\begin{equation}
    W_n(D_G)\geq \frac{1}{\xi}(1-e^{-\xi\gamma}) W_n(\Tilde{D}),
\end{equation}
where $D_G$ is the treatment assignment rule that is obtained by Algorithm \ref{algo}.
\end{proposition} 
A proof is provided in Appendix \ref{apppro:greedy}. This proof is similar to  \citet{kitagawa2023individualized}. The first part of Proposition \ref{pro:greedy} implies that the performance guarantee is a non-trivial bound. Although the curvature and submodularity ratio of our objective function are unknown, for a particular application, it is possible to evaluate them empirically. As a consequence,
\begin{equation}\label{eq:regretgreedy}
    \mathbb{E}_{\varepsilon^{n}}\Big[ W_n(\Tilde{D})-W_n(\hat{D}_{G})\Big \vert S,\sigma^{data}\Big]\leq \mathcal{O}(1)(1-\frac{1}{\xi}(1-e^{-\xi\gamma})),
\end{equation}
where $\mathcal{O}(1)$ captures the $\mathbb{E}_{\varepsilon^{n}}[W_n(\Tilde{D}) \vert S,\sigma^{data}]$. Combining Eq.\ref{eq:regretgreedy} with Theorem \ref{thm:uncerntain}, we obtain our main theorem:
\begin{theorem}{(\textbf{Regret Bound})}\label{thm:regbound}
Let $D^*$ denote the maximizer of $W(D)$ and $D_G$ be the assignment vector obtained by Algorithm \ref{algo}. Under Assumptions \ref{ass:epsilon} to \ref{ass:fullrank}, given curvature $\xi$ and submodularity ratio $\gamma$, the regret is bounded from above by: 
\begin{equation}
  \mathbb{E}_{\varepsilon^{n}}\left[R(\hat{D}_{G})\vert S,\sigma^{data}\right]\leq \mathcal{O}\big(\log(1+C_3\sqrt{n})/\sqrt{n}\big)+\mathcal{O}(1)(1-\frac{1}{\xi}(1-e^{-\xi\gamma})), \label{eq:theorem 5.2 regret bound}
    \end{equation}
    where $C_3 = 4\vert\underline{\omega}\vert\widebar{z}^2$.
\end{theorem}
Theorem \ref{thm:regbound} is our key result. 
The first term in Eq.\ref{eq:theorem 5.2 regret bound} characterizes the sampling uncertainty, whose convergence rate depends on the network size. 
The dependence upon the parameters in the utility function, network structure, and private information distribution are shown implicitly via the terms $C_1$ in Lemma \ref{lemma:unctheta}, $C_5$ and $C_6$ in Lemma \ref{lemma:b}, and $C_4$ in Lemma \ref{lemma:a}. The second term comes from the use of a greedy algorithm, and converges to a constant.

\section{Empirical Application}\label{sec:empirical}
We illustrate our proposed method using data from \citet{banerjee2013diffusion}, which explores the impact of information provision on microfinance adoption. \citet{banerjee2013diffusion} studies a microfinance loan program. This program was introduced by Bharatha Swamukti Samsthe (BSS), a non-governmental microfinance institution in India, and implemented across 43 villages in Karnataka. BSS invited influential units, such as teachers, leaders of self-help groups, and shopkeepers, to an informational meeting about the availability of microfinance (the treatment). In total, 1262 units were assigned treatment, an average of 25.75 per village. After the intervention, researchers collected data on the network structure and household characteristics—including access to electricity, latrine quality, and per capita counts of beds and rooms in all participating villages. The number of households in each village varied from 107 to 341, with 10 to 51 households per village receiving information about the program. The program commenced in 2007, and the survey of microfinance adoption was completed by early 2011. We treat each household’s decision to purchase microfinance as an equilibrium outcome within a simultaneous decision network game. 

We consider each village as a distinct target population. Structural parameters in the payoff function for each village are estimated separately using the two-step maximum likelihood estimation (MLE) method of \citet{leung2015two}. This setup assumes that the training data, which includes several villages, acts as a representative sample, with each village in the training dataset mirroring a corresponding village in the target population in terms of covariates and network structure. 

In the first stage of our analysis, following \citet{arcidiacono2011conditional}, we estimate the conditional choice probability using a flexible Logit approach (Chi-square goodness of fit test result is provided in Appendix \ref{chisquare}). This includes each unit's covariates and their second powers, as well as the covariates of directly linked neighbors and interactions among these covariates, which is under the network decaying dependence assumption \citep{xu2018social}. Although our estimator allows for incorporating covariates from neighbors at higher levels of linkage, we focus on directly linked neighbors' covariates in this estimation. 
We treat these estimates as the true parameters and assess the presence of strategic complementarity in each village. We find strategic complementarities in 16 of the 43 villages in the dataset\footnote{The indices of those 16 villages in the original data set are: 1, 4, 6, 7, 12, 14, 17, 18, 20, 24, 25, 29, 31, 39, 40, and 41. To enhance clarity, we discard their original indices and re-label them as villages 1 to 16.}, which are the focus of this exercise. We assume the policymaker utilizes all available covariates to determine the treatment allocation mechanism. We assume that the private information follows a logistic distribution, and we define the measure of closeness between units $i$ and $j$ to be $m(X_i,X_j) = \frac{1}{1+\vert X_i-X_j\vert}$.

\subsection{Policy evaluation}
In this application, the objective is to maximize engagement welfare, measured as the microfinance participation rate, evaluated at the minimal equilibrium under a treatment capacity constraint (as in Eq.\ref{eq:target}) within our target population. To ensure comparability with the original study, we set the capacity constraint equal to the number of treatments used by Bharatha Swamukti Samsthe (BSS). We compare our method (`Robust') with two different treatment allocation regimes: the allocation rule adopted by BSS in the original study (`Original'), and a random allocation rule (`Random').

Table \ref{tableemprical} presents predicted village-level microfinance take-up probabilities under three different treatment allocations. For each allocation rule, we report both the upper and lower bounds of the prediction set. The first column lists the 16 villages that exhibit strategic complementarities. The second column (\textit{Sample Avg.}) contains the empirical average take-up rate for these villages. The third column (\textit{Welfare under Original}) shows the average adoption rates for the original treatment allocation used by BSS. \textbf{Four villages—Villages 2, 3, 6, and 12—exhibit multiple equilibria under this rule}. To further assess our proposed method's performance, we generate 500 random treatment allocations within the capacity constraint for each village. The average purchasing probability across these 500 allocations is reported in the fourth column (\textit{Welfare under Random}). Under random allocation, \textbf{multiple equilibria arise in Villages 2, and 9, highlighting that the occurrence of multiple equilibria can vary with the allocation method used}. The share of households adopting microfinance according to the robust optimal treatment allocation is shown as 
\textit{Welfare under Robust}, \textbf{where the multiple equilibria only presents in the Village 2}.

\begin{table}[ht]
\setstretch{1}
\begingroup
\setlength{\tabcolsep}{3pt} 
\renewcommand{\arraystretch}{0.3}
 \begin{adjustwidth}{0cm}{}
 \linespread{1}
\footnotesize
\centering 
\begin{threeparttable}
 \begin{tabular}{@{}lcccccc@{}}
   \hline
   \toprule
    \multirow{2}{*}{\textit{Village}}&\multirow{2}{*}{\textit{Sample Avg.}}& \multicolumn{3}{c}{ \textbf{Welfare under}}&\multicolumn{2}{c}{\textbf{Welfare Gain$^*$}}\\
    &&\textit{Original}&\textit{Random}&\textit{Robust}&\textit{Level }&\textit{Percentage}\\
  \midrule
  \textbf{1}& $0.24$&[$0.25$, $0.25$] 
  &[$0.20$, $0.20$] &[$0.41$, $0.41$]&$0.16$& $66\%$\\[3pt]
  \textbf{2}& $0.08$& [$0.04$, $0.07$] 
  &[$0.03$, $0.04$] &[$0.13$, $0.16$]&$0.10$ & $270\%$\\[3pt]
  \textbf{3}& $0.18$& [$0.17$, $0.23$]
  &[$0.25$, $0.25$]&[$0.37$, $0.37$]&$0.20$ & $122\%$\\[3pt]
  \textbf{4}& $0.30$& [$0.26$, $0.26$] 
  &[$0.29$, $0.29$] &[$0.44$, $0.44$]&$0.18$ & $68\%$\\[3pt]
  \textbf{5}& $0.15$& [$0.15$, $0.15$] 
  &[$0.16$, $0.16$] &[$0.37$, $0.37$]&$0.22$ & $146\%$\\[3pt]
  \textbf{6}& $0.17$& [$0.15$, $0.19$] 
  &[$0.17$, $0.17$] &[$0.39$, $0.39$]&$0.24$ & $158\%$\\[3pt]
  \textbf{7}& $0.19$& [$0.48$, $0.48$] 
  &[$0.40$, $0.40$] &[$0.66$, $0.66$] &$0.18$ & $37\%$\\[3pt]
  \textbf{8}& $0.19$& [$0.17$, $0.17$] 
  &[$0.18$, $0.18$] &[$0.28$, $0.28$] &$0.11$ & $64\%$\\[3pt]
  \textbf{9}& $0.19$& [$0.19$, $0.19$] 
  &[$0.20$, $0.23$] &[$0.33$, $0.33$] &$0.14$ & $72\%$\\[3pt]
  \textbf{10}& $0.24$& [$0.24$, $0.24$] 
  &[$0.23$, $0.23$] &[$0.28$, $0.28$]&$0.05$ & $20\%$\\[3pt]
  \textbf{11}& $0.23$& [$0.22$, $0.22$] 
  &[$0.22$, $0.22$] &[$0.34$, $0.34$]&$0.12$ & $57\%$\\[3pt]
  \textbf{12}& $0.10$& [$0.09$, $0.10$]
  &[$0.11$, $0.11$] &[$0.30$, $0.30$]&$0.20$ & $223\%$\\[3pt]
  \textbf{13}& $0.15$& [$0.13$, $0.13$] 
  &[$0.15$, $0.15$] &[$0.45$, $0.45$]&$0.31$ & $238\%$\\[3pt]
  \textbf{14}& $0.21$& [$0.23$, $0.23$] 
  &[$0.19$, $0.19$] &[$0.46$, $0.46$]&$0.23$ & $101\%$\\[3pt]
  \textbf{15}& $0.16$& [$0.18$, $0.18$] 
  &[$0.18$, $0.18$] &[$0.41$, $0.41$]&$0.24$ & $132\%$\\[3pt]
  \textbf{16}& $0.16$& [$0.16$, $0.16$] 
  &[$0.19$, $0.19$] &[$0.29$, $0.29$]&$0.13$ & $85\%$\\[3pt]
\bottomrule
   \end{tabular}
   \end{threeparttable}
   \caption{Comparison using 16 Indian villages microfinance data from \cite{banerjee2013diffusion}\\
   $\ast$ Minimal Welfare Gain compares the minimal  simulated equilibrium welfare under the Robust allocation method and the simulated equilibrium welfare under the Original allocation implemented by BSS.}
\label{tableemprical}
   \end{adjustwidth}
   \endgroup
\end{table}
\begin{figure}[ht]
    \centering
\includegraphics[height=9cm,width=0.7\linewidth]{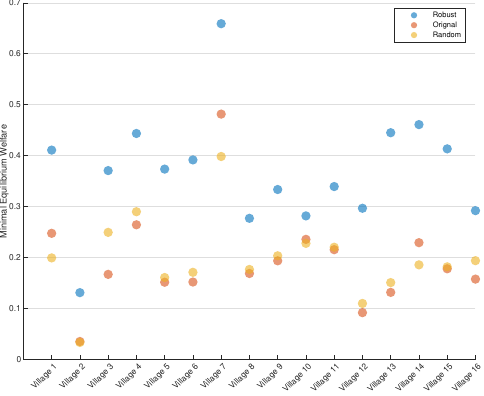}
    \caption{Comparisons between three approaches}
    \label{fig:comparisionlower}
\end{figure}

Note first that the equilibrium average share of households adopting microfinance under the original allocation closely tracks the observed data for all villages except for Village 7 \footnote{The goodness of fit test indicates that the estimation for Village 7, referred to as Village 17 in Table \ref{tab:chisquare}, may not adequately fit the data, potentially due to inaccuracies in the first-stage Conditional Choice Probability (CCP) estimation.}. Second, we find that the equilibrium average share of households purchasing microfinance under random allocation is similar to the original BSS allocation method. When comparing our robust optimal treatment allocation regime with the original allocation rule, our method consistently outperforms the original rule in terms of both minimal and maximal equilibrium welfare.  As depicted in Figure \ref{fig:comparisionlower}, improvements in welfare with minimal equilibrium vary from $20\%$ to $270\%$. Notice that the welfare at the minimal equilibrium of our approach surpasses the maximal welfare under the other two approaches. This suggests that the information diffusion facilitated by the original treatment may not have significantly impacted adoption rates. Additionally, \citet{wang2024graph} finds that households with higher centrality, such as the leaders selected by BSS, tend to have a lower borrowing probability compared to less central households. It is possible that more central households have greater access to alternative borrowing sources within their networks, thus diminishing their need for microfinance, and reducing the spillover effects through strategic interactions.

\section{Extension}\label{sec:complete}
\subsection{Complete Information Game}
In a complete information setting, units observe all the characteristics of other units participating in the game. This means that units are informed of others' choices before making their own decisions, allowing them to play the best response to the observed actions rather than basing their actions on beliefs, as is common in a private information setting.
As a consequence, unit $i$'s decision rule is:
\begin{equation}
        Y_i=\mathds{1}\Big\{U_i(1,Y_{-i},X,D,G)\geq 0\Big\},\quad \forall i\in\mathcal{N}.
    \end{equation}
One main distinction from incomplete information settings is that the solution concept transitions to a pure-strategy Nash equilibrium. A pure-strategy Nash equilibrium is defined by a set of actions $y^*=\{y_1^*,...,y_N^*\}$ such that
\begin{equation}
    U_i(y_i^*,y_{-i}^*,X,D,G)\geq U_i(y_i',y_{-i}^*,X,D,G) 
\end{equation}
for any $y_i'\in\mathcal{Y}$ and for all $i\in\mathcal{N}$. We denote the set of all such equilibria as $\Sigma(X,D,G,\varepsilon)\coloneqq\{y^*\}$, given covariates $X$, treatment allocation $D$, network structure $G$, and the idiosyncratic shock $\varepsilon$. To simplify the notation, we subsequently refer to it as $\Sigma(\varepsilon)$. Let $\xi:\Sigma\rightarrow[0,1]$ denote the probability distribution over equilibria, and let $\Delta(\Sigma)\coloneqq\{\xi:\sum_{y^*\in\Sigma}\xi(y^*)=1\}$ denote the set of all the probability distributions.

In scenarios with strategic complementarity, there exists a maximal and a minimal Nash equilibrium, denoted by $\widebar{y}^*$ and $\underline{y}^*$. For our counterfactual analysis, which is analogous to the framework established in Theorem \ref{thm:equality} under an incomplete information setting, we propose the following:
\begin{proposition}\label{pro:complete}
For a supermodular game, the least favorable equilibrium selection rule $\underline{\lambda}$ and the most favorable equilibrium selection rule $\widebar{\lambda}$ are:
    \begin{equation}
        \underline{\lambda} \coloneqq \delta_{\underline{y}^*}, \quad \widebar{\lambda} \coloneqq \delta_{\widebar{y}^*},
    \end{equation}
    where $\delta_{y}$ is the Dirac measure on $\Sigma$. 
In addition, the following conditions are satisfied:
\begin{equation}
    \begin{split}
\inf_{\lambda\in\Lambda}\sum_{i=1}^N \Pr(Y_i=1\vert X,D,G,\lambda)
      =\sum_{i=1}^N\inf_{\lambda\in\Lambda} \Pr(Y_i=1\vert X,D,G,\lambda),
    \end{split}
\end{equation}
\begin{equation}
    \begin{split}
\sup_{\lambda\in\Lambda}\sum_{i=1}^N \Pr(Y_i=1\vert X,D,G,\lambda)
      =\sum_{i=1}^N\sup_{\lambda\in\Lambda} \Pr(Y_i=1\vert X,D,G,\lambda).
    \end{split}
\end{equation}   
\end{proposition}
The proof of Proposition \ref{pro:complete} mirrors that of Theorem \ref{thm:equality}, with the primary modification being the substitution of Bayesian Nash equilibrium with Nash equilibrium. Computing the conditional choice probability differs from the previous analysis since it is no longer a simultaneous equation system. With complete information, the conditional choice probability is given by:
\begin{equation}
   \begin{split}
       & \Pr(Y_i=1\vert X,D,G,\underline{\lambda}) = \int 1\left\{\exists y_{-i}:(1,y_{-i})\in \Sigma(\varepsilon) \:\text{and}\: \forall y_{-i}, (0,y_{-i})\notin \Sigma(\varepsilon)   \right\}dF_{\varepsilon}
   \end{split}
\end{equation}
\begin{equation}
   \begin{split}
        \Pr(Y_i=1\vert X,D,G,\overline{\lambda}) = \int 1\left\{\exists y_{-i}:(1,y_{-i})\in \Sigma(\varepsilon) \right\}dF_{\varepsilon}
   \end{split}
\end{equation}
The above expression is hard to compute. To further simplify the computation, let us define event $A\coloneqq\left\{\exists y_{-i}:(1,y_{-i})\in \Sigma(\varepsilon) \right\}$ and event $B\coloneqq \left\{\exists y_{-i}, (0,y_{-i})\in \Sigma(\varepsilon)   \right\}$. Given $\Pr(A\cap B^c\vert X,D,G) = \Pr(A\cup B\vert X,D,G)-\Pr(B\vert X,D,G)$, we hence have:
\begin{equation}
    \Pr(Y_i=1\vert X,D,G,\underline{\lambda}) = 1-\Pr(Y_i=0\vert X,D,G,\overline{\lambda}).
\end{equation}
As a consequence, it is enough to compute $\Pr(Y_i=1\vert X,D,G,\widebar{\lambda})$ and $\Pr(Y_i=0\vert X,D,G,\widebar{\lambda})$, which are given by:
\begin{equation}\label{eq:uppermax}
    \begin{split}
        \Pr(Y_i=1\vert X,D,G,\overline{\lambda}) 
        =\int 1\left\{\max_{y_{-i}:(1,y_{-i})\in \Sigma(\varepsilon)}U_i(1,y_{-i},X,D,G)\geq 0\right\}dF_\varepsilon,
    \end{split}
\end{equation}
\begin{equation}\label{eq:uppermin}
    \begin{split}
        \Pr(Y_i=0\vert X,D,G,\overline{\lambda}) 
        =\int 1\left\{\min_{y_{-i}:(1,y_{-i})\in \Sigma(\varepsilon)}U_i(1,y_{-i},X,D,G)<0\right\}dF_\varepsilon.
    \end{split}
\end{equation}
Let us define \(i^C\) as the complement of unit \(i\) and their neighbor set, i.e., \(i^C \coloneqq \mathcal{N} \setminus (\mathcal{N}_i \cup \{i\})\). Define \(y_{\mathcal{N}_i}\) as the collection of neighbors' choices of unit \(i\). Consequently, \(y_{-i}\) can be expressed as \((y_{\mathcal{N}_i}, y_{i^C})\), where \(y_{i^C}\) represents the choices of units in \(i^C\). For the optimization problems defined in Eq.\ref{eq:uppermax} (maximization) and Eq.\ref{eq:uppermin} (minimization), it is necessary to explore all possible equilibria for each value of \(\varepsilon\) within the network game. Given our utility function specification, the choice of unit \(i\) depends only on \(k \in i^C\) through the choices of units directly connected with \(i\). Thus, we can simplify the maximization problem in Eq.\ref{eq:uppermax} to:
\begin{equation}
    \max_{y_{\mathcal{N}_i}, y_{i^C}} U_i(1, y_{\mathcal{N}_i}, X, D, G)
\end{equation}
with constraints:
\begin{equation}\label{eq:con1}
    y_j = 1\{U_j(1, y_{\mathcal{N}_j \setminus \{i\}}, X, D, G) \geq 0\}, \forall j \in \mathcal{N}_i,
\end{equation}
\begin{equation}\label{eq:con2}
    y_k = 1\{U_k(1, y_{\mathcal{N}_k}, X, D, G) \geq 0\}, \forall k \in i^C.
\end{equation}
These constraints (Eq.\ref{eq:con1} and Eq.\ref{eq:con2}) ensure that \((1, y_{\mathcal{N}_i}, y_{i^C})\) forms a Nash equilibrium. In the optimization, \(U_i(1, y_{\mathcal{N}_i}, X, D, G)\) does not depend directly on \(y_{i^C}\) but needs to confirm that \((1, y_{\mathcal{N}_{i}}, y_{i^C})\) is a Nash equilibrium for any given \(y_{\mathcal{N}_{i}}\).
If for some \(y_{\mathcal{N}_{i}}\), multiple \(y_{i^C}\) ensure \(y_{-i}\) as a Nash equilibrium, the existence of any \(y_{i^C}\) that satisfies this condition is sufficient for our purposes. We search for \(y_{\mathcal{N}_{i}} \in \mathcal{Y}^{\vert \mathcal{N}_i \vert}\) that maximizes \(U_i(y_{\mathcal{N}_i}, X, D, G)\), denoted as \(y^*_{\mathcal{N}_{i}}\). Given the supermodular nature of our game, where neighbors' choices are strategic complements to unit \(j\)'s choice, we select \(y_{i^C}\) such that \((y_{\mathcal{N}_{i}}, y_{i^C})\) constitutes the largest Nash equilibrium for the given \(y_{\mathcal{N}_{i}}\), leveraging the increasing monotonicity between \(y_{\mathcal{N}_{i}}\) and \(y_{i^C}\). We then search the $y_{\mathcal{N}_{i}}$ that maximizes the objective function.




\section{Conclusion}\label{sec:conclude}

This paper proposes a method for constructing individualized treatment allocations to maximize equilibrium welfare robust to the presence of multiple equilibria in large simultaneous decision games with complementarity. Our approach, takes into account the inherent complexity introduced by the presence of multiple Nash equilibria, and the resulting incompleteness. We refrain from making assumptions about the equilibrium selection mechanism, which leads to both analytical and numerical challenges in evaluating counterfactual equilibrium welfare. Due to the inherent uncertainties in our model, we use the maximin welfare criterion to evaluate treatment allocation rules. This leads to treatment allocation rules that are optimized to maximize the worst-case equilibrium social welfare, ensuring their robustness. The use of a greedy optimization algorithm further enhances the applicability of our approach.

We acknowledge that several questions remain open, and there are multiple ways in which our work can be extended. First, we have not explored counterfactual analysis within the broader framework of general simultaneous decision games. Second, although we parametrize the utility function and the distribution of idiosyncratic shock in this work, adopting a non-parametric utility function and a non-parametric distribution of idiosyncratic shock could significantly enhance the robustness and applicability of our approach. Third, while we have assumed independence among idiosyncratic shocks, recent literature, such as \citet{grieco2014discrete} and \citet{de2020testable}, have begun to relax this assumption, suggesting another avenue for refining our model.

\begin{singlespace}
\bibliographystyle{ecta}  
\bibliography{ref}
\end{singlespace}

\newpage
\linespread{1}
\begin{appendices}

\begin{center}\LARGE {Supplementary Materials of Robust Network Targeting with Multiple Nash
Equilibria}
\end{center}

\bigskip

\section{Chi-Square Goodness of Fit Test}\label{chisquare}
\begin{table}[h!]
\setstretch{1}
\linespread{1}
\renewcommand{\arraystretch}{0.5}
\footnotesize
\centering
\begin{threeparttable}
\begin{tabular}{@{}lcccc|lcccc@{}}
\toprule
\multicolumn{5}{c}{Number of Rooms} & \multicolumn{5}{c}{Number of Rooms} \\
\cmidrule(r){1-5} \cmidrule(l){6-10}
\textit{Village} & \textit{1-2} & \textit{3-4} & \textit{5-6} & \textit{$\geq 7$} & \textit{Village} & \textit{1-2} & \textit{3-4} & \textit{5-6} & \textit{$\geq 7$} \\[1pt]
\midrule
\textbf{1} & 0.08 & 0.04 & 0.61 & 0.38 & \textbf{23} & 0.06 & 0.22 & 3.27 & 0.00 \\[1pt]
\textbf{2} & 92.90 & 17.03 & 0.13 & -- & \textbf{24} & 0.31 & 0.19 & 0.60 & -- \\[1pt]
\textbf{3} & 0.75 & 0.14 & 0.53 & -- & \textbf{25} & 0.06 & 0.86 & 1.52 & -- \\[1pt]
\textbf{4} & 4.81 & 0.50 & 0.62 & 0.00 & \textbf{26} & 0.04 & 0.14 & 0.15 & -- \\[1pt]
\textbf{5} & 0.02 & 0.02 & 2.08 & 0.15 & \textbf{27} & 0.20 & 0.13 & 0.80 & 0.01 \\[1pt]
\textbf{6} & 0.35 & 0.15 & 0.02 & 0.04 & \textbf{28} & 345.59 & 115.88 & 6.39 & 1.76 \\[3pt]
\textbf{7} & 0.82 & 0.45 & 0.74 & 0.00 & \textbf{29} & 0.02 & 0.01 & 1.03 & 0.08 \\[1pt]
\textbf{8} & 0.07 & 0.10 & 0.00 & 0.45 & \textbf{30} & 2.42 & 1.02 & 0.00 & 0.00 \\[1pt]
\textbf{9} & 1.76 & 0.81 & 0.05 & -- & \textbf{31} & 0.18 & 0.97 & 0.19 & -- \\[1pt]
\textbf{10} & 0.01 & 0.02 & 1.38 & 0.01 & \textbf{32} & 0.76 & 0.36 & 0.00 & 0.87 \\[1pt]
\textbf{11} & 0.09 & 0.03 & 0.34 & 0.10 & \textbf{33} & 1.61 & 0.12 & 0.72 & 1.55 \\[1pt]
\textbf{12} & 0.00 & 1.16 & 1.47 & 0.35 & \textbf{34} & 0.12 & 0.31 & 0.11 & -- \\[1pt]
\textbf{13} & 1.67 & 0.06 & 0.26 & 3.79 & \textbf{35} & 0.21 & 0.45 & 0.45 & 0.09 \\[1pt]
\textbf{14} & 0.01 & 0.10 & 0.92 & 0.00 & \textbf{36} & 0.11 & 0.80 & 0.68 & 0.93 \\[1pt]
\textbf{15} & 0.05 & 0.00 & 0.58 & 1.11 & \textbf{37} & 0.04 & 0.03 & 0.17 & 0.03 \\[1pt]
\textbf{16} & 0.03 & 0.04 & 0.24 & 0.01 & \textbf{38} & 0.00 & 0.09 & 0.32 & -- \\[1pt]
\textbf{17} & 80.36 & 7.58 & 0.37 & 0.34 & \textbf{39} & 0.35 & 0.03 & 0.09 & -- \\[1pt]
\textbf{18} & 2.20 & 1.76 & 0.73 & 5.54 & \textbf{40} & 0.43 & 0.02 & 0.03 & -- \\[1pt]
\textbf{19} & 0.00 & 2.02 & 0.60 & 0.24 & \textbf{41} & 0.01 & 0.00 & 0.07 & 0.00 \\[1pt]
\textbf{20} & 0.04 & 0.23 & 0.01 & 0.89 & \textbf{42} & 49.43 & 10.70 & 1.26 & 0.00 \\[3pt]
\textbf{21} & 0.04 & 0.00 & 0.25 & 0.14 & \textbf{43} & 0.82 & 0.00 & 0.58 & 0.03 \\[1pt]
\textbf{22} & 0.42 & 0.93 & 2.04 & 0.05 &  &  &  &  &  \\[3pt]
\bottomrule
\end{tabular}
\caption{\footnotesize Chi-square values based on number of rooms for $43$ Indian villages microfinance data from \citet{banerjee2013diffusion}\\
At $0.05$ significance level, the critical value is given by $7.815$.\\
$'-'$ symbol indicates that there are no households with more than 7 rooms.}
\label{tab:chisquare}
\end{threeparttable}
\end{table}

\section{Lemmas}\label{appendixA}
We first introduce notation. We define $\widebar{m}\coloneqq \max_{ij}\vert m_{ij}\vert$, and $\underline{F}_{\varepsilon}\coloneqq \min_{\substack{ \theta\in\Theta\\z\in \mathcal{Z}}}F_{\varepsilon}(z^{\intercal}\theta)$. In addition, we define $\upsilon\coloneqq\min\{\underline{F}_{\varepsilon},1-\widebar{F}_{\varepsilon}\}$. We measure the distance between parameters $\theta$ with the $L_1$ metric, which we denote by $\Vert\theta-\theta'\Vert_1\coloneqq \sum_{k=1}^{d_{\theta}}\vert\theta_{k}-\theta_{k}'\vert$. For a $K\times L$ matrix $A$, $\Vert A\Vert_{\infty}$ denotes the operator norm of $A$ induced by the $L_{\infty}$ norm, which is given as: $\displaystyle\Vert A\Vert_{\infty}=\max_{k=1,...,K}\sum_{l=1}^L\vert A_{kl}\vert.$

\subsection{Proof of Corollary \ref{coro:utility}}\label{app:oroofcoroutility}
\begin{corollary} (\textbf{Utilitarian Welfare at Equilibrium})
Under Assumption \ref{ass:epsilon}, given the specification of our utility function,
the predicted set of the expected utilitarian welfare under a counterfactual policy $D$ is given as:
\begin{equation}
    \begin{split}
        W_{X,G,\lambda}(D)\in&\Big[\frac{1}{N}\sum_{i=1}^N\alpha_i f(\alpha_i)+\frac{1}{N}\sum_{i=1}^N\sum_{j\neq i}\beta_{ij} \underline{\sigma}^*_i\underline{\sigma}^*_j, \frac{1}{N}\sum_{i=1}^N\alpha_i f(\alpha_i)+\frac{1}{N}\sum_{i=1}^N\sum_{j\neq i}\beta_{ij} \widebar{\sigma}^{*}_i\widebar{\sigma}^{*}_j\Big],
    \end{split}
\end{equation}
where 
\begin{equation}
    f(\alpha_i) = 
\begin{cases} 
\Pr(Y_i=1\vert X,D,G,\underline{\lambda}) & \text{if } \alpha_i > 0 \\
\Pr(Y_i=1\vert X,D,G,\widebar{\lambda}) & \text{if } \alpha_i \leq 0 .
\end{cases}
\end{equation} 
\end{corollary}
\begin{proof}
Given
    \begin{equation}
        \begin{split}
            W_{X,G,\lambda}(D)=\frac{1}{N}\sum_{i=1}^N \mathbb{E}\left[U_i(Y,X,D,G)-\varepsilon_iY_i\vert X,D,G,\lambda\right],
        \end{split}
    \end{equation}
    we have, 
    \begin{equation}\label{appeq:uwep}
      W_{X,G,\lambda}(D)=  \frac{1}{N}\sum_{i=1}^N \alpha_i \Pr(Y_i=1\vert X,D,G,\lambda)+\frac{1}{N}\sum_{i=1}^N\sum_{j\neq i}\beta_{ij}\Pr(Y_iY_j=1\vert X,D,G,\lambda).
    \end{equation}
    Therefore, if $\alpha_i>0,$ $\alpha_i \Pr(Y_i=1\vert X,D,G,\lambda)$ achieves its lower bound by choosing the equilibrium selection rule $\underline{\lambda}$. When $\alpha_i\leq0,$ $\alpha_i \Pr(Y_i=1\vert X,D,G,\lambda)$ achieves its lower bound by choosing the equilibrium selection rule $\widebar{\lambda}$. For the second term, since $\beta_{ij}\geq 0$ for all $i,j\in\mathcal{N},$ it achieves its upper bound by choosing the equilibrium selection rule as $\widebar{\lambda}$ and it achieves its lower bound by choosing the equilibrium selection rule as $\underline{\lambda}$. This is because
    \begin{equation}
        \begin{split}
            &\quad\Pr(Y_iY_j=1\vert X,D,G,\lambda)\\&= \sum_{\sigma^*\in\Sigma}\lambda(\sigma^*\vert X,D,G)\int \mathds{1}\Big\{\alpha_i+\sum_{k\neq i}\beta_{ik}\sigma_k^*(X,D,G)\geq\varepsilon_i\Big\} \mathds{1}\Big\{\alpha_j+\sum_{k\neq j}\beta_{jk}\sigma_k^*(X,D,G)\geq\varepsilon_j\Big\} dF_{\varepsilon_i}dF_{\varepsilon_j}\\
    &=\sum_{\sigma^*\in\Sigma}\lambda(\sigma^*\vert X,D,G)\int \mathds{1}\Big\{\alpha_i+\sum_{k\neq i}\beta_{ik}\sigma_k^*(X,D,G)\geq\varepsilon_i\Big\} dF_{\varepsilon}\int\mathds{1}\Big\{\alpha_j+\sum_{k\neq j}\beta_{jk}\sigma_k^*(X,D,G)\geq\varepsilon_j\Big\} dF_{\varepsilon},
        \end{split}
    \end{equation}
    where the second equality holds by Assumption \ref{ass:epsilon}. Therefore,
    \begin{equation}
         \Pr(Y_iY_j=1\vert X,D,G,\lambda)= \sum_{\sigma^*\in\Sigma}\lambda(\sigma^*\vert X,D,G)\Pr(Y_i=1\vert X,D,G,\sigma^*)\Pr(Y_j=1\vert X,D,G,\sigma^*).
    \end{equation}
    From Theorem \ref{thm:equality}, $\{\Pr(Y_j=1\vert X,D,G,\sigma^*)\}_{i=1}^N$ achieves their upper bound under the most favorable equilibrium selection rule $\widebar{\lambda}$, where $\widebar{\sigma}^*$ happens with probability 1. Therefore, 
    \begin{equation}
        \begin{split}
            \Pr(Y_iY_j=1\vert X,D,G,\widebar{\lambda}) &= \Pr(Y_i=1\vert X,D,G,\widebar{\sigma}^*)\Pr(Y_j=1\vert X,D,G,\widebar{\sigma}^*)\\
            &\geq \Pr(Y_iY_j=1\vert X,D,G,\lambda),\quad\forall\lambda\in\Lambda.
        \end{split}
    \end{equation}
    By the symmetric argument, we have
    \begin{equation}
       \begin{split}
            \Pr(Y_iY_j=1\vert X,D,G,\underline{\lambda})&= \Pr(Y_i=1\vert X,D,G,\underline{\sigma}^*)\Pr(Y_j=1\vert X,D,G,\underline{\sigma}^*)\\
            &\leq \Pr(Y_iY_j=1\vert X,D,G,\lambda),\quad\forall\lambda\in\Lambda.
       \end{split}
    \end{equation}
    Therefore, for all $i,j\in\mathcal{N}$,
    \begin{equation}\label{appeq:lowup}
        \beta_{ij}\Pr(Y_iY_j=1\vert X,D,G,\lambda)\in[\beta_{ij}\underline{\sigma}^*_i\underline{\sigma}^*_j,\beta_{ij}\widebar{\sigma}^*_i\widebar{\sigma}^*_j].
    \end{equation}
    Plugging Eq.\ref{appeq:lowup} into Eq.\ref{appeq:uwep} completes the proof.
 \end{proof}

\subsection{Proof of Lemma \ref{lemma:unctheta}}\label{app:unctheta}
\begin{lemma}
Under Assumptions \ref{ass:epsilon} and \ref{ass:paraspace}, 
    \begin{equation}
        \mathbb{E}_{\varepsilon^{n}}\Big[\max_{D\in\mathcal{D}}\lvert W_n(D)-W(D)\rvert\Big \vert S,\sigma^{data}\Big]\leq C_1\mathbb{E}_{\varepsilon^{n}}\Big[ \Vert\hat{\theta}-\theta\Vert_1\Big \vert S,\sigma^{data}\Big]
    \end{equation}
    where $C_1$ is a constant that only depends on the distribution $F_{\varepsilon^{n}}$, the maximum link 
 in the network $\widebar{N}$, and the support of true parameter $\theta_0$.
\end{lemma}
\begin{proof}
By the Triangle inequality, an upper bound for our objective function is:
\begin{equation}\label{eq:aggwel}
   \begin{split}
        \lvert W_n(D) - W(D) \rvert &= \Big\lvert\frac{1}{N}\sum_{i=1}^N\big(\underline{\sigma}_i(\hat{\theta})-\underline{\sigma}_i(\theta)\big)\Big\rvert\leq \frac{1}{N}\sum_{i=1}^N\lvert\underline{\sigma}_i(\hat{\theta})-\underline{\sigma}_i(\theta)\rvert.
   \end{split}
\end{equation}
We know:
\begin{equation}\label{eq:pru}
   \begin{split}
        \underline{\sigma}_i(\theta) &= \int \mathbbm{1}\Big\{\alpha_i+\sum_{j\neq i}\beta_{ij}\underline{\sigma}_j(\theta) -\varepsilon_i\geq 0  \Big\}dF(\varepsilon)\\
        &=F_{\varepsilon}\Big(\theta_0+\theta_1D_i+\theta_2^{\intercal} X_i+ \theta_3^{\intercal} X_iD_i+\frac{\theta_4}{\mathcal{N}_i}\sum_{j\neq i}G_{ij}m_{ij}D_j+\frac{1}{\mathcal{N}_i}\sum_{j\neq i}(\theta_5+\theta_6 D_iD_j)G_{ij}m_{ij}\underline{\sigma}_j(\theta)\Big).
   \end{split}
\end{equation}
Let $r(i,\theta)\coloneqq \theta_0+\theta_1D_i+\theta_2^{\intercal} X_i+ \theta_3^{\intercal} X_iD_i+\frac{\theta_4}{\mathcal{N}_i}\sum_{j\neq i}G_{ij}m_{ij}D_j+\frac{1}{\mathcal{N}_i}\sum_{j\neq i}(\theta_5+\theta_6 D_iD_j)G_{ij}m_{ij}\underline{\sigma}_j(\theta)$. Therefore,
\begin{equation}\label{eq:gap}
    \begin{split}
        \big\lvert\underline{\sigma}_i(\hat{\theta})-\underline{\sigma}_i(\theta)\big\rvert=\big\lvert F_{\varepsilon}\big(r(i,\hat{\theta})\big)-F_{\varepsilon}\big(r(i,\theta)\big)\big\rvert.
        \end{split}
\end{equation}
By the Mean Value Theorem, 
\begin{equation}
    \begin{split}
        \big\lvert\underline{\sigma}_i(\hat{\theta})-\underline{\sigma}_i(\theta)\big\rvert&=\big\vert\nabla_{\theta} F_{\varepsilon}\big(r_{i}(\Tilde{\theta})\big)(\hat{\theta}-\theta)\big\vert\\
        &=\big\vert F_{\varepsilon}'\big(r(i,\Tilde{\theta})\big)\nabla_{\theta}r(i,\Tilde{\theta})(\hat{\theta}-\theta)\big\vert\\
        &= F_{\varepsilon}'\big(r(i,\Tilde{\theta})\big)\big\vert\nabla_{\theta}r(i,\Tilde{\theta})(\hat{\theta}-\theta)\big\vert\\
        &\leq\tau\big\vert\nabla_{\theta}r(i,\Tilde{\theta})(\hat{\theta}-\theta)\big\vert.
    \end{split}
\end{equation}
For some $\Tilde{\theta}\in\mathbb{R}^{d_{\theta}}$ on the segment from $\theta$ to $\hat{\theta}$. By the Cauchy–Schwarz inequality, we have:
\begin{equation}\label{appeq:muupp}
   \big\lvert\underline{\sigma}_i(\hat{\theta})-\underline{\sigma}_i(\theta)\big\rvert\leq  \tau\Vert\nabla_{\theta}r(i,\Tilde{\theta})\Vert_2\Vert\hat{\theta}-\theta\Vert_2\leq  \tau\Vert\nabla_{\theta}r(i,\Tilde{\theta})\Vert_1\Vert\hat{\theta}-\theta\Vert_1.
\end{equation}
To deal with the simultaneity within $\nabla_{\theta}r(i,\Tilde{\theta})$, define
\begin{equation}
    \nabla_{\theta}r(\theta)\coloneqq\Big[
    \begin{array}{ccc}
    \nabla_{\theta_0}r(\theta)\quad
    \cdots\quad
    \nabla_{\theta_{d_{\theta}}}r(\theta)
  \end{array}
    \Big],
\end{equation}
where
\begin{equation}
    \nabla_{\theta_k}r(\theta) = \left[
    \begin{array}{ccc}
    \nabla_{\theta_k}r(1,\theta)\\
    \vdots\\
    \nabla_{\theta_{k}}r(N,\theta)
  \end{array}
  \right],
\end{equation}
for all $k=1,...,d_{\theta}$. Then, 
\begin{equation}\label{appeq:inf1}
  \Vert\nabla_{\theta}r(i,\theta)\Vert_1\leq \Vert\nabla_{\theta}r(\theta)\Vert_{\infty},\quad\forall i\in\mathcal{N},  
\end{equation}
where recall $\Vert\nabla_{\theta}r(\theta)\Vert_{\infty}=\max_{i=1,...,N}\Vert \nabla_{\theta}r(i,\theta)\Vert_1$ is the operator norm induced by the $L_{\infty}$ norm. To bound $\Vert\nabla_{\theta}r(\theta)\Vert_{\infty}$, we define an implicit function $I:\mathbb{R}^N\times\mathbb{R}^{d_{\theta
}}\rightarrow \mathbb{R}^N$ such that 
\begin{equation}
    \begin{split}
    \quad  I(r,\theta) = \left[
  \begin{array}{lcl}
    r(1,\theta)-a_1-\displaystyle\frac{\theta_4}{\vert\mathcal{N}_1\vert}\sum_{j\neq 1}G_{1j}m_{1j}D_j-\displaystyle\frac{1}{\vert\mathcal{N}_1\vert}\sum_{j\neq 1}(\theta_5+\theta_6 D_1D_j)G_{1j}m_{1j}F_{\varepsilon}\big(r(j,\theta)\big)\\
    \vdots\\
    r(N,\theta)-a_N-\displaystyle\frac{\theta_4}{\vert\mathcal{N}_N\vert}\sum_{j\neq N}G_{Nj}m_{Nj}D_j-\frac{1}{\vert\mathcal{N}_N\vert}\displaystyle\sum_{j\neq N}(\theta_5+\theta_6 D_ND_j)G_{Nj}m_{Nj}F_{\varepsilon}\big(r(j,\theta)\big)\\
  \end{array}
\right],
    \end{split}
\end{equation}
where $a_i=\theta_0+\theta_1D_i-\theta_2^{\intercal} X_i- \theta_3^{\intercal} X_iD_i$. By the Implicit Function Theorem, $\nabla_{\theta}r(\theta)$ is given by:
\begin{equation}\label{apeq:grir}
    \begin{split}
    \nabla_{\theta}r(\theta) = - \big(\nabla_{r} I(r,\theta)\big)^{-1} \nabla_{\theta}I(r,\theta),   
    \end{split}
\end{equation}
where $\nabla_{r} I(r,\theta)$ is:
\begin{equation}
    \nabla_{r} I(r,\theta) =\left[
  \begin{array}{ccc}
    1& -\frac{1}{\vert \mathcal{N}_1\vert}F_{\varepsilon}'(r(1,\theta))(\theta_5+\theta_6D_1D_2)G_{12}m_{12}&\cdots\\
    -\frac{1}{\vert \mathcal{N}_2\vert}F_{\varepsilon}'(r(2,\theta))(\theta_5+\theta_6D_2D_1)G_{21}m_{21}&1 &\cdots\\
    \vdots &  \vdots &\ddots \\
   -\frac{1}{\vert \mathcal{N}_N\vert}F_{\varepsilon}'(r(N,\theta))(\theta_5+\theta_6D_ND_1)G_{N1}m_{N1}& -\frac{1}{\vert \mathcal{N}_N\vert}F_{\varepsilon}'(r(N,\theta))(\theta_5+\theta_6D_ND_2)G_{N2}m_{N2}&\cdots\\
  \end{array}
\right],
\end{equation}
and $\nabla_{\theta}I(r,\theta)$ is:
\footnotesize
\begin{equation}
   \begin{split}
        &\nabla_{\theta}I(r,\theta) =\\
    &\left[
  \begin{array}{ccccccc}
    -1&-D_1&-X_1^{\intercal}&-D_1X_1^{\intercal}&-\displaystyle\frac{\displaystyle\sum_{j\neq 1}G_{1j}m_{1j}D_j}{\vert\mathcal{N}_1\vert}&-\displaystyle\frac{\displaystyle\sum_{j\neq 1}F_{\varepsilon}'(r(j,\theta))G_{1j}m_{1j}}{\vert \mathcal{N}_1\vert}&-\displaystyle\frac{\displaystyle\sum_{j\neq 1}F_{\varepsilon}'(r(j,\theta))D_1D_jG_{1j}m_{1j}}{\vert \mathcal{N}_1\vert}\\
   -1&-D_2&-X_2^{\intercal}&-D_2X_2^{\intercal}&-\displaystyle\frac{\displaystyle\sum_{j\neq 2}G_{2j}m_{2j}D_j}{\vert\mathcal{N}_2\vert}&-\displaystyle\frac{\displaystyle\sum_{j\neq 2}F_{\varepsilon}'(r(j,\theta))G_{2j}m_{2j}}{\vert \mathcal{N}_2\vert}&-\displaystyle\frac{\displaystyle\sum_{j\neq 2}F_{\varepsilon}'(r(j,\theta))D_2D_jG_{2j}m_{2j}}{\vert \mathcal{N}_2\vert}\\
    \vdots & \vdots & \vdots & \vdots & \vdots & \vdots & \vdots\\
   -1&-D_N&-X_N^{\intercal}&-D_NX_N^{\intercal}&-\displaystyle\frac{\displaystyle\sum_{j\neq N}G_{Nj}m_{Nj}D_j}{\vert\mathcal{N}_N\vert}&-\displaystyle\frac{\displaystyle\sum_{j\neq N}F_{\varepsilon}'(r(j,\theta))G_{Nj}m_{Nj}}{\vert \mathcal{N}_N\vert}&-\displaystyle\frac{\displaystyle\sum_{j\neq N}F_{\varepsilon}'(r(j,\theta))D_ND_jG_{Nj}m_{Nj}}{\vert \mathcal{N}_N\vert}\\
  \end{array}
\right].
   \end{split}
\end{equation}
\normalsize
Therefore, supremum norm of Eq.\ref{apeq:grir} is bounded by
\begin{equation}\label{appeq:upr}
    \begin{split}
        \Vert\nabla_{\theta}r(\theta)\Vert_\infty &= \Vert\big(\nabla_{r} I(r,\theta)\big)^{-1} \nabla_{\theta}I(r,\theta)\Vert_\infty\leq \Vert\big(\nabla_{r} I(r,\theta)\big)^{-1}\Vert_{\infty}\Vert \nabla_{\theta}I(r,\theta)\Vert_\infty.
    \end{split}
\end{equation}
The last inequality holds because norm $\Vert A\Vert_{\infty}$ on matrix $A$ is the operator norm induced by $L_{\infty}$ norm, which is a matrix norm. This ensures that it satisfies the submultiplicativity property. Therefore, we have the inequality. Assuming $\nabla_{r} I(r,\Tilde{\theta})$ is non-singular, and by invoking Lemma \ref{aplemma:contir} — a corollary of Berge's Maximum Theorem — the norm $\Vert\big(\nabla_{r} I(r,\Tilde{\theta})\big)^{-1}\Vert_{\infty}$ is a continuous function with respect to the entries of $\nabla_{r} I(r,\Tilde{\theta})$. Given, for all $i,j\in\mathcal{N}$, $D_i$ and $G_{ij}$ are binary, $X_i$ has bounded support, $\theta$ is in a compact parameter space (Assumption \ref{ass:paraspace}), and $F'_{\varepsilon}(\cdot)\in[0,\tau]$, the Extreme Value Theorem (Lemma \ref{app:evtheorem}) guarantees the existence of a uniform maximum of $\Vert\big(\nabla_{r} I(r,\Tilde{\theta})\big)^{-1}\Vert_{\infty}$ among all the values of $X$, $D$, $G$, which only depends on the support of each variable. We denote this uniform maximum as $\zeta$.
For $\Vert \nabla_{\theta}I(r,\Tilde{\theta})\Vert_\infty$ in Eq.\ref{appeq:upr}, we know:
\begin{equation}\label{appeq:eqi}
    \begin{split}
        \Vert \nabla_{\theta}I(r,\Tilde{\theta})\Vert_\infty&=\displaystyle\max_{i\in\mathcal{N}}1+D_i+\Vert X_i\Vert_1+D_i\Vert X_i\Vert_1+\frac{1}{\vert \mathcal{N}_i\vert}\displaystyle\sum_{j\neq i}D_jG_{ij}\vert m_{ij}\vert\\
        &\quad+\frac{1}{\vert \mathcal{N}_i\vert}\displaystyle\sum_{j\neq i}F_{\varepsilon}'(r(j,\Tilde{\theta}))(1+D_iD_j)G_{ij}\vert m_{ij}\vert.
    \end{split}
\end{equation}
Furthermore, Eq.\ref{appeq:eqi} is upper bounded by $2+2\max_{i\in\mathcal{N}}\Vert X_i\Vert_1+2\widebar{m}\tau$, where $\widebar{m}\coloneqq\max_{i,j\in\mathcal{N}}\vert m_{ij}\vert$. Therefore, we have
\begin{equation}\label{appeq:rinfi}
 \Vert \nabla_{\theta}I(r,\Tilde{\theta})\Vert_\infty\leq   2+2\Vert X\Vert_\infty+\widebar{m}+2\widebar{m}\tau.
\end{equation}
Combining Eq.\ref{appeq:muupp}, Eq.\ref{appeq:inf1} and Eq.\ref{appeq:rinfi}, we have:
\begin{equation}
    \frac{1}{N}\sum_{i=1}^N\big\lvert\underline{\sigma}_i(\hat{\theta})-\underline{\sigma}_i(\theta)\big\rvert\leq\zeta\tau(2+2\Vert X\Vert_\infty+\widebar{m}+2\widebar{m}\tau)\Vert\hat{\theta}-\theta\Vert_1.
\end{equation}
To complete the proof, let $C_1= \zeta\tau(2+2\Vert X\Vert_\infty+\widebar{m}+2\widebar{m}\tau)$.

\end{proof}

\subsection{Proof of Lemma \ref{pro.samplingunc}}\label{app:prooflemmasmpling}
\begin{lemma}{(\textbf{Sampling Uncertainty})}
Under Assumption \ref{ass:epsilon}, \ref{ass:smooth}, \ref{ass:fullrank}, and \ref{ass:treatment}, we have
    \begin{equation}
        \mathbb{E}_{\varepsilon^{n}}\big[\Vert\hat{\theta}-\theta_0 \Vert_1\big\vert S,\sigma^{data}\big]\leq C_2\mathbb{E}_{\varepsilon^{n}}\big[\Vert\nabla_{\theta}\mathbb{G}_n(\Tilde{\theta} )\Vert_1\big\vert S,\sigma^{data}\big],
    \end{equation}
    where $C_2$ is a constant that depends on the distribution $F_{\varepsilon^{n}}$, and the dimension and supports of the parameter space, the covariates space $\mathcal{X}$, the network space $\mathcal{G}$ and the treatment allocation space $\mathcal{D}$.
\end{lemma}
\begin{proof}
    The Hessian matrices in the Taylor expansions of $\hat{\mathbb{M}}(\theta_0)-\hat{\mathbb{M}}(\hat{\theta})$ and $M(\hat{\theta})-M(\theta_0)$ are: 
\begin{equation}
    \begin{split}
        \nabla_{\theta}^2\hat{\mathbb{M}}(\acute{\theta})=&\frac{1}{n}\sum_{i=1}^n \big[Y_i\omega_0(\hat{Z}_i^{\intercal}\acute{\theta})-(1-Y_i)\omega_1(\hat{Z}_i^{\intercal}\acute{\theta})\big]\hat{Z}_i\hat{Z}_i^{\intercal},
    \end{split}
\end{equation}
\begin{equation}
    \begin{split}
        \nabla_{\theta}^2 M(\grave{\theta})=&\frac{1}{n}\sum_{i=1}^n\mathbb{E}_{\varepsilon^{n}} \Big[\big[Y_i\omega_0(Z_i^{\intercal}\grave{\theta})-(1-Y_i)\omega_1(Z_i^{\intercal}\grave{\theta})\big]Z_iZ_i^{\intercal}\Big\vert S,\sigma^{data}\Big],
    \end{split}
\end{equation}
where
\begin{equation}
    \begin{split}
        \omega_0(a)=\frac{ F_{\varepsilon}''(a)F_{\varepsilon}(a)-F_{\varepsilon}'(a)^2}{F_{\varepsilon}(a)^2},\quad \omega_1(a)= \frac{ F_{\varepsilon}''(a) [1-F_{\varepsilon}(a)]+F_{\varepsilon}'(a)^2}{[1-F_{\varepsilon}(a)]^2},
    \end{split}
\end{equation}
and $a\in\mathbb{R}$. 
Combining the above equations with Eq.\ref{eq:mvt} and Eq.\ref{eq:mvvt}, we have: 
\begin{equation}\label{eq:mmqu}
    \begin{split}
        \hat{\mathbb{M}}(\theta_0)-\hat{\mathbb{M}}(\hat{\theta})=\frac{1}{2n}\sum_{i=1}^n \Big[Y_i\omega_0(\hat{Z}_i^{\intercal}\acute{\theta})-(1-Y_i)\omega_1(\hat{Z}_i^{\intercal}\acute{\theta})\Big]\Big[(\hat{\theta}-\theta_0)^{\intercal}\hat{Z}_i\Big]^2\leq 0.
    \end{split}
\end{equation}
\begin{equation}\label{eq:mqu}
    M(\hat{\theta})-M(\theta_0)=\frac{1}{2n}\sum_{i=1}^n \Big[\sigma_i\omega_0(Z_i^{\intercal}\grave{\theta})-(1-\sigma_i)\omega_1(Z_i^{\intercal}\grave{\theta})\Big]\Big[(\hat{\theta}-\theta_0)^{\intercal}Z_i\Big]^2\leq 0,
\end{equation}
where $\sigma_i=\mathbb{E}_{\varepsilon^{n}}[Y_i\vert S,\sigma^{data}]$. Although Eq.\ref{eq:mmqu} and Eq.\ref{eq:mqu} are non-positive, they do not pin down the sign of the coefficient on the quadratic term (i.e., $Y_i\omega_0(\hat{Z}_i^{\intercal}\acute{\theta})-(1-Y_i)\omega_1(\hat{Z}_i^{\intercal}\acute{\theta})$ and $\sigma_i\omega_0(Z_i^{\intercal}\grave{\theta})-(1-\sigma_i)\omega_1(Z_i^{\intercal}\grave{\theta})$). Under Assumption \ref{ass:smooth}, we have $\omega_0(a) < 0$ and $\omega_1(a) > 0$ for all $a \in \mathbb{R}$. Furthermore, Assumption \ref{ass:paraspace} ensures that $\Theta$ is a compact parameter space, which guarantees that $\hat{\theta}$ resides within this compact set. Consequently, $\acute{\theta}$ and $\grave{\theta}$ are also confined within the same compact set. Given this setup, all elements in $Z_i$ and $\hat{Z}_i$ for each $i=1, \ldots, n$ also exist within a compact set. Therefore, the products $\hat{Z}_i^{\intercal}\acute{\theta}$ and $Z_i^{\intercal}\grave{\theta}$ are confined within a compact set, denoted as $\Xi$. We define the following bounds for $\omega_0$ and $\omega_1$ across $\Xi$:
\begin{equation}
    \widebar{\omega}_0\coloneqq \max_{x\in\Xi}\omega_0(x),\quad
    \widebar{\omega}_1\coloneqq \max_{x\in\Xi}-\omega_1(x),\quad\underline{\omega}_0\coloneqq \min_{x\in\Xi}\omega_0(x),\quad
    \underline{\omega}_1\coloneqq \min_{x\in\Xi}-\omega_1(x).
\end{equation}
By Assumption \ref{ass:epsilon} and \ref{ass:smooth}, the probability density function of $\varepsilon^{n}$ and its derivative are bounded. In addition, under Assumption \ref{ass:epsilon} they are continuous functions. Therefore, the Extreme Value Theorem guarantees the existence of $\widebar{\omega}_0$ and $\widebar{\omega}_1$. Define $\widebar{\omega}$ as $\max\{\widebar{\omega}_0, \widebar{\omega}_1\}$, and $\underline{\omega}$ as $\min\{ \underline{\omega}_0, \underline{\omega}_1\}$. We have $\omega_0(\hat{Z}_i^{\intercal}\acute{\theta})\leq \widebar{\omega}$ and $\omega_0(Z_i^{\intercal}\grave{\theta})\leq \widebar{\omega}$ for any $i=1,...,n$. In addition, $-\omega_1(\hat{Z}_i^{\intercal}\acute{\theta})\leq \widebar{\omega}$ and $-\omega_1(Z_i^{\intercal}\grave{\theta})\leq \widebar{\omega}$ for any $i=1,...,n$. Combining the above arguments with Eq.\ref{eq:mmqu} and Eq.\ref{eq:mqu}:
\begin{equation}\label{eq:gnuniform}
   \mathbb{G}_n(\hat{\theta})-\mathbb{G}_n(\theta_0)\geq -\widebar{\omega}(\hat{\theta}-\theta_0)^{\intercal}\frac{1}{2n}\sum_{i=1}^n(Z_iZ_i^{\intercal}+\hat{Z}_i\hat{Z}_i^{\intercal})(\hat{\theta}-\theta_0).
\end{equation}
Under Assumption \ref{ass:fullrank}, $\frac{1}{n}\sum_{i=1}^n \hat{Z}_i\hat{Z}_i^\intercal$ and $\frac{1}{n}\sum_{i=1}^n Z_iZ_i^\intercal$ are positive definite matrices for all $X\in\mathcal{X}^n$ and $G\in\mathcal{G}$. This guarantees that the smallest eigenvalues of $\frac{1}{n}\sum_{i=1}^n \hat{Z}_i\hat{Z}_i^\intercal$ and $\frac{1}{n}\sum_{i=1}^n Z_iZ_i^\intercal$ are positive. However, the smallest eigenvalue still depends on the size of the training sample. Lemma \ref{lemma:cz} addresses this by guaranteeing the existence of a uniform lower bound on the smallest eigenvalue. Lemma \ref{lemma:cz} is a corollary of \textit{Berge's Maximum Theorem} (Lemma \ref{app:bergetheorem}). A proof is provided in Appendix \ref{app:contio}. Given that any element in $Z_i$ and $\hat{Z}_i$ is a linear combination of products between $\mathsf{X}_i$, $\mathsf{G}_{ij}$, $\sigma^{data}_i$, $\mathsf{D}_i$, and $\mathsf{D}\in\{0,1\}^n$, $\sigma^{data}\in[0,1]^n$, $\mathsf{G}\in\{0,1\}^{n\times n}$, and $\mathsf{X}\in\mathcal{X}^n$ are both compact, the \textit{Extreme Value Theorem} (Lemma \ref{app:evtheorem}) guarantees the existence of a minimum smallest eigenvalue, which only depends on the bound of each element\footnote{Although the equilibrium \(\sigma\) and $\hat{\sigma}^{data}$ may not be continuous functions of \(\mathsf{X}\), \(\mathsf{D}\), and \(\mathsf{G}\), the inherent compactness of \(Z_i\) and \(\hat{Z}_i\) for all \(i = 1, \ldots, n\) suffices to apply the Extreme Value Theorem.} and is independent of the training data. 
Assumption \ref{ass:treatment} guarantees that the average number of treated units in the training data is non-zero for any network size. Combined with the above arguments, this ensures that the matrices formed by $\frac{1}{n}\sum_{i=1}^n \hat{Z}_i\hat{Z}_i^\intercal$ and $\frac{1}{n}\sum_{i=1}^n Z_iZ_i^\intercal$ have strictly positive minimum smallest eigenvalues. We denote these minimum smallest eigenvalues as $\underline{\varsigma_{min}^0}$ and $\underline{\varsigma_{min}^1}$, respectively. In addition, let $\underline{\varsigma_{min}}=\min\{\underline{\varsigma_{min}^0}, \underline{\varsigma_{min}^1}\}$. Therefore,
\begin{equation}\label{eq:etamin}
    (\hat{\theta}-\theta_0)^{\intercal}\frac{1}{2n}\sum_{i=1}^n(Z_iZ_i^{\intercal}+\hat{Z}_i\hat{Z}_i^{\intercal})(\hat{\theta}-\theta_0)\geq \underline{\varsigma_{min}}\Vert \hat{\theta}-\theta_0\Vert^2_2>  0.
\end{equation}
Combining Eq.\ref{eq:naivbound} with Eq.\ref{eq:gnuniform} and Eq.\ref{eq:etamin}, we conclude that the sampling uncertainty of $\hat{\theta}$ is characterized by the empirical process $\mathbb{G}_n(\cdot)$. Formally:
\begin{equation}
        \mathbb{E}_{\varepsilon^{n}}\big[\Vert\hat{\theta}-\theta_0 \Vert_1\big\vert S,\sigma^{data}\big]\leq \frac{d_{\theta}}{-\widebar{\omega}\underline{\varsigma_{min}}}\mathbb{E}_{\varepsilon^{n}}\big[\Vert\nabla_{\theta}\mathbb{G}_n(\Tilde{\theta} )\Vert_1\big\vert S,\sigma^{data}\big].
    \end{equation}
    To complete the proof, let $C_2 = \frac{d_{\theta}}{-\widebar{\omega}\underline{\varsigma_{min}}}$.
\end{proof}

\subsection{Proof of Lemma \ref{lemma:gaussian}}\label{app:prooflemmagaus}
\begin{lemma}
Under Assumption \ref{ass:epsilon} to \ref{ass:fullrank}, we have
    \begin{equation}
      \mathbb{E}_{\varepsilon^{n}}\big[\sup_{\theta\in\Theta}\Vert\nabla_{\theta}\mathbb{G}_n(\theta )\Vert_1\big\vert S,\sigma^{data}\big]\leq \frac{C_3+C_4\sqrt{\log(n)}}{\sqrt{n}},
    \end{equation}
    where $C_3$ and $C_4$ are constants that depend only on the support of covariates, the distribution of $\varepsilon$, $C_{\sigma}$, the covariates space $\mathcal{X}$, the network space $\mathcal{G}$ and the treatment allocation space $\mathcal{D}$.
\end{lemma}
\begin{proof}
    Recall we are using a two-step ML estimation procedure, so the first step of estimation introduces additional sampling uncertainty through $\hat{\sigma}^{data}$. To separate the sampling uncertainty of the first and the second steps, we introduce $\mathbb{M}(\theta)$, which is the likelihood function evaluated at the true 
equilibrium in the training data $\sigma^{data}$:
\begin{equation}
    \mathbb{M}(\theta) =\frac{1}{n}\sum_{i=1}^n Y_i\log\big(F_{\varepsilon}(Z_{i}^\intercal\boldsymbol{\theta})\big)+(1-Y_i)\log\big(1-F_{\varepsilon}(Z_{i}^\intercal\boldsymbol{\theta})\big).
\end{equation}
We then decompose the initial empirical process $\mathbb{G}_n(\theta)$ into two parts: 
\begin{equation}\label{eq:truem}
      \mathbb{G}_n(\theta) =  \hat{\mathbb{M}}(\theta)-\mathbb{M}(\theta)+\mathbb{M}(\theta)-M(\theta).
\end{equation}
The first term measures the uncertainty of using the estimated equilibrium $\hat{\sigma}^{data}$, and the second term measures the uncertainty of using the estimated parameter $\hat{\theta}$. Rewrite the gradient of the empirical process as
\begin{equation}
   \begin{split}
       \nabla_{\theta}\mathbb{G}_n(\Tilde{\theta})=\nabla_{\theta} \hat{\mathbb{M}}(\Tilde{\theta})-\nabla_{\theta} \mathbb{M}(\Tilde{\theta})+\nabla_{\theta} \mathbb{M}(\Tilde{\theta})-\nabla_{\theta} M(\Tilde{\theta}).
   \end{split}
\end{equation}
By the triangle inequality, we have:
\begin{equation}\label{eq:gradtrain}
    \Vert\nabla_{\theta}\mathbb{G}_n(\Tilde{\theta}')\Vert_1\leq \Vert\nabla_{\theta} \hat{\mathbb{M}}(\Tilde{\theta})-\nabla_{\theta} \mathbb{M}(\Tilde{\theta})\Vert_1+\Vert\nabla_{\theta} \mathbb{M}(\Tilde{\theta})-\nabla_{\theta} M(\Tilde{\theta})\Vert_1.
\end{equation}
The gradient of $\mathbb{G}_n(\cdot)$, $\hat{\mathbb{M}}(\cdot)$, $\mathbb{M}(\cdot)$ and $M(\cdot)$ are $d_{\theta}\times 1$ vectors. Let $\nabla_{k}\mathbb{G}_n(\cdot)$, $\nabla_{k}\hat{\mathbb{M}}(\cdot)$, $\nabla_{k}\mathbb{M}(\cdot)$ and $\nabla_{k}M(\cdot)$ denote their $k$-th elements.
In addition, we define a sequence of empirical processes $\{\mathbb{B}_k(\theta)\}_{k=1}^{d_{\theta}}$ where $\mathbb{B}_k(\theta)\coloneqq \nabla_{k}\mathbb{M}(\theta)-\nabla_{k}M(\theta)$, and a sequence of stochastic processes $\{\mathbb{A}_k(\theta)\}_{k=1}^{d_{\theta}}$ where $\mathbb{A}_k(\theta)\coloneqq \nabla_{k}\hat{\mathbb{M}}(\theta)-\nabla_{k}\mathbb{M}(\theta)$.
The first term in Eq.\ref{eq:gradtrain} is then
\begin{equation}\label{eq:epiria}
   \Vert\nabla_{\theta} \hat{\mathbb{M}}(\Tilde{\theta})-\nabla_{\theta} \mathbb{M}(\Tilde{\theta})\Vert_1= \sum_{k=1}^{d_{\theta}}  \vert \mathbb{A}_k(\Tilde{\theta})\vert,
\end{equation}
and the second term in Eq.\ref{eq:gradtrain} is
\begin{equation}\label{eq:epirib}
    \begin{split}
       \Vert\nabla_{\theta} \mathbb{M}(\Tilde{\theta})-\nabla_{\theta} M(\Tilde{\theta})\Vert_1&=\sum_{k=1}^{d_{\theta}}  \vert \mathbb{B}_k(\Tilde{\theta})\vert.
    \end{split}
\end{equation}
Combining Eq.\ref{eq:gradtrain} with Eq.\ref{eq:epiria} and Eq.\ref{eq:epirib}, we conclude
\begin{equation}\label{eq:eptwp}
    \begin{split}
        \mathbb{E}_{\varepsilon^{n}}\big[\Vert\nabla_{\theta}\mathbb{G}_n(\Tilde{\theta} )\Vert_1\big\vert S,\sigma^{data}\big]\leq\sum_{k=1}^{d_{\theta}}\mathbb{E}_{\varepsilon^n}\big[\vert \mathbb{A}_k(\Tilde{\theta})\vert\big\vert S,\sigma^{data}\big]+\sum_{k=1}^{d_{\theta}}\mathbb{E}_{\varepsilon^{n}}\big[\vert \mathbb{B}_k(\Tilde{\theta})\vert\big\vert S,\sigma^{data}\big].
    \end{split}
\end{equation}
An upper bound for the first term in Eq.\ref{eq:eptwp} is provided by Lemma \ref{lemma:a}. An upper bound for the second term in Eq.\ref{eq:eptwp} is provided by Lemma \ref{lemma:b}. Combining them,
\begin{equation}
      \mathbb{E}_{\varepsilon^{n}}\big[\sup_{\theta\in\Theta}\Vert\nabla_{\theta}\mathbb{G}_n(\theta )\Vert_1\big\vert S,\sigma^{data}\big]\leq d_{\theta}\frac{C_A\sqrt{1+\ln(2)}+1}{\sqrt{n}}+\frac{d_{\theta}^{3/2}C_{B1}}{\sqrt{n}}\sqrt{\log(1+C_{B2}\sqrt{n})}.
    \end{equation}
In addition, given $n\geq 2$,
\begin{equation}
    \log(1+C_{B2}\sqrt{n})\leq 2\log(C_{B2}\sqrt{n})=2\log(C_{B2})+\log(n).
\end{equation}
As a consequence,
\begin{equation}
   \mathbb{E}_{\varepsilon^{n}}\big[\sup_{\theta\in\Theta}\Vert\nabla_{\theta}\mathbb{G}_n(\theta )\Vert_1\big\vert S,\sigma^{data}\big]\leq d_{\theta}\frac{C_A\sqrt{1+\ln(2)}+1}{\sqrt{n}}+\frac{d_{\theta}^{3/2}C_{B1}}{\sqrt{n}}\sqrt{\log(C_{B2})}+\frac{d_{\theta}^{3/2}C_{B1}}{\sqrt{n}}\sqrt{\log(n)}.
\end{equation}
To complete the proof, let $C_3=(C_A\sqrt{1+\ln(2)}+1)d_{\theta}+d_{\theta}^{3/2}C_{B1}$, and $C_4=d_{\theta}^{3/2}C_{B1}$.
\end{proof}

\section{Theorems and Propositions}\label{appendixC}
\subsection{Proof of Theorem \ref{thm:equality}}\label{app:thequ}
\begin{theorem} For a supermodular game, the least favorable equilibrium selection rule $\underline{\lambda}$ and the most favorable equilibrium selection rule $\widebar{\lambda}$ are given as:
    \begin{equation}
        \underline{\lambda} \coloneqq \delta_{\underline{\sigma}^*}, \quad \widebar{\lambda} \coloneqq \delta_{\widebar{\sigma}^*},
    \end{equation}
    where $\delta_{\sigma}$ is the Dirac measure on $\Sigma$. 
In addition, the following conditions are satisfied:
\begin{equation}
    \begin{split}
\inf_{\lambda\in\Lambda}\sum_{i=1}^N \Pr(Y_i=1\vert X,D,G,\lambda)
      =\sum_{i=1}^N\inf_{\lambda\in\Lambda} \Pr(Y_i=1\vert X,D,G,\lambda),
    \end{split}
\end{equation}
\begin{equation}
    \begin{split}
\sup_{\lambda\in\Lambda}\sum_{i=1}^N \Pr(Y_i=1\vert X,D,G,\lambda)
      =\sum_{i=1}^N\sup_{\lambda\in\Lambda} \Pr(Y_i=1\vert X,D,G,\lambda).
    \end{split}
\end{equation}
\end{theorem}
\begin{proof}
Recall that the conditional choice probability is:
\begin{equation}
    \begin{split}
         \Pr(Y_i=1\vert X,D,G,\lambda)&=\sum_{y_{-i}\in \mathcal{Y}^{N-1}}\Pr(Y_i=1,Y_{-i}=y_{-i}\vert X,D,G,\lambda)\\ &=\sum_{y_{-i}\in \mathcal{Y}^{N-1}}\int \lambda(1,y_{-i}\vert X,D,G,\varepsilon) dF_\varepsilon\\ &=\int\sum_{y_{-i}\in \mathcal{Y}^{N-1}} \lambda(1,y_{-i}\vert X,D,G,\varepsilon) dF_\varepsilon.
    \end{split}
\end{equation}
Summing over all the units yields 
\begin{equation}
    \begin{split}
        \sum_{i=1}^N \Pr(Y_i=1\vert X,D,G,\lambda)&=\sum_{i=1}^N\sum_{y_{-i}\in \mathcal{Y}^{N-1}}\Pr(Y_i=1,Y_{-i}=y_{-i}\vert X,D,G,\lambda) \\&=\sum_{i=1}^N \sum_{y_{-i}\in \mathcal{Y}^{N-1}}\int\lambda(1,y_{-i}\vert X,D,G,\varepsilon) dF_\varepsilon\\&=\sum_{i=1}^N \int\sum_{y_{-i}\in \mathcal{Y}^{N-1}}\lambda(1,y_{-i}\vert X,D,G,\varepsilon) dF_\varepsilon
    \end{split}
\end{equation}
Given the properties of a supermodular game, there always exists a maximal pure strategy Bayesian Nash equilibrium $\overline{y}_\varepsilon$ and a minimal pure strategy Bayesian Nash equilibrium $\underline{y}_\varepsilon$ for all $\varepsilon$. Recall Eq.\ref{eq:equilisig}, for a given $\varepsilon\in\mathbb{R}^N$, these two extreme equilibria can be represented by:
\begin{equation}
    \underline{y}^i_\varepsilon=\mathds{1}\Big\{\alpha_i+\sum_{j\neq i}\beta_{ij}\underline{\sigma}_j^*(X,D,G)\geq \varepsilon_i\Big\},\quad \forall i\in\mathcal{N}.
    \end{equation}
    \begin{equation}
        \widebar{y}^i_\varepsilon=\mathds{1}\Big\{\alpha_i+\sum_{j\neq i}\beta_{ij}\widebar{\sigma}_j^*(X,D,G)\geq \varepsilon_i\Big\},\quad \forall i\in\mathcal{N}.
    \end{equation}
Therefore, $\underline{y}_{\varepsilon}$ happens with probability 1 under our defined least favorable equilibrium selection rule $\underline{\lambda}$, and $\widebar{y}_{\varepsilon}$ happens with probability 1 under our defined most favorable equilibrium selection rule $\widebar{\lambda}$.  We know that $\overline{y}_\varepsilon\geq \underline{y}_\varepsilon$ for any $\varepsilon$ where the order in here is \textit{product order}. For any Bayesian Nash equilibrium $y_\varepsilon\in\Sigma(X,D,G,\varepsilon)$, we must have $\overline{y}_\varepsilon^i\geq y_\varepsilon^i\geq\underline{y}_\varepsilon^i$ for any $i\in\mathcal{N}$, $\varepsilon\in\mathbb{R}^N$. Therefore, there are only three possible scenarios for each unit $i$:
\begin{itemize}
    \item $\overline{y}_\varepsilon^i= y_\varepsilon^i=\underline{y}_\varepsilon^i=1$.
    \item $\overline{y}_\varepsilon^i= y_\varepsilon^i=\underline{y}_\varepsilon^i=0$.
    \item $\overline{y}_\varepsilon^i=1$, $\underline{y}_\varepsilon^i=0$ and $y^i_\varepsilon\in\{0,1\}$.
\end{itemize}
Recall $\Pr(Y_i=1\vert X,D,G,\lambda,\varepsilon) = \sum_{y_{-i}\in \mathcal{Y}^{N-1}}\lambda(1,y_{-i}\vert X,D,G,\varepsilon)$, for any $\varepsilon\in\mathbb{R}^N$, we must have: 
\begin{itemize}
    \item when $\overline{y}_\varepsilon^i= \underline{y}_\varepsilon^i=1$, $\Pr(Y_i=1\vert X,D,G,\lambda,\varepsilon)=1$ for all $\lambda\in\Lambda$ and for all $i\in\mathcal{N}$;
    \item when $\overline{y}_\varepsilon^i= \underline{y}_\varepsilon^i=0$, $\Pr(Y_i=1\vert X,D,G,\lambda,\varepsilon)=0$ for all $\lambda\in\Lambda$ and for all $i\in\mathcal{N}$;
    \item when $\overline{y}_\varepsilon^i=1$ and $\underline{y}_\varepsilon^i=0$, $\Pr(Y_i=1\vert X,D,G,\underline{\lambda},\varepsilon)=0$ for all $i\in\mathcal{N}$;
     \item when $\overline{y}_\varepsilon^i=1$ and $\underline{y}_\varepsilon^i=0$, $\Pr(Y_i=1\vert X,D,G,\overline{\lambda},\varepsilon)=1$ for all $i\in\mathcal{N}$.
\end{itemize}
Therefore, for all $\varepsilon\in\mathbb{R}^N$ and $i\in\mathcal{N}$, 
\begin{equation}
    \Pr(Y_i=1\vert X,D,G,\underline{\lambda},\varepsilon)\leq \Pr(Y_i=1\vert X,D,G,\lambda,\varepsilon),\quad \forall \lambda\in\Lambda,
\end{equation}

\begin{equation}
    \Pr(Y_i=1\vert X,D,G,\overline{\lambda},\varepsilon)\geq \Pr(Y_i=1\vert X,D,G,\lambda,\varepsilon),\quad \forall \lambda\in\Lambda.
\end{equation}
As a consequence, the following two conditions are also satisfied:
\begin{equation}\label{eq:pryi}
     \Pr(Y_i=1\vert X,D,G,\underline{\lambda}) \leq \Pr(Y_i=1\vert X,D,G,\lambda), \quad \forall \lambda\in\Lambda,
\end{equation}

\begin{equation}\label{eq:pryiup}
     \Pr(Y_i=1\vert X,D,G,\overline{\lambda}) \geq \Pr(Y_i=1\vert X,D,G,\lambda), \quad \forall \lambda\in\Lambda,
\end{equation}
Therefore, we must have:
\begin{equation}
    \sum_{i=1}^N \Pr(Y_i=1\vert X,D,G,\underline{\lambda}) \leq \sum_{i=1}^N \Pr(Y_i=1\vert X,D,G,\lambda), \quad \forall \lambda\in\Lambda,
\end{equation}
\begin{equation}
    \sum_{i=1}^N \Pr(Y_i=1\vert X,D,G,\overline{\lambda}) \geq \sum_{i=1}^N \Pr(Y_i=1\vert X,D,G,\lambda), \quad \forall \lambda\in\Lambda.
\end{equation}
As a consequence,
\begin{equation}
    \underline{\lambda} = \arg\min_{\lambda\in\Lambda}\sum_{i=1}^N \Pr(Y_i=1\vert X,D,G,\lambda).
\end{equation}

\begin{equation}
    \overline{\lambda} = \arg\max_{\lambda\in\Lambda}\sum_{i=1}^N \Pr(Y_i=1\vert X,D,G,\lambda).
\end{equation}
In addition, given Eq.\ref{eq:pryi} and Eq.\ref{eq:pryiup}, we have: 
\begin{equation}
    \underline{\lambda} = \arg\min_{\lambda\in\Lambda} \Pr(Y_i=1\vert X,D,G,\lambda), \quad\forall i\in\mathcal{N}.
\end{equation}

\begin{equation}
    \overline{\lambda} = \arg\max_{\lambda\in\Lambda} \Pr(Y_i=1\vert X,D,G,\lambda), \quad\forall i\in\mathcal{N}.
\end{equation}
Therefore, 
\begin{equation}
    \begin{split}
\inf_{\lambda\in\Lambda}\sum_{i=1}^N \Pr(Y_i=1\vert X,D,G,\lambda)
      =\sum_{i=1}^N\inf_{\lambda\in\Lambda} \Pr(Y_i=1\vert X,D,G,\lambda),
    \end{split}
\end{equation}
\begin{equation}
    \begin{split}
\sup_{\lambda\in\Lambda}\sum_{i=1}^N \Pr(Y_i=1\vert X,D,G,\lambda)
      =\sum_{i=1}^N\sup_{\lambda\in\Lambda} \Pr(Y_i=1\vert X,D,G,\lambda).
    \end{split}
\end{equation}
\end{proof}
\subsection{Proof of Theorem \ref{thm:uncerntain}}\label{appthm:uncertain}

\begin{theorem}{(\textbf{Sampling Uncertainty of Regret})}
    Under Assumption \ref{ass:epsilon} to \ref{ass:treatment}, the sampling uncertainty of the two-step MLE estimator is bounded by:
    \begin{equation}
    \mathbb{E}_{\varepsilon^{n}}\big[\Vert\hat{\theta}-\theta_0 \Vert_1\big\vert S,\sigma^{data}\big]\leq C_2\frac{C_3+C_4\log(n)}{\sqrt{n}}    
    \end{equation}
    In addition, the sampling uncertainty of the empirical welfare is bounded by:
    \begin{equation}
\mathbb{E}_{\varepsilon^{n}}\Big[\max_{D\in\mathcal{D}}\lvert W_n(D)-W(D)\rvert\Big\vert S,\sigma^{data}\big]\leq C_1C_2\frac{C_3+C_4\log(n)}{\sqrt{n}}.
    \end{equation}
\end{theorem}

\begin{proof} Recall
    \begin{equation}\label{appeq:comb}
        \mathbb{E}_{\varepsilon^{n}}\big[\Vert\hat{\theta}-\theta_0 \Vert_1\big\vert S,\sigma^{data}\big]\leq C_2\mathbb{E}_{\varepsilon^{n}}\big[\Vert\nabla_{\theta}\mathbb{G}_n(\Tilde{\theta} )\Vert_1\big\vert S,\sigma^{data}\big].
    \end{equation}
    Plugging Lemma \ref{lemma:gaussian} into Eq.\ref{appeq:comb} leads to
    \begin{equation}
     \mathbb{E}_{\varepsilon^{n}}\big[\Vert\hat{\theta}-\theta_0 \Vert_1\big\vert S,\sigma^{data}\big]\leq C_2\frac{C_3+C_4\log(n)}{\sqrt{n}}.
    \end{equation}
    Combining Lemma \ref{lemma:unctheta} with Eq.\ref{appeq:comb}, we conclude:
    \begin{equation}
        \mathbb{E}_{\varepsilon^{n}}\Big[\max_{D\in\mathcal{D}}\lvert W_n(D)-W(D)\rvert\Big\vert S,\sigma^{data}\big]\leq C_1C_2\frac{C_3+C_4\log(n)}{\sqrt{n}}. 
    \end{equation}
\end{proof}

\subsection{Proof of Proposition \ref{pro:greedy}}\label{apppro:greedy}
 \begin{proposition}
    Under Assumptions \ref{ass:epsilon} and Assumptions \ref{ass:paraspace}, the curvature $\xi$ of $W_n(\mathcal{D})$ and the submodularity ratio $\gamma$ of $W_n(\mathcal{D})$ are in $(0,1)$. 
         The greedy algorithm enjoys the following approximation guarantee for the problem in Eq.\ref{eq:wn}:
\begin{equation}
    W_n(D_G)\geq \frac{1}{\xi}(1-e^{-\xi\gamma}) W_n(\Tilde{D}),
\end{equation}
where $D_G$ is the treatment assignment rule that is obtained by Algorithm \ref{algo}.
\end{proposition}
\begin{proof}

The curvature is defined as the smallest value of $\xi$ such that
\begin{equation}
W_n(R\cup \{k\})-W_n(R)\geq (1-\xi)[W_n(S\cup \{k\})-W_n(S)]\quad   \forall S\subseteq R\subseteq \mathcal{N}, \forall k\in \mathcal{N}\setminus R.
\end{equation}
As a consequence,
\begin{equation}\label{eq2}
        \xi =\max_{S\subseteq R\subset\mathcal{N},k\in\mathcal{N}\setminus R} 1-\frac{W_n(R\cup \{k\})-W_n(R)}{W_n(S\cup \{k\})-W_n(S)}.
\end{equation}
The submodularity ratio of a non-negative set function is the largest $\gamma$ such that
\begin{equation}
    \sum_{k\in R\setminus S} W_n(S\cup \{k\})-W_n(S)\geq \gamma [W_n(S\cup R)-W_n(S)], \quad\forall S,R\subseteq\mathcal{N}.
\end{equation}
As a consequence,
\begin{equation}
    \begin{split}
        \gamma&= \min_{S\neq R} \frac{\sum_{k\in R\setminus S} [W_n(S\cup \{k\})-W_n(S)]}{W_n(S\cup R)-W_n(S)}\\
    \end{split}
\end{equation}
Recall the utility specification in Eq.\ref{eq:specifi}, we denote $\hat{\theta}_0+\hat{\theta}_1 D_i+ X_i^{\intercal}\hat{\theta}_2+X_i^{\intercal}\hat{\theta}_3D_i$ as $\hat{\alpha}_{1i}$ and $\hat{\theta}_0+ X_i'\hat{\theta}_2$ as $\hat{\alpha}_{0i}$. To connect to the set function notation, we further denote $D_{\mathcal{R}}=\{D_i=1:i\in\mathcal{R}\}$. Therefore, for $i\in \mathcal{R}$, we have:
\begin{equation}
    \begin{split}       W_n^i(\mathcal{R}\cup \{k\})-W_n^i(\mathcal{R}) &= F_{\varepsilon}\Big(\hat{\alpha}_{1i}+\frac{1}{\vert\mathcal{N}_i\vert}\sum_{j\neq i}\hat{\theta}_4m_{ij}G_{ij}D_j+\frac{1}{\vert\mathcal{N}_i\vert}\sum_{j\neq i}\hat{\theta}_5m_{ij}G_{ij}\underline{\sigma}_j+\frac{1}{\vert\mathcal{N}_i\vert}\sum_{j\neq i}\hat{\theta}_6D_jm_{ij}G_{ij}\underline{\sigma}_j\Big)\\
        &-F_{\varepsilon}\Big(\hat{\alpha}_{1i}+\frac{1}{\vert\mathcal{N}_i\vert}\sum_{j\neq i}\hat{\theta}_4m_{ij}G_{ij}D_j'+\frac{1}{\vert\mathcal{N}_i\vert}\sum_{j\neq i}\hat{\theta}_5m_{ij}G_{ij}\underline{\sigma}_j'+\frac{1}{\vert\mathcal{N}_i\vert}\sum_{j\neq i}\hat{\theta}_6D_j'm_{ij}G_{ij}\underline{\sigma}_j'\Big),
    \end{split}
\end{equation}
where $\underline{\sigma}_b=\Pr(Y_b=1\vert X,G,D_{\mathcal{R}\cup \{k\}},\underline{\lambda};\hat{\theta})$ and $\underline{\sigma}_b'=\Pr(Y_b=1\vert X,G,D_{\mathcal{R}},\underline{\lambda};\hat{\theta})$ for all $b= 1,...,N$. 
For $m\in\mathcal{N}\setminus\mathcal{R}\cup\{k\}$, their empirical welfare is given as:
\begin{equation}
    \begin{split}
        W_n^m(\mathcal{R}\cup \{k\})-W_n^m(\mathcal{R}) &=F_{\varepsilon}\Big(\hat{\alpha}_{0i}+\frac{1}{\vert\mathcal{N}_i\vert}\sum_{j\neq i}\hat{\theta}_4m_{ij}G_{ij}D_j+\frac{1}{\vert\mathcal{N}_m\vert}\sum_{j\neq m}\hat{\theta}_5m_{mj}G_{mj}\underline{\sigma}_j\Big)\\
        &-F_{\varepsilon}\Big(\hat{\alpha}_{0i}+\frac{1}{\vert\mathcal{N}_i\vert}\sum_{j\neq i}\hat{\theta}_4m_{ij}G_{ij}D_j'+\frac{1}{\vert\mathcal{N}_m\vert}\sum_{j\neq m}\hat{\theta}_5m_{mj}G_{mj}\underline{\sigma}_j'\Big).
    \end{split}
\end{equation}
For the unit $k$, her empirical welfare is given as:
\begin{equation}
   \begin{split}
       &\quad W_n^k(\mathcal{R}\cup \{k\})-W_n^k(\mathcal{R}) \\&=
   F_{\varepsilon}\Big(\hat{\alpha}_{1k}+\frac{1}{\vert\mathcal{N}_k\vert}\sum_{j\neq k}\hat{\theta}_4m_{kj}G_{kj}D_j+\frac{1}{\vert\mathcal{N}_k\vert}\sum_{j\neq k}\hat{\theta}_5m_{kj}G_{kj}\underline{\sigma}_j+\frac{1}{\vert\mathcal{N}_k\vert}\sum_{j\neq k}\hat{\theta}_6D_jm_{kj}G_{kj}\underline{\sigma}_j\Big)\\&-F_{\varepsilon}\Big(\hat{\alpha}_{0k}+\frac{1}{\vert\mathcal{N}_k\vert}\sum_{j\neq k}\hat{\theta}_4m_{kj}G_{kj}D_j'+\frac{1}{\vert\mathcal{N}_k\vert}\sum_{j\neq k}\hat{\theta}_5m_{kj}G_{kj}\underline{\sigma}_j'\Big).
   \end{split}
\end{equation}
In addition, the empirical welfare increments from assigning unit $k$ treatment is given as:
\begin{equation}
    \begin{split}
        W_n(\mathcal{R}\cup \{k\})-W_n(\mathcal{R})&=\sum_{i\in\mathcal{R}}W_n^i(\mathcal{R}\cup \{k\})-W_n^i(\mathcal{R})+\sum_{m\in\mathcal{N}\setminus\mathcal{R}\cup\{k\}}W_n^m(\mathcal{R}\cup \{k\})-W_n^m(\mathcal{R})
        \\&+W_n^k(\mathcal{R}\cup \{k\})-W_n^k(\mathcal{R}).
    \end{split}
\end{equation}
Applying the Mean Value Theorem, and Assumption \ref{ass:epsilon}, $W_n(\mathcal{R}\cup \{k\})-W_n(\mathcal{R})$ is upper bounded by:
\begin{equation}
   \begin{split}
        &W_n(\mathcal{R}\cup \{k\} )-W_n(\mathcal{R})\\&
        \leq \frac{\tau}{N}\sum_{i\in\mathcal{R}}\Big(\frac{1}{\vert\mathcal{N}_i\vert}\sum_{j\neq i}\hat{\theta}_5 m_{ij}G_{ij}(\underline{\sigma}_j-\underline{\sigma}_j')+\frac{1}{\vert\mathcal{N}_i\vert}\sum_{j\neq i,k}\hat{\theta}_6 D_jm_{ij}G_{ij}(\underline{\sigma}_j-\underline{\sigma}_j')+\frac{1}{\vert\mathcal{N}_i\vert}(\hat{\theta}_4+\hat{\theta}_6\underline{\sigma}_k) m_{ik}G_{ik}\Big)\\
        &+\frac{\tau}{N}\sum_{m\in\mathcal{N}\setminus\mathcal{R}\cup\{k\}}\Big(\frac{1}{\vert\mathcal{N}_m\vert}\hat{\theta}_4 m_{mk}G_{mk}+\frac{1}{\vert\mathcal{N}_m\vert}\sum_{j\neq m}\hat{\theta}_5 m_{mj}G_{mj}(\underline{\sigma}_j-\underline{\sigma}_j')\Big)\\
        &+\frac{\tau}{N}\Big(\hat{\theta}_1+X_i^{\intercal}\hat{\theta}_3+\frac{1}{\vert\mathcal{N}_k\vert}\sum_{j\neq k}\hat{\theta}_5m_{kj}G_{kj}(\underline{\sigma}_j-\underline{\sigma}_j')+\frac{1}{\vert\mathcal{N}_k\vert}\sum_{j\neq k}\hat{\theta}_6D_jm_{kj}G_{kj}\underline{\sigma}_j \Big),
   \end{split}
\end{equation}
where $\vert \theta\vert$ denotes the element-wise absolute value of $\theta$. Since $(\underline{\sigma}_i-\underline{\sigma}_i')\in[0,1]$ and $D_i,\underline{\sigma}_i\in\{0,1\}$ for all $i\in\mathcal{N}$, and recall $\widebar{m}\coloneqq \max_{ij}\vert m_{ij}\vert$, we can further upper bound the above equation by:
\begin{equation}
     \begin{split}
        &W_n(\mathcal{R}\cup \{k\} )-W_n(\mathcal{R})\\&
        \leq \frac{\tau}{N}\sum_{i\in\mathcal{R}}\Big(\frac{1}{\vert\mathcal{N}_i\vert}\sum_{j\neq i}\hat{\theta}_5\widebar{m}G_{ij}+\frac{1}{\vert\mathcal{N}_i\vert}\sum_{j\neq i}\hat{\theta}_6 \widebar{m}G_{ij}+\frac{\hat{\theta}_4\widebar{m}G_{ik}}{\underline{N}}\Big)\\
        &+\frac{\tau}{N}\sum_{m\in\mathcal{N}\setminus\mathcal{R}\cup\{k\}}\Big(\frac{\hat{\theta}_4\widebar{m}G_{mk}}{\underline{N}}+\frac{1}{\vert\mathcal{N}_m\vert}\sum_{j\neq m}\hat{\theta}_5 \widebar{m}G_{mj}\Big)\\
        &+\frac{\tau}{N}\Big(\hat{\theta}_1+X_i^{\intercal}\hat{\theta}_3+\frac{1}{\vert\mathcal{N}_k\vert}\sum_{j\neq k}\hat{\theta}_5\widebar{m}G_{kj}+\frac{1}{\vert\mathcal{N}_k\vert}\sum_{j\neq k}\hat{\theta}_6 \widebar{m}G_{kj} \Big).
   \end{split}
\end{equation}
Given $\frac{1}{\vert\mathcal{N}_i\vert}\sum_{j\neq i}G_{ij}=1$ for all $i,j\in\mathcal{N}$, we can further upper bound the empirical welfare increase by:
\begin{equation}
    \begin{split}
        &\quad W_n(\mathcal{R}\cup \{k\} )-W_n(\mathcal{R})\\&
        \leq \frac{\tau}{N}\sum_{i\in\mathcal{R}}(\hat{\theta}_5 +\hat{\theta}_6 )\widebar{m}+\frac{\tau}{N}\sum_{m\in\mathcal{N}\setminus\mathcal{R}\cup\{k\}}\hat{\theta}_5 \widebar{m}+\frac{\tau}{N}\Big(\hat{\theta}_1+X_i^{\intercal}\hat{\theta}_3+(\hat{\theta}_5 +\hat{\theta}_6 )\widebar{m} \Big)+\frac{\tau}{N}\frac{\hat{\theta}_4\widebar{m}\widebar{N}}{\underline{N}}.
   \end{split}
\end{equation}
Summarizing all the units together, we have:
\begin{equation}
  W_n(\mathcal{R}\cup \{k\} )-W_n(\mathcal{R})\leq \tau(\hat{\theta}_5 +\hat{\theta}_6 )\widebar{m} +\frac{\tau}{N}(\hat{\theta}_1+\widebar{X}^{\intercal}\hat{\theta}_3+\frac{\hat{\theta}_4\widebar{m}\widebar{N}}{\underline{N}}), 
\end{equation}
where $\widebar{X}=\{X_i:\max_{X_i\in\mathcal{X}}\widebar{X}^{\intercal}\hat{\theta}_3\}$. The lower bound of $W(\mathcal{R}\cup \{k\})-W(\mathcal{R})$ is:
\begin{equation}\label{appeq:loweff}
    \begin{split}
        &\quad W_n(\mathcal{R}\cup \{k\} )-W_n(\mathcal{R})\\&
        \geq \frac{\underline{F}_{\varepsilon}}{N}\sum_{i\in\mathcal{R}}\Big(\frac{1}{\vert\mathcal{N}_i\vert}\sum_{j\neq i}\hat{\theta}_5 m_{ij}G_{ij}(\underline{\sigma}_j-\underline{\sigma}_j')+\frac{1}{\vert\mathcal{N}_i\vert}\sum_{j\neq i,k}\hat{\theta}_6 D_jm_{ij}G_{ij}(\underline{\sigma}_j-\underline{\sigma}_j')+\frac{1}{\vert\mathcal{N}_i\vert}(\hat{\theta}_4+\hat{\theta}_6\underline{\sigma}_k) m_{ik}G_{ik}\Big)\\
        &+\frac{\underline{F}_{\varepsilon}}{N}\sum_{m\in\mathcal{N}\setminus\mathcal{R}\cup\{k\}}\Big(\frac{1}{\vert\mathcal{N}_m\vert}\hat{\theta}_4 m_{mk}G_{mk}+\frac{1}{\vert\mathcal{N}_m\vert}\sum_{j\neq m}\hat{\theta}_5 m_{mj}G_{mj}(\underline{\sigma}_j-\underline{\sigma}_j')\Big)\\
        &+\frac{\underline{F}_{\varepsilon}}{N}\Big(\hat{\theta}_1+X_i^{\intercal}\hat{\theta}_3+\frac{1}{\vert\mathcal{N}_k\vert}\sum_{j\neq k}\hat{\theta}_5m_{kj}G_{kj}(\underline{\sigma}_j-\underline{\sigma}_j')+\frac{1}{\vert\mathcal{N}_k\vert}\sum_{j\neq k}\hat{\theta}_6D_jm_{kj}G_{kj}\underline{\sigma}_j  \Big).
     \end{split}
\end{equation}
There are three different effects of assigning treatment to unit $k$. The first effect is the direct treatment effect on unit $k$, which is the third term in Eq.\ref{appeq:loweff}:
\begin{equation}
    \underline{\sigma}_k-\underline{\sigma}_k'\geq \underline{F}_{\varepsilon}\Big(\hat{\theta}_1+X_i^{\intercal}\hat{\theta}_3+\frac{1}{\vert\mathcal{N}_k\vert}\sum_{j\neq k}\hat{\theta}_5m_{kj}G_{kj}(\underline{\sigma}_j-\underline{\sigma}_j')+\frac{1}{\vert\mathcal{N}_k\vert}\sum_{j\neq k}\hat{\theta}_6D_jm_{kj}G_{kj}\underline{\sigma}_j \Big).
\end{equation}
Given $\underline{\sigma}_j-\underline{\sigma}_j'\geq 0$ and $\underline{\sigma}_j\geq 0$ for all $j\in\mathcal{N}$, we further bounds the direct effect from below by:
\begin{equation}\label{appeq:sigk}
    \underline{\sigma}_k-\underline{\sigma}_k'\geq\underline{F}_{\varepsilon}(\hat{\theta}_1+X_k^{\intercal}\hat{\theta}_3).
\end{equation}
For the units in the treated and untreated groups, the indirect treatment effects manifest differently. Specifically, for units $i$ in the treated group, their indirect treatment effects are given by the first term in Eq.\ref{appeq:loweff}:
\begin{equation}\label{appeq:first188}
    \underline{\sigma}_i-\underline{\sigma}_i'\geq\underline{F}_{\varepsilon}\Big(\frac{1}{\vert\mathcal{N}_i\vert}\sum_{j\neq i}\hat{\theta}_5 m_{ij}G_{ij}(\underline{\sigma}_j-\underline{\sigma}_j')+\frac{1}{\vert\mathcal{N}_i\vert}\sum_{j\neq i,k}\hat{\theta}_6 D_jm_{ij}G_{ij}(\underline{\sigma}_j-\underline{\sigma}_j')+\frac{1}{\vert\mathcal{N}_i\vert}(\hat{\theta}_4+\hat{\theta}_6\underline{\sigma}_k) m_{ik}G_{ik}\Big).
\end{equation}
We then further bound the indirect effects in Eq.\ref{appeq:first188} by:
\begin{equation}\label{appeq:sigi}
     \underline{\sigma}_i-\underline{\sigma}_i'\geq\underline{F}_{\varepsilon}\frac{1}{\vert\mathcal{N}_i\vert}\hat{\theta}_4 m_{ik}G_{ik}\geq 
     \frac{\underline{F}_{\varepsilon}}{\widebar{N}}\hat{\theta}_4 m_{ik}G_{ik}.
\end{equation}
For units $m$ not in the treated group, the indirect treatment effects, which is given by the second term of Eq.\ref{appeq:loweff}, can be quantified as follows: 
\begin{equation}
    \underline{\sigma}_m-\underline{\sigma}_m'\geq\underline{F}_{\varepsilon}\Big(\frac{1}{\vert\mathcal{N}_m\vert}\hat{\theta}_4 m_{mk}G_{mk}+\frac{1}{\vert\mathcal{N}_m\vert}\sum_{j\neq m}\hat{\theta}_5 m_{mj}G_{mj}(\underline{\sigma}_j-\underline{\sigma}_j')\Big).
\end{equation}
This is further bounded below by:
\begin{equation}\label{appeq:sigm}
    \underline{\sigma}_m-\underline{\sigma}_m'\geq\frac{\underline{F}_{\varepsilon}}{\widebar{N}}\hat{\theta}_4 m_{mk}G_{mk}. 
\end{equation}
Combining Eq.\ref{appeq:sigk}, Eq.\ref{appeq:sigi} and Eq.\ref{appeq:sigm} with Eq.\ref{appeq:loweff} leads to
\begin{equation}
  W_n(\mathcal{R}\cup \{k\} )-W_n(\mathcal{R})\geq \frac{\underline{F}_{\varepsilon}}{N\widebar{N}}\sum_{i\in\mathcal{N}\setminus\{k\}}\Big(\hat{\theta}_4 m_{ik}G_{ik}\Big)+\frac{\underline{F}_{\varepsilon}}{N}(\hat{\theta}_1+X_i^{\intercal}\hat{\theta}_3).  
\end{equation}
Given that unit $k$ has at least $\underline{N}$ neighbors, we have:
\begin{equation}
   W_n(\mathcal{R}\cup \{k\} )-W_n(\mathcal{R})\geq\frac{\underline{F}_{\varepsilon} \underline{m}\hat{\theta}_4\underline{N}}{N\widebar{N}}+\frac{\underline{F}_{\varepsilon}}{N}(\hat{\theta}_1+X_i^{\intercal}\hat{\theta}_3).
\end{equation}
Therefore, 
\begin{equation}
    \frac{W_n(R\cup \{k\})-W_n(R)}{W_n(S\cup \{k\})-W_n(S)}\geq \frac{\frac{\underline{F}_{\varepsilon} \underline{m}\hat{\theta}_4\underline{N}}{N\widebar{N}}+\frac{\underline{F}_{\varepsilon}}{N}(\hat{\theta}_1+X_i^{\intercal}\hat{\theta}_3)}{\tau(\hat{\theta}_5 +\hat{\theta}_6 )\widebar{m} +\frac{\tau}{N}(\hat{\theta}_1+\widebar{X}^{\intercal}\hat{\theta}_3+\frac{\hat{\theta}_4\widebar{m}\widebar{N}}{\underline{N}})},
\end{equation}
which ranges between (0,1). As a consequence, the submodularity ratio 
\begin{equation}
    \xi = \max_{S\subseteq R\subset\mathcal{N},k\in\mathcal{N}\setminus R} 1-\frac{W_n(R\cup \{k\})-W_n(R)}{W_n(S\cup \{k\})-W_n(S)}\in(0,1).
\end{equation}
In addition, the curvature
\begin{equation}
    \gamma = \min_{S\neq R} \frac{\sum_{k\in R\setminus S} [W_n(S\cup \{k\})-W_n(S)]}{W_n(S\cup R)-W_n(S)}\in(0,1).
\end{equation}
Combining with \citet[Theorem 1]{bian2017guarantees}, we finish the proof.
\end{proof}

\section{Preliminary Lemmas}
\subsection{Lemma \ref{aplemma:contir}}

\begin{lemma}\label{aplemma:contir}
The $\Vert \nabla_{\sigma} I(\sigma,\theta)\big)^{-1}\Vert_{\infty}$ is a continuous function with respect to any entries of $\nabla_{\sigma} I(\sigma,\theta)$.
\end{lemma}

\begin{proof}
    Let $A$ denote $ \nabla_{\sigma} I(\sigma,\theta)$ and $A^{-1}$ denote $ \big(\nabla_{\sigma} I(\sigma,\theta)\big)^{-1}$. By the definition of the Uniform norm,
    \begin{equation}
        \Vert \nabla_{\sigma} I(\sigma,\theta)\big)^{-1}\Vert_{\infty}=\displaystyle\max_{i\in\mathcal{N}}\sum_{j=1}^N\vert A_{ij}^{-1}\vert_1.
    \end{equation}
    To prove this maximum is a continuous function, two conditions of Berge's Maximum Theorem must be satisfied:
    \begin{itemize}
        \item \textbf{Continuous Function}: Let $f_i(A) = \sum_{j=1}^N \vert A_{ij}^{-1}\vert$. Then, $f_i(A) =\Vert A^{-1}_{i,:}\Vert$, where $A_{i,:}$ denotes the $i$-th row of matrix $A$. $f_i(A)$ is continuous with respect to the entries of $A$ by matrix inverse operation and the continuity of the $\ell_1$ norm calculation with respect to the vector entries.
        \item \textbf{Compact Parameter Space}: The set over which the maximum is taken (the set of row indices $i$) is trivially compact as it is finite.
    \end{itemize}
    Therefore, by Berge's Maximum Theorem, $\Vert \nabla_{\sigma} I(\sigma,\theta)\big)^{-1}\Vert_{\infty}$ is a continuous function with respect to any entry of $\nabla_{\sigma} I(\sigma,\theta)$.
\end{proof}

\subsection{Lemma \ref{lemma:cz}}\label{app:contio}
\begin{lemma}\label{lemma:cz}
    Under Assumption \ref{ass:epsilon}, the smallest and largest eigenvalues of $\frac{1}{n}\sum_{i=1}^n \hat{Z}_i\hat{Z}_i^\intercal$ and $\frac{1}{n}\sum_{i=1}^n Z_iZ_i^\intercal$ are continuous functions of any element of $Z_i$ and $\hat{Z}_i$, for any $i=1,...,n$.
\end{lemma}
\begin{proof}
   Denote the matrix $\frac{1}{N}\sum_{i=1}^N Z_iZ_i^\intercal$ as $B$. $B$ is a symmetric matrix.  Therefore, its smallest ($\eta_{min}$) and largest ($\eta_{max}$) eigenvalues are given by
    \begin{equation}\label{eq:lmin}
        \eta_{min}(B)=\min_{\theta\neq 0\in\mathbb{R}^{d_\theta}}\frac{\theta^\intercal B\theta}{\theta^\intercal\theta},
    \end{equation}
    \begin{equation}\label{eq:lmax}
        \eta_{max}(B)=\max_{\theta\neq 0\in\mathbb{R}^{d_\theta}}\frac{\theta^\intercal B\theta}{\theta^\intercal\theta}.
    \end{equation}
    Since we are studying the continuous property of $\eta_{min}$ ($\eta_{max}$) to $Z$ given a $\theta$, we restrict $\theta$ to be a unit vector (i.e., $\Vert \theta\Vert_2=1$) without lose of generality. Rewrite above equations as:
    \begin{equation}
        \eta_{min}(B)=\min_{\theta:\Vert \theta\Vert_2=1}\theta^\intercal B\theta.
    \end{equation}
    \begin{equation}
        \eta_{max}(B)=\max_{\theta:\Vert \theta\Vert_2=1}\theta^\intercal B\theta.
    \end{equation}
    Now, we apply the Berge Maximum Theorem to Eq.\ref{eq:lmin} and Eq.\ref{eq:lmax} to study the continuous property. Two conditions of the Berge Maximum Theorem must be satisfied:
    \begin{itemize}
        \item \textbf{Continuous Function}: Given $\theta^{\intercal}B\theta$ is a quadratic function, and $\Vert \theta\Vert_2=1$ (i.e., $\theta\neq 0$), it must be a continuous function w.r.t. $\theta$ and $Z_i$ for all $i\in\mathcal{N}$.
        \item \textbf{Compact Parameter Space}: Given $\Vert \theta\Vert_2=1$, the parameter space is compact.
    \end{itemize}
    Therefore, by Berge's Maximum Theorem, the largest and smallest eigenvalues are continuous functions of any element of $Z_i$. By employing a symmetric argument to $\frac{1}{n}\sum_{i=1}^n \hat{Z}_i\hat{Z}_i^\intercal$, we finish the proof.
\end{proof}

\subsection{Lemma \ref{lemma:a}}
\begin{lemma}\label{lemma:a} Under Assumptions \ref{ass:epsilon}, \ref{ass:ccp}, and \ref{ass:smooth}:
    \begin{equation}
        \mathbb{E}_{\varepsilon^{n}}\big[\sup_{\theta\in\Theta}\vert \mathbb{A}_k(\theta)\vert\big\vert S,\sigma^{data}\big]\leq C_A\sqrt{\frac{1+\ln(2)}{n}},
    \end{equation}
    where $C_A$ is a constant that depends only on the support of covariates, the distribution of $\varepsilon$ and $C_{\sigma}$.
\end{lemma}
\begin{proof}
    Recall $
        \mathbb{A}_k(\theta)= \nabla_{k}\hat{\mathbb{M}}(\theta)-\nabla_{k}\mathbb{M}(\theta)
    $, and 
    \begin{equation}
        \nabla_{k}\hat{\mathbb{M}}(\theta) = \frac{1}{n}\sum_{i=1}^n \Big[Y_i\frac{F_{\varepsilon}'(\hat{Z}_i^{\intercal}\theta)}{F_{\varepsilon}(\hat{Z}_i^{\intercal}\theta)}-(1-Y_i)\frac{F_{\varepsilon}'(\hat{Z}_i^{\intercal}\theta)}{1-F_{\varepsilon}(\hat{Z}_i^{\intercal}\theta)}\Big]Z_{ik},
    \end{equation}
    \begin{equation}
        \nabla_{k}\mathbb{M}(\theta)=\frac{1}{n}\sum_{i=1}^n \Big[Y_i\frac{F_{\varepsilon}'(Z_i^{\intercal}\theta)}{F_{\varepsilon}(Z_i^{\intercal}\theta)}-(1-Y_i)\frac{F_{\varepsilon}'(Z_i^{\intercal}\theta)}{1-F_{\varepsilon}(Z_i^{\intercal}\theta)}\Big]Z_{ik}.
    \end{equation}
    Let $\Tilde{\theta}\coloneqq\arg\sup_{\theta\in\Theta}\vert \mathbb{A}_k(\theta)\vert$. Therefore,
    \begin{equation}
      \mathbb{A}_k(\Tilde{\theta}) =   \frac{1}{n}\sum_{i=1}^n \Big[Y_i\Big(\frac{F_{\varepsilon}'(\hat{Z}_i^{\intercal}\Tilde{\theta})}{F_{\varepsilon}(\hat{Z}_i^{\intercal}\Tilde{\theta})}-\frac{F_{\varepsilon}'(Z_i^{\intercal}\Tilde{\theta})}{F_{\varepsilon}(Z_i^{\intercal}\Tilde{\theta})}\Big)-(1-Y_i)\Big(\frac{F_{\varepsilon}'(\hat{Z}_i^{\intercal}\Tilde{\theta})}{1-F_{\varepsilon}(\hat{Z}_i^{\intercal}\Tilde{\theta})}-\frac{F_{\varepsilon}'(Z_i^{\intercal}\Tilde{\theta})}{1-F_{\varepsilon}(Z_i^{\intercal}\Tilde{\theta})}\Big)\Big]Z_{ik}.
    \end{equation}
    Applying the Mean value theorem, 
    \begin{equation}
        \frac{F_{\varepsilon}'(\hat{Z}_i^{\intercal}\Tilde{\theta})}{F_{\varepsilon}(\hat{Z}_i^{\intercal}\Tilde{\theta})}-\frac{F_{\varepsilon}'(Z_i^{\intercal}\Tilde{\theta})}{F_{\varepsilon}(Z_i^{\intercal}\Tilde{\theta})}=(\hat{Z}_i-Z_i)^{\intercal}\nabla_{Z}\frac{F_{\varepsilon}'(\Tilde{Z}_i^{\intercal}\Tilde{\theta})}{F_{\varepsilon}(\Tilde{Z}_i^{\intercal}\Tilde{\theta})},
    \end{equation}
    and
    \begin{equation}
        \frac{F_{\varepsilon}'(\hat{Z}_i^{\intercal}\Tilde{\theta})}{1-F_{\varepsilon}(\hat{Z}_i^{\intercal}\Tilde{\theta})}-\frac{F_{\varepsilon}'(Z_i^{\intercal}\Tilde{\theta})}{1-F_{\varepsilon}(Z_i^{\intercal}\Tilde{\theta})}=(\hat{Z}_i-Z_i)^{\intercal}\nabla_{Z}\frac{F_{\varepsilon}'(\Tilde{Z}_i^{\intercal}\Tilde{\theta})}{1-F_{\varepsilon}(\Tilde{Z}_i^{\intercal}\Tilde{\theta})},
    \end{equation}
for some $\Tilde{Z}_i\in\mathbb{R}^{d_{\theta}}$ on the segment from $\hat{Z}_i$ to $Z_i$. Therefore,
\begin{equation}
    \begin{split}
        \vert\mathbb{A}_k(\Tilde{\theta})\vert &= \frac{1}{n}\sum_{i=1}^n \vert Y_i\omega_0(\Tilde{Z}_i^{\intercal}\Tilde{\theta})(\hat{Z}_i-Z_i)^{\intercal}\Tilde{Z}_i-(1-Y_i)\omega_1(\Tilde{Z}_i^{\intercal}\Tilde{\theta})(\hat{Z}_i-Z_i)^{\intercal}\Tilde{Z}_i\vert Z_{ik}\\
        &\leq \frac{1}{n}\sum_{i=1}^n \Big[Y_i\vert\omega_0(\Tilde{Z}_i^{\intercal}\Tilde{\theta})\Vert(\hat{Z}_i-Z_i)^{\intercal}\Tilde{Z}_i\vert+(1-Y_i)\vert\omega_1(\Tilde{Z}_i^{\intercal}\Tilde{\theta})\Vert(\hat{Z}_i-Z_i)^{\intercal}\Tilde{Z}_i\vert\Big]Z_{ik}\\
        &\leq\frac{1}{n}\sum_{i=1}^n \vert\underline{\omega}\Vert(\hat{Z}_i-Z_i)^{\intercal}\Tilde{Z}_i\vert Z_{ik}.
    \end{split}
\end{equation}
By the Cauchy–Schwarz inequality, we have:
\begin{equation}\label{appeq:cwa}
    \begin{split}
        \vert\mathbb{A}_k(\Tilde{\theta})\vert&\leq \frac{1}{n}\sum_{i=1}^n \vert\underline{\omega}\Vert\hat{Z}_i-Z_i\Vert_1 \Vert\Tilde{Z}_i\Vert_{\infty} Z_{ik}\\
        &\leq \vert\underline{\omega}\vert\bar{z}^2\frac{1}{n}\sum_{i=1}^n \Vert\hat{Z}_i-Z_i\Vert_1,
    \end{split}
\end{equation}
where $\widebar{z}\coloneqq\max_{i=1,...,n}\Vert Z_i\Vert_{\infty}$. Recall the definition of $\hat{Z}_i$ from Eq.\ref{eq:Zdefine}:
\begin{equation}
    \hat{Z}_i = \Big(1,\mathsf{D}_i, \mathsf{X}_i^\intercal, \mathsf{X}_i^\intercal \mathsf{D}_i,\frac{1}{\vert \mathcal{N}_i\vert}\sum_{j\neq i}m_{ij}\mathsf{G}_{ij}\mathsf{D}_j,\frac{1}{\vert \mathcal{N}_i\vert}\sum_{j\neq i}m_{ij}\mathsf{G}_{ij}\hat{\sigma}^{data}_j,\frac{1}{\vert \mathcal{N}_i\vert}\sum_{j\neq i}m_{ij}\mathsf{G}_{ij}\hat{\sigma}^{data}_j\mathsf{D}_i\mathsf{D}_j\Big)^\intercal.
\end{equation}
Therefore, we rewrite $\Vert\hat{Z}_i-Z_i\Vert_1$ as:
\begin{equation}
    \begin{split}
        \Vert\hat{Z}_i-Z_i\Vert_1 &= \Big\vert\frac{1}{\vert \mathcal{N}_i\vert}\sum_{j\neq i}m_{ij}\mathsf{G}_{ij}(\hat{\sigma}^{data}_j-\sigma^{data}_j)\Big\vert+\Big\vert\frac{1}{\vert \mathcal{N}_i\vert}\sum_{j\neq i}m_{ij}\mathsf{G}_{ij}(\hat{\sigma}^{data}_j-\sigma^{data}_j)\mathsf{D}_i\mathsf{D}_j\Big\vert.
    \end{split}
\end{equation}
By triangle inequality,
\begin{equation}
    \Vert\hat{Z}_i-Z_i\Vert_1\leq \frac{1}{\vert \mathcal{N}_i\vert}\sum_{j\neq i}\vert m_{ij}\vert \mathsf{G}_{ij}(1+\mathsf{D}_i\mathsf{D}_j)\vert\hat{\sigma}^{data}_j-\sigma^{data}_j\vert.
\end{equation}
Applying Lemma \ref{applemma:sigma}, and defining $\widebar{m}\coloneqq \max_{i,j\in\mathcal{N}} \vert m_{ij}\vert$, we have
\begin{equation}\label{appeq:zhatdif}
    \mathbb{E}_{\varepsilon^{n}}[\Vert\hat{Z}_i-Z_i\Vert_1\vert S,\sigma^{data}]\leq 3\widebar{m}C_{\sigma}\sqrt{\frac{1+\ln(2)}{n}}.
\end{equation}
Plug Eq.\ref{appeq:zhatdif} into Eq.\ref{appeq:cwa},
\begin{equation}
    \mathbb{E}_{\varepsilon^{n}}\big[\sup_{\theta\in\Theta}\vert \mathbb{A}_k(\theta)\vert\big\vert S,\sigma^{data}\big]\leq 3\vert\underline{\omega}\vert\bar{z}^2\widebar{m}C_{\sigma}\sqrt{\frac{1+\ln(2)}{n}}.
\end{equation}
Setting $C_A = 3\vert\underline{\omega}\vert\bar{z}^2\widebar{m}C_{\sigma}$, completes the proof.
\end{proof}

\subsection{Lemma \ref{lemma:b}}
\begin{lemma}\label{lemma:b}
   Define $\widebar{z}\coloneqq\max_{i=1,..,n}\Vert Z_i\Vert_{\infty}$. Under Assumption \ref{ass:epsilon} to \ref{ass:fullrank}, we have:
    \begin{equation}
    \mathbb{E}_{\varepsilon^n}\Big[\sup_{\theta\in\Theta} \vert \mathbb{B}_k(\theta)\vert\Big\vert S,\sigma^{data}\Big]\leq \frac{1}{\sqrt{n}}\Big(1+2C_4\sqrt{d_{\theta}\log(1+C_5\sqrt{n})}\Big),
    \end{equation}
    where $C_{B1}$ is a universal constant that only depend on the distribution $F_{\varepsilon^{n}}$ and $C_{B2}=4\vert\underline{\omega}\vert\widebar{z}^2$.
\end{lemma}
\begin{proof}
For a given $\delta\geq 0$ and associated covering number $H=N_c(\delta,\Theta,L_1)$, let $\mathbb{U}\coloneqq\{\theta^1,...,\theta^{H}\}$ be a $\delta$-cover of $\Theta$. For any $\theta\in\Theta$, we can find some $\theta^\ell$ such that $\Vert \theta-\theta^\ell \Vert_1\leq \delta$. Let $\Tilde{\theta}\coloneqq\arg\sup_{\theta\in\Theta} \vert \mathbb{B}_k(\theta)\vert$. Therefore,
\begin{equation}\label{eq:coverH}
    \begin{split}
        \vert \mathbb{B}_k(\Tilde{\theta})\vert&=\vert \mathbb{B}_k(\Tilde{\theta})-\mathbb{B}_k(\theta^\ell)+\mathbb{B}_k(\theta^\ell)\vert\\
        &\leq \vert \mathbb{B}_k(\Tilde{\theta})-\mathbb{B}_k(\theta^\ell)\vert+\vert \mathbb{B}_k(\theta^\ell)\vert\\
        &\leq \sup_{\substack{\gamma,\gamma'\in\Theta\\ \Vert\gamma-\gamma'\Vert_1\leq \delta}}\vert \mathbb{B}_k(\gamma)-\mathbb{B}_k(\gamma')\vert+\max_{\ell=1,...,H}\vert \mathbb{B}_k(\theta^{\ell})\vert.
    \end{split}
\end{equation}
Apply Lemma \ref{applema:lps} to bound the first term in Eq.\ref{eq:coverH}:
    \begin{equation}\label{appeqq:supfirst}
        \mathbb{E}_{\varepsilon^{n}}\Big[\sup_{\substack{\gamma,\gamma'\in\Theta\\ \Vert\gamma-\gamma'\Vert_1\leq \delta}}\vert \mathbb{B}_k(\gamma)-\mathbb{B}_k(\gamma')\vert\Big\vert S,\sigma^{data}\Big]\leq \vert\underline{\omega}\vert\bar{z}^2\delta,
    \end{equation}
    where $\widebar{z}\coloneqq\max_{i=1,..,n}\Vert Z_i\Vert_{\infty}$. To bound the second term in Eq.\ref{eq:coverH}, we introduce $\{\Tilde{\varepsilon}_i\}^{n}_{i=1}$, an independent copy of $\varepsilon^{n}$ that follows the same distribution $F_{\varepsilon^{n}}$. Hence, the associated $\{\Tilde{Y}_i\}_{i=1}^n$ (i.e., $\Tilde{Y}_i=\mathbbm{1}\{\alpha_i+\sum_{j\neq i}\beta_{ij}\sigma^{data}_j-\Tilde{\varepsilon}_i\geq 0\}$) has the same distribution as $\{Y_i\}_{i=1}^N$ conditional on the $S$ and $\sigma^{data}$. We denote the expectation with respect to $\Tilde{\varepsilon}$ as $\mathbb{E}_{\Tilde{\varepsilon}}(\cdot)$. Recall that the criterion function is:
\begin{equation}
    m_{Y_i,Z_i}(\theta)\coloneqq Y_i\log(F_{\varepsilon}(Z_i^\intercal\theta))+(1-Y_i)\log(1-F_{\varepsilon}(Z_i^\intercal\theta)).
\end{equation}
Denote the empirical measure of our criterion function with $\{\Tilde{Y}_i\}_{i=1}^N$ as $\Tilde{\mathbb{M}}(\theta)$:
\begin{equation}
    \Tilde{\mathbb{M}}(\theta)\coloneqq \frac{1}{n}\sum_{i=1}^n m_{\Tilde{Y}_i,Z_i}(\theta).
\end{equation}
By definition of $\Tilde{Y}_i$, we have $M(\theta)\coloneqq\mathbb{E}_{\varepsilon^{n}}[\mathbb{M}(\theta)\vert S,\sigma^{data}]=\mathbb{E}_{\Tilde{\varepsilon}}[\Tilde{\mathbb{M}}(\theta)\vert S,\sigma^{data}]$. Therefore,
\begin{equation}
    \begin{split}
       &\quad\mathbb{E}_{\varepsilon^{n}}\Big[\max_{\ell=1,...,H}\vert \mathbb{B}_k(\theta^{\ell})\vert\Big\vert  S,\sigma^{data}\Big]\\
       &=\mathbb{E}_{\varepsilon^{n}}\Big[\max_{\ell=1,...,H}\Big\vert\nabla_{k} \mathbb{M}(\theta^\ell)-\nabla_{k}\mathbb{E}_{\Tilde{\varepsilon}}[\Tilde{\mathbb{M}}(\theta^\ell) \vert S,\sigma^{data}]\Big\vert  S,\sigma^{data}\Big]\\
       &=\mathbb{E}_{\varepsilon^{n}}\big[\max_{\ell=1,...,H}\big\vert \frac{1}{n}\sum_{i=1}^n \big[\nabla_{k}m_{Y_i,Z_i}(\theta^\ell)-\nabla_{k}\mathbb{E}_{\Tilde{\varepsilon}}[m_{\Tilde{Y}_i,Z_i}(\theta^\ell)\big\vert S,\sigma^{data}]\big]\big\vert \Big\vert S,\sigma^{data}\Big]\\
       &=\mathbb{E}_{\varepsilon^{n}}\Big[\max_{\ell=1,...,H}\big\vert \frac{1}{n}\sum_{i=1}^n \mathbb{E}_{\Tilde{\varepsilon}}\big[\nabla_{k}m_{Y_i,Z_i}(\theta^\ell)-\nabla_{k}m_{\Tilde{Y}_i,Z_i}(\theta^\ell)\big\vert S,\sigma^{data}\big]\big\vert \Big\vert S,\sigma^{data}\Big]\\
       &\quad(\text{By Leibniz rule})\\
       &\leq \mathbb{E}_{\varepsilon^{n},\Tilde{\varepsilon}}\Big[\max_{\ell=1,...,H}\big\vert \frac{1}{n}\sum_{i=1}^n \big[\nabla_{k}m_{Y_i,Z_i}(\theta^\ell)-\nabla_{k}m_{\Tilde{Y}_i,Z_i}(\theta^\ell)
       \big]\big\vert \Big\vert S,\sigma^{data}\Big].
       \end{split}
\end{equation}
Define i.i.d Rademacher variables $\nu\coloneqq(\nu_1,...,\nu_n)$ such that $\Pr(\nu_i=1)=\Pr(\nu_i=-1)=\frac{1}{2}$. Since $m_{Y_i,Z_i}(\theta^{\ell})-m_{\Tilde{Y}_i,Z_i}(\theta^{\ell})\sim \nu_i[m_{Y_i,Z_i}(\theta^{\ell})-m_{\Tilde{Y}_i,Z_i}(\theta^{\ell})]$, we have
\begin{equation}
    \begin{split}
        &\quad\mathbb{E}_{\varepsilon^{n},\Tilde{\varepsilon}}\Big[\max_{\ell=1,...,H}\big\vert \frac{1}{n}\sum_{i=1}^n \big[\nabla_{k}m_{Y_i,Z_i}(\theta^\ell)-\nabla_{k}m_{\Tilde{Y}_i,Z_i}(\theta^\ell) 
       \big]\big\vert\Big\vert  S,\sigma^{data}\Big]\\
       &=\mathbb{E}_{\varepsilon^{n},\Tilde{\varepsilon},\nu}\Big[\max_{\ell=1,...,H}\big\vert \frac{1}{n}\sum_{i=1}^n \nu_i\big[\nabla_{k}m_{Y_i,Z_i}(\theta^\ell)-\nabla_{k}m_{\Tilde{Y}_i,Z_i}(\theta^\ell)
       \big]\big\vert\Big\vert  S,\sigma^{data}\Big]\\
       &\leq \mathbb{E}_{\varepsilon^{n},\Tilde{\varepsilon},\nu}\Big[\max_{\ell=1,...,H}\big\vert \frac{1}{n}\sum_{i=1}^n \nu_i\nabla_{k}m_{Y_i,Z_i}(\theta^\ell)\big\vert
       +\max_{\ell=1,...,H}\big\vert \frac{1}{n}\sum_{i=1}^n \nu_i\nabla_{k}m_{\Tilde{Y}_i,Z_i}(\theta^\ell) 
       \big\vert \Big\vert S,\sigma^{data}\Big]\\
       &=2\mathbb{E}_{\varepsilon^{n},\nu}\Big[\max_{\ell=1,...,H}\big\vert \frac{1}{n}\sum_{i=1}^n \nu_i\nabla_{k}m_{Y_i,Z_i}(\theta^\ell)\big\vert
       \Big\vert S,\sigma^{data}\Big]\\
       &=2\mathbb{E}_{\varepsilon^{n}}\Big[\mathbb{E}_{\nu}\Big[\max_{\ell=1,...,H}\big\vert \frac{1}{n}\sum_{i=1}^n \nu_i\nabla_{k}m_{Y_i,Z_i}(\theta^\ell)\big\vert
       \Big\vert S,\sigma^{data}\Big]\Big\vert
        S,\sigma^{data}\Big].
    \end{split}
\end{equation}
By Lemma \ref{applemma:symetric}, $\frac{1}{n}\sum_{i=1}^n \nu_i\nabla_{k}m_{Y_i,Z_i}(\theta^\ell)$ is a sub-Gaussian process with parameter $\tau/\sqrt{n\upsilon^2}$. Therefore, by the upper bound of sub-Gaussian maxima (Lemma \ref{applemma:subgaussianmax}), we have: 
\begin{equation}\label{appeqq:submax}
    \mathbb{E}\Big[\max_{\ell=1,...,H}\vert \mathbb{B}_k(\theta^{\ell})\vert\Big\vert S,\sigma^{data}\Big]\leq \frac{2\tau}{\sqrt{n}\upsilon}\sqrt{\log(N_c(\delta,\Theta,L_1))}.
\end{equation}
Now, apply Lemma \ref{applemma:covering} to bound the $L_1$-metric entropy $\log(N_c(\delta,\Theta,L_1))$:
\begin{equation}\label{appeq:metricentro}
   \log(N_c(\delta,\Theta,L_1))\leq  d_{\theta}\log\big(1+\frac{2}{\delta}\big).
\end{equation}
Combining Eq.\ref{appeqq:submax} with Eq.\ref{appeq:metricentro}, we have:
\begin{equation}\label{appeqq:supsec}
\mathbb{E}\Big[\max_{\ell=1,...,H}\vert \mathbb{B}_k(\theta^{\ell})\vert\Big\vert S,\sigma^{data}\Big]\leq\frac{2\tau\sqrt{d_{\theta}}}{\sqrt{n}\upsilon}\sqrt{\log\big(1+\frac{2}{\delta}\big)}.
\end{equation}
Combining Eq.\ref{eq:coverH} with Eq.\ref{appeqq:supfirst} and Eq.\ref{appeqq:supsec}, we have:
\begin{equation}
   \mathbb{E}\Big[ \vert \mathbb{B}_k(\Tilde{\theta})\vert\Big\vert S,\sigma^{data}\Big]\leq \underline{\omega}\bar{z}^2\delta+\frac{2\tau\sqrt{d_{\theta}}}{\sqrt{n}\upsilon}\sqrt{\log\big(1+\frac{2}{\delta}\big)}.
\end{equation}
By choosing $\delta=\frac{1}{\underline{\omega}\bar{z}^2\sqrt{n}}$, we conclude:
\begin{equation}
    \mathbb{E}\Big[ \sup_{\theta\in\Theta}\vert \mathbb{B}_k(\theta)\vert\Big\vert S,\sigma^{data}\Big]\leq \frac{1}{\sqrt{n}}\Big(1+\frac{2\tau}{\upsilon}\sqrt{d_{\theta}\log(1+2\vert\underline{\omega}\vert\bar{z}^2\sqrt{n})}\Big).
\end{equation}
To finish the proof, define $C_{B1} = 2\tau/\upsilon$, and $C_{B2}=2\vert\underline{\omega}\vert\bar{z}^2$.
\end{proof}

\subsection{Lemma \ref{applemma:sigma}}
\begin{lemma}\label{applemma:sigma}
Under Assumption \ref{ass:ccp}, for all $i=1,...,n$,
    \begin{equation}
       \mathbb{E}_{\varepsilon^{n}}\big[ \vert \hat{\sigma}^{data}_i- \sigma^{data}_i\vert\big\vert S,\sigma^{data}\big]\leq C_{\sigma}\sqrt{\frac{1+\ln(2)}{n}}.
    \end{equation}
\end{lemma}
\begin{proof}
  This proof follows the same proof strategy as Lemma 5.1 in \citet{KITAGAWA2023109}. Recall that for any nonnegative random variable $Y,$ $\E(Y)=\int_{0}^{\infty} \Pr(Y\geq t) dt.$ Hence, for any $a>0,$
\begin{equation}
\begin{split}
\E(\vert\hat{\sigma}^{data}_i-\sigma^{data}_i\vert^2)&=\int_{0}^{\infty} \Pr(\vert\hat{\sigma}^{data}_i-\sigma^{data}_i\vert^2\geq t)dt\\
&=\int_{0}^{a} \Pr(\vert\hat{\sigma}^{data}_i-\sigma^{data}_i\vert^2\geq t)dt+\int_{a}^{\infty} \Pr(\vert\hat{\sigma}^{data}_i-\sigma^{data}_i\vert^2\geq t)dt\\
&\leq a+\int_{a}^{\infty} \Pr(\vert\hat{\sigma}^{data}_i-\sigma^{data}_i\vert^2\geq t)dt.
\end{split}
\end{equation}
Assumption \ref{ass:ccp} implies that $\Pr(\vert\hat{\sigma}^{data}_i-\sigma^{data}_i\vert\geq \sqrt{t})\leq 2e^{-Nt/C_{\sigma}^2}.$ Hence,
\begin{equation}
    \begin{split}
        \E(\vert\hat{\sigma}^{data}_i-\sigma^{data}_i\vert^2)&\leq a+\int_{a}^{\infty} \Pr(\vert\hat{\sigma}^{data}_i-\sigma^{data}_i\vert^2\geq t)dt\\
        &= a+\int_{a}^{\infty}\Pr(\vert\hat{\sigma}^{data}_i-\sigma^{data}_i\vert\geq \sqrt{t})dt\\
        &\leq a+2\int_{a}^{\infty}e^{-nt/C_{\sigma}^2}dt\\
        &=a+2\frac{C_{\sigma}^2}{n}e^{-Na/C_{\sigma}^2}.
    \end{split}
\end{equation}
Set $a=C_{\sigma}^2\ln(2)/n$ and we have
\begin{equation}
    \E(\vert\hat{\sigma}^{data}_i-\sigma^{data}_i\vert^2)\leq \frac{\ln(2)C_{\sigma}^2}{n}+\frac{C_{\sigma}^2}{n}=C_{\sigma}^2\frac{1+\ln(2)}{n}.
\end{equation} 
Therefore,
\begin{equation}
    \E(\vert\hat{\sigma}^{data}_i-\sigma^{data}_i\vert)\leq \sqrt{(\E(\vert\hat{\sigma}^{data}_i-\sigma^{data}_i\vert^2)}\leq C_{\sigma}\sqrt{\frac{1+\ln(2)}{n}}.
\end{equation}
\end{proof}

\subsection{Lemma \ref{applema:lps}: Lipschitz Property}
\begin{lemma}{(\textbf{Lipschitz Property})}\label{applema:lps}
Define $\widebar{z}\coloneqq\max_{i=1,...,n}\Vert Z_i\Vert_{\infty}$. The following condition on $\Tilde{\mathbb{G}}_n(\cdot)$ is satisfied: 
    \begin{equation}
        \sup_{\substack{\gamma,\gamma'\in\Theta\\ \Vert\gamma-\gamma'\Vert_1\leq \delta}}\vert \mathbb{B}_k(\gamma)-\mathbb{B}_k(\gamma')\vert\leq \vert\underline{\omega}\vert\bar{z}^2\delta.
    \end{equation}
\end{lemma}
\begin{proof}
First, we have:
\begin{equation}\label{eqapp:mmga}
    \begin{split}
     \sup_{\substack{\gamma,\gamma'\in\Theta\\ \Vert\gamma-\gamma'\Vert_1\leq \delta}}\vert \mathbb{B}_k(\gamma)-\mathbb{B}_k(\gamma')\vert\leq\sup_{\substack{\gamma,\gamma'\in\Theta\\ \Vert\gamma-\gamma'\Vert_1\leq \delta}}\vert \nabla_k\mathbb{M}(\gamma)- \nabla_k\mathbb{M}(\gamma')\vert+ \sup_{\substack{\gamma,\gamma'\in\Theta\\ \Vert\gamma-\gamma'\Vert_1\leq \delta}}\vert \nabla_k M(\gamma')-\nabla_k M(\gamma)\vert.
    \end{split}
\end{equation}
The first term in Eq.\ref{eqapp:mmga} is:
\begin{equation}
    \begin{split}
        \nabla_k\mathbb{M}(\gamma)- \nabla_k\mathbb{M}(\gamma')&= \frac{1}{n}\sum_{i=1}^n Y_i\Big[\frac{F_{\varepsilon}'(Z_i^{\intercal}\gamma)}{F_{\varepsilon}(Z_i^{\intercal}\gamma)}-\frac{F_{\varepsilon}'(Z_i^{\intercal}\gamma')}{F_{\varepsilon}(Z_i^{\intercal}\gamma')}\Big]Z_{ik}\\&\quad-\frac{1}{n}\sum_{i=1}^n(1-Y_i)\Big[\frac{F_{\varepsilon}'(Z_i^{\intercal}\gamma)}{1-F_{\varepsilon}(Z_i^{\intercal}\gamma)}-\frac{F_{\varepsilon}'(Z_i^{\intercal}\gamma')}{1-F_{\varepsilon}(Z_i^{\intercal}\gamma')}\Big]Z_{ik}.
    \end{split}
\end{equation}
The second term in Eq.\ref{eqapp:mmga} is:
\begin{equation}
    \begin{split}
       \nabla_k M(\gamma)- \nabla_k M(\gamma')&= \frac{1}{n}\sum_{i=1}^n\mathbb{E}_{\varepsilon^{n}} \Big[Y_i\big[\frac{F_{\varepsilon}'(Z_i^{\intercal}\gamma)}{F_{\varepsilon}(Z_i^{\intercal}\gamma)}-\frac{F_{\varepsilon}'(Z_i^{\intercal}\gamma')}{F_{\varepsilon}(Z_i^{\intercal}\gamma')}\big]Z_{ik}\Big\vert S,\sigma^{data}\Big]\\
       &\quad-\frac{1}{n}\sum_{i=1}^n\mathbb{E}_{\varepsilon^{n}}\Big[(1-Y_i)\big[\frac{F_{\varepsilon}'(Z_i^{\intercal}\gamma)}{1-F_{\varepsilon}(Z_i^{\intercal}\gamma)}-\frac{F_{\varepsilon}'(Z_i^{\intercal}\gamma')}{1-F_{\varepsilon}(Z_i^{\intercal}\gamma')}\big]Z_{ik}\Big\vert S,\sigma^{data}\Big].
    \end{split}
\end{equation}
Applying the Mean Value Theorem to both, we have:
\begin{equation}
    \begin{split}
        \nabla_k\mathbb{M}(\gamma)- \nabla_k\mathbb{M}(\gamma')&= \frac{1}{n}\sum_{i=1}^n\Big[Y_i\omega_0(Z_i^\intercal\acute{\gamma})-(1-Y_i)\omega_1(Z_i^\intercal\acute{\gamma})\Big]Z_{ik}Z_i^{\intercal}(\gamma-\gamma'),
    \end{split}
\end{equation}
\begin{equation}
    \begin{split}
       \nabla_k M(\gamma)- \nabla_k M(\gamma')&= \frac{1}{n}\sum_{i=1}^n\big[\sigma_{i}^n\omega_0(Z_i^\intercal\grave{\gamma})-(1-\sigma_{i})\omega_1(Z_i^\intercal\grave{\gamma}^n)\big]Z_{ik}Z_i^{\intercal}(\gamma-\gamma').
    \end{split}
\end{equation}
where $\sigma_{i}^n=\mathbb{E}_{\varepsilon^{n}}[Y_i\vert S,\sigma^{data}]$, and for some $\acute{\gamma}\in\mathbb{R}^{d_{\theta}}$, $ \grave{\gamma}\in\mathbb{R}^{d_{\theta}}$ on the segment from $\gamma$ to $\gamma'$. 
Then,
\begin{equation}\label{eq:gramgam}
    \begin{split}
         \vert  \nabla_k\mathbb{M}(\gamma)-  \nabla_k\mathbb{M}(\gamma')\vert&\leq\frac{1}{n}\sum_{i=1}^n\vert Y_i\omega_0(Z_i^\intercal\grave{\gamma})-(1-Y_i)\omega_1(Z_i^\intercal\grave{\gamma})\vert\cdot\vert Z_{ik}\vert\cdot 
         \vert Z_i^{\intercal}(\gamma-\gamma')\vert
             \end{split}.
\end{equation}
By Assumption \ref{ass:smooth}, $\omega_0(a)<0$ and $\omega_1(a)>0$ for all $a\in\mathbb{R}$. Recall that
\begin{equation}\label{eq:minomega}
    \underline{\omega}_0\coloneqq \min_{x\in\Xi}\omega_0(x),\quad
    \underline{\omega}_1\coloneqq \min_{x\in\Xi}-\omega_1(x),\quad 
    \underline{\omega}\coloneqq\min\{ \underline{\omega}_0, \underline{\omega}_1\}.
\end{equation}
Combining Eq.\ref{eq:gramgam} with Eq.\ref{eq:minomega}, we have
\begin{equation}
    \begin{split}
         \vert  \nabla_k\mathbb{M}(\gamma)-  \nabla_k\mathbb{M}(\gamma')\vert&\leq \frac{\vert\underline{\omega}\vert}{n}\sum_{i=1}^n\vert Z_i^{\intercal}(\gamma-\gamma')\vert\cdot\vert Z_{ik}\vert\\
         &\leq \frac{\vert\underline{\omega}\vert}{n}\sum_{i=1}^n\Vert Z_i\Vert_{2}^2\vert \vert\gamma-\gamma'\Vert_1\\
         &(\text{By Holder's Inequality})\\
         &\leq \vert\underline{\omega}\vert\bar{z}^2\Vert \gamma-\gamma'\Vert_1,
    \end{split}
\end{equation}
where $\widebar{z}\coloneqq\max_{i=1,...,n}\Vert Z_i\Vert_{\infty}$. By the same argument, 
\begin{equation}
    \vert \nabla_kM(\gamma)- \nabla_kM(\gamma')\vert\leq\vert\underline{\omega}\vert\bar{z}^2\Vert \gamma-\gamma'\Vert_1.
\end{equation}
Combining above two equations with Eq.\ref{eqapp:mmga} gives:
\begin{equation}
    \sup_{\substack{\gamma,\gamma'\in\Theta\\ \Vert\gamma-\gamma'\Vert_1\leq \delta}}\vert \mathbb{B}_k(\gamma)-\mathbb{B}_k(\gamma')\vert\leq \vert\underline{\omega}\vert\bar{z}^2\delta.
\end{equation}
\end{proof}

\subsection{Lemma \ref{applemma:symetric}: Sub-Guassian Process}
\begin{lemma}{(\textbf{Sub-Guassian Process})}\label{applemma:symetric}
Define $\widebar{z}\coloneqq\max_{i=1,...,n}\Vert Z_i\Vert_{\infty}$. $\frac{1}{n}\sum_{i=1}^n \nu_i\nabla_{k}m_{Y_i,Z_i}(\theta^\ell)$ is a sub-Gaussian process with parameter $\tau/\sqrt{n\upsilon^2}$.
\end{lemma}
\begin{proof}
We start from the expectation of the moment-generating function of $1/n\sum_{i=1}^n\nu_i\nabla_{k}m_{Y_i,Z_i}(\cdot)$, which is 
\begin{equation}
    \begin{split}
        &\quad\mathbb{E}_{\nu}\Big[\exp\big[\frac{1}{n}\sum_{i=1}^n s\nu_i\nabla_{k}m_{Y_i,Z_i}(\theta^\ell) \big]\Big\vert S,\sigma^{data}\Big]=\prod_{i=1}^n \mathbb{E}_{\nu}\Big[\exp\big[\frac{s}{n}\nu_i\nabla_{k}m_{Y_i,Z_i}(\theta^\ell)\big]\Big\vert S,\sigma^{data}\Big],
    \end{split}
\end{equation}
where the equality holds as $\{\varepsilon_i\}_{i=1}^n, \{v_i\}_{i=1}^n$ are i.i.d. In addition, the gradient of $m_{Y_i,Z_i}(\theta^\ell)$ is:
\begin{equation}
    \nabla_{\theta}m_{Y_i,Z_i}(\theta^\ell)= \Big[Y_i\frac{F_{\varepsilon}'(Z_i^{\intercal}\theta^\ell)}{F_{\varepsilon}(Z_i^{\intercal}\theta^\ell)}-(1-Y_i)\frac{F_{\varepsilon}'(Z_i^{\intercal}\theta^\ell)}{1-F_{\varepsilon}(Z_i^{\intercal}\theta^\ell)}\Big]Z_{i}.
\end{equation}
Therefore,
\begin{equation}
     \begin{split}
         &\quad\mathbb{E}_{\nu}\Big[\exp\big[\frac{1}{n}\sum_{i=1}^n s\nu_i\nabla_{k}m_{Y_i,Z_i}(\theta^\ell)\big] \Big\vert S,\sigma^{data}\Big]\\
         &=\prod_{i=1}^n \mathbb{E}_{\nu}\Big[\exp\Big[\frac{s\nu_i}{n}\Big(Y_i\frac{F_{\varepsilon}'(Z_i^\intercal\theta^\ell)}{F_{\varepsilon}(Z_i^\intercal\theta^\ell)}-(1-Y_i)
         \frac{F_{\varepsilon}'(Z_i^\intercal\theta^\ell)}{1-F_{\varepsilon}(Z_i^\intercal\theta^\ell)}\Big)Z_{ik}\Big]\Big\vert S,\sigma^{data}\Big].
     \end{split}
\end{equation}
By Hoeffding's Lemma (Lemma \ref{applemma:hoeffding}),
\begin{equation}
    \begin{split}
         &\quad\mathbb{E}_{\nu}\Big[\exp\big[\frac{1}{n}\sum_{i=1}^n s\nu_i\nabla_{k}m_{Y_i,Z_i}(\theta^\ell)\big] \Big\vert S,\sigma^{data}\Big]\\
         &\leq \prod_{i=1}^n \exp\Big[\frac{s^2}{2n^2}\Big(Y_i\frac{F_{\varepsilon}'(Z_i^\intercal\theta^\ell)}{F_{\varepsilon}(Z_i^\intercal\theta^\ell)}-(1-Y_i)
         \frac{F_{\varepsilon}'(Z_i^\intercal\theta^\ell)}{1-F_{\varepsilon}(Z_i^\intercal\theta^\ell)}\Big)^2Z_{ik}^2\Big]\\
         &\leq \prod_{i=1}^n \exp\Big[\frac{s^2}{2n^2}\Big(Y_i\frac{\tau}{\upsilon}+(1-Y_i)
         \frac{\tau}{\upsilon}\Big)^2Z_{ik}^2\Big]\\
         &=  \exp\Big[\frac{s^2\tau^2}{2n^2\upsilon^2}\sum_{i=1}^nZ_{ik}^2\Big]\\
         &\leq \exp\Big[\frac{s^2\tau^2}{2n\upsilon^2}\Bar{z}^2\Big].
    \end{split}
\end{equation}
Recall $\upsilon\coloneqq\min\{\underline{F}_{\varepsilon},1-\widebar{F}_{\varepsilon}\}$, where $\underline{F}_{\varepsilon}\coloneqq \min_{\substack{ \theta\in\Theta\\z\in \mathcal{Z}}}F_{\varepsilon}(z^{\intercal}\theta),$ and 
$\widebar{F}_{\varepsilon}\coloneqq \max_{\substack{ \theta\in\Theta\\z\in \mathcal{Z}}}F_{\varepsilon}(z^{\intercal}\theta).$ Therefore, $\frac{1}{n}\sum_{i=1}^n \nu_i\nabla_{k}m_{Y_i,Z_i}(\theta^\ell)$ is a sub-Gaussian process with parameter $\tau/\sqrt{n\upsilon^2}$.
\end{proof}

\subsection{Lemma \ref{applemma:covering}: Covering Number}
\begin{lemma}{(\textbf{Covering Number})}\label{applemma:covering}
The $\delta$-covering number of a compact parameter space $\Theta\in\mathbb{R}^{d_{\theta}}$ with $L_1$ metric $N_c(\delta,\Theta,L_1)$ is upper bounded by $(1+\frac{1}{\delta})^{d_{\theta}}$.
\end{lemma}
\begin{proof}
    As parameter space $\Theta$ is compact, there exists a constant $C_{\theta}$ such that $\sup_{\theta\in\Theta}\Vert\theta\Vert_1\leq C_{\theta}<\infty$. Let us denote $C_{\theta}$-ball as $B\coloneqq\{\theta\in\mathbb{R}^{d_{\theta}}\mid \Vert \theta\Vert_1\leq C_{\theta}\}$. Then, the covering number of the parameter space $N_c(\delta,\Theta,L_1)$ is bounded by the covering number of the $C_{\theta}$-ball $N_c(\delta,B,L_1)$. Applying Lemma \ref{applemma:volume}, we have:
    \begin{equation}
        N_c(\delta,B,L_1) \leq \frac{\operatorname{vol}\left((1+\frac{2}{\delta})B\right)}{\operatorname{vol}(B)}=\big(1+\frac{2}{\delta}\big)^{d_{\theta}},
    \end{equation}
    where the first inequality holds as the $C_{\theta}$-ball is defined using the same metric as the covering number. Therefore,
    \begin{equation}
      N_c(\delta,\Theta,L_1)\leq  \big(1+\frac{2}{\delta}\big)^{d_{\theta}}. 
    \end{equation}
\end{proof}

\section{Results from Previous Literature}\label{appsec:previous}
\begin{lemma}{(\textbf{Extreme Value Theorem})}\label{app:evtheorem}
If $f$ is continuous on a closed interval $[a,b]$, then $f$
 attains both an absolute maximum value and an absolute minimum value at some numbers in $[a,b]$
.    
\end{lemma}
\begin{lemma}{(\textbf{Berge's Maximum Theorem} \citep{berge1963})}\label{app:bergetheorem}
Let $X \subseteq \mathbb{R}^L$ and $Y \subseteq \mathbb{R}^K$, let $f : X \times Y \rightarrow \mathbb{R}$ be a continuous function and $\Gamma : X \rightarrow Y$ be a compact-valued and continuous correspondence. Then the function $v : X \rightarrow \mathbb{R}$ such that $v(x) = \sup_{y \in \Gamma(x)} f(x, y)$ is continuous.    
\end{lemma}
\begin{lemma}{(\textbf{Volume ratios and Metric Entropy} \citep[\S Lemma~5.7]{wainwright2019high})}\label{applemma:volume}
    Consider a pair of norms $\|\cdot\|$ and $\|\cdot\|'$ on $\mathbb{R}^d$, and let $B$ and $B'$ be their corresponding unit balls (i.e., $B = \{\theta \in \mathbb{R}^d \mid \|\theta\| \leq 1\}$, with $B'$ similarly defined). Then the $\delta$-covering number of $B$ in the $\|\cdot\|'$-norm obeys the bounds
\begin{equation}
    \left(\frac{1}{\delta}\right)^d \frac{\operatorname{vol}(B)}{\operatorname{vol}(B')} \leq N_c(\delta; B, \|\cdot\|') \leq \frac{\operatorname{vol}\left(\frac{2}{\delta}B + B'\right)}{\operatorname{vol}(B')}.
\end{equation}

\end{lemma}

\begin{lemma}{(\textbf{Hoeffding’s Lemma})}\label{applemma:hoeffding}
    Let $X$ be a random variable with $\mathbb{E}X = 0$, $a \leq X \leq b$. Then, for $s > 0$,
\begin{equation}
\mathbb{E}(e^{sX}) \leq e^{s^2(b-a)^2/8}.
\end{equation}
\end{lemma}

\begin{lemma}{(\textbf{Upper bounds for Sub-Gaussian maxima})}\label{applemma:subgaussianmax}
    Let $\lambda > 0$, $n \geq 2$, and let $Y_1, \ldots, Y_n$ be real-valued random variables such that, for all $s > 0$ and $1 \leq i \leq n$, $\mathbb{E}(e^{sY_i}) \leq e^{\lambda^2 s^2/2}$ holds. Then,
\begin{enumerate}
    \item[(i)] $\mathbb{E}(\max_{i \leq n} Y_i) \leq \lambda \sqrt{2 \ln n}$,
    \item[(ii)] $\mathbb{E}(\max_{i \leq n} |Y_i|) \leq \lambda \sqrt{2 \ln (2n)}\leq 2\lambda \sqrt{ \ln (n)}$.
\end{enumerate}
\end{lemma}

\begin{lemma}[Hoeffding's inequality \citep{hoeffding1994probability}]\label{lemmahoeff}
Let $X_1,...,X_n$ be independent bounded random variables such that $X_i$ falls in the interval $[a_i,b_i]$ with probability one. Denote their sum by $S_n=\sum_{i=1}^n X_i$. Then for any $\varepsilon>0$ we have
\begin{equation}
    \Pr\{S_n-\E S_n\geq \varepsilon\}\leq e^{-2e^2/\sum_{i=1}^n (b_i-a_i)^2},
\end{equation}
and
\begin{equation}
    \Pr\{S_n-\E S_n\leq -\varepsilon\}\leq e^{-2e^2/\sum_{i=1}^n (b_i-a_i)^2}.
\end{equation}
\end{lemma}
\end{appendices}

\end{document}